\documentclass[conference]{IEEEtran}
\input{symbols}

\begin{document}

\title{\FLname: Efficient and Robust \\ Secure \emph{Quantized} Aggregation\textsuperscript{$\star$}\thanks {\textsuperscript{$\star$}Please cite the conference version of this paper published at IEEE~SaTML'24~\cite{SATML:BMP0ST24}.}}

\author{
    \IEEEauthorblockN{
        Yaniv~Ben-Itzhak\orcidlink{0000-0002-3844-5940}\IEEEauthorrefmark{1},
        Helen~Möllering\orcidlink{0000-0001-9371-3592}\IEEEauthorrefmark{2}, 
        Benny~Pinkas\orcidlink{0000-0002-9053-3024}\IEEEauthorrefmark{3},
        Thomas~Schneider\orcidlink{0000-0001-8090-1316}\IEEEauthorrefmark{2},
        Ajith~Suresh\orcidlink{0000-0002-5164-7758}\IEEEauthorrefmark{4},\\
        Oleksandr~Tkachenko\orcidlink{0000-0001-9232-6902}\IEEEauthorrefmark{5},
        Shay~Vargaftik\orcidlink{0000-0002-0982-7894}\IEEEauthorrefmark{1}, 
        Christian~Weinert\orcidlink{0000-0003-4906-6871}\IEEEauthorrefmark{6},
        Hossein~Yalame\orcidlink{0000-0001-6438-534X}\IEEEauthorrefmark{2},
        Avishay~Yanai\orcidlink{0000-0003-4060-0150}\IEEEauthorrefmark{1}\smallskip
    }
    \IEEEauthorblockA{
        \IEEEauthorrefmark{1}
        VMware Research Group
        \IEEEauthorrefmark{2}
        Technical University of Darmstadt
        \IEEEauthorrefmark{3}
        Aptos Labs and Bar-Ilan University\\
        \IEEEauthorrefmark{4}
        Technology Innovation Institute
        \IEEEauthorrefmark{5}
        DFINITY Foundation
        \IEEEauthorrefmark{6}
        Royal Holloway, University of London
    }
}

\maketitle

\begin{abstract}
Secure aggregation is commonly used in federated learning~(FL) to alleviate privacy concerns related to the central aggregator seeing all parameter updates in the clear.
Unfortunately, most existing secure aggregation schemes ignore two critical orthogonal research directions that aim to~(i) significantly reduce client-server communication and~(ii)~mitigate the impact of malicious clients.
However, both of these additional properties are essential to facilitate cross-device~FL with thousands or even millions of~(mobile) participants.

In this paper, we unite both research directions by introducing~\FLname, the first secure aggregation framework for~FL that operates~\emph{efficiently} on quantized inputs and simultaneously provides robustness against malicious clients.
Our framework leverages~(novel) multi-party computation~(MPC) techniques and supports multiple linear~(1-bit) quantization schemes, including ones that utilize the randomized~Hadamard transform and~Kashin's representation.

Our theoretical results are supported by extensive evaluations.
We show that with~\emph{no overhead for clients} and moderate overhead for the server compared to transferring and processing quantized updates in plaintext, we obtain comparable accuracy for standard~FL benchmarks.
Moreover, we demonstrate the robustness of our framework against state-of-the-art poisoning attacks.
\end{abstract}

\section{Introduction}
\label{sec:introduction}
Federated learning~(FL)~\cite{PMLR:McMahanMRHA17} is a paradigm for large-scale distributed machine learning, where in each training round a subset of clients locally updates a global model that is then centrally aggregated. FL quickly gained popularity due to its promises of data privacy, resource efficiency, and ability to handle dynamic participants.

However, in terms of~\emph{privacy}, the central aggregator learns the individual client updates in the clear and thus can infer sensitive details about the clients' private input data~\cite{SP:MelisSCS19,ndss:PyrgelisTC18,ccs:GanjuWYGB18,sp:ShokriSSS17,ARXIV:WXWZ19}.
Hence, many secure aggregation schemes, e.g., \cite{ccs:BonawitzIKMMPRS17,SPW:FereidooniMMMMN21,ARXIV:MMRPVF22}, have been developed, where the aggregator only learns the aggregation result, i.e., the global model~(we refer to~\cite{ftml:KairouzMABBBBCC21,access:OuadrhiriA22} for a discussion of differential privacy in~FL as an orthogonal privacy-enhancing paradigm).
Most prominently, in the~\enquote{SecAgg} protocol~\cite{ccs:BonawitzIKMMPRS17}, clients exchange masks with peers to blind their model updates such that the masks cancel out during aggregation and reveal only the exact result.
However, this approach requires an interactive setup between clients and thus is less reliable when dealing with real-world problems such as client dropouts~(except for special variants~\cite{ARXIV:YSHLYA21,ARXIV:FastsecAgg,conext:LiGBL21}).
Moreover, recent research has shown that a~\emph{single malicious aggregator} can reconstruct individual training data from clients' inputs, despite the use of secure aggregation~\cite{WEB:BDSSSP22,CCS:PFA22,ARXIV:BDSSSP23}.

In terms of~\emph{resource efficiency} of plain~FL, clients have to send parameter updates~(also known as {\em gradients}), which typically consist of thousands or millions of coordinates with one floating point number per coordinate. For cross-device~FL involving mobile clients with limited upload bandwidth, this quickly becomes infeasible when dealing with increasingly large model architectures, where gradients consist of millions of coordinates.
Therefore, quantization schemes that exploit the noise resiliency of gradient-based~FL methods~(e.g., federated averaging~\cite{PMLR:McMahanMRHA17}) have been proposed to significantly compress client updates~(typically replacing the representation of each coordinate by a single bit, instead of a~32-bit floating point number), e.g., \cite{ICML:SureshYKM17, ARXIV:CKMT18, nips:VargaftikBPMBM21, nips:AlistarhH0KKR18,icalp:BasatMV21,ICML:BernsteinWAA18,interspeech:SeideFDLY14}.

Unfortunately, so far the~ML community has worked on optimizing FL for resource efficiency, while the security community has separately worked on optimizing secure aggregation for privacy. One of the very few exceptions is a work~\cite{icml:ChenCKS22} that combines~SecAgg with sketching~\cite{icml:RothchildPUISB020} to compress gradients. 
Besides the mentioned disadvantages of~SecAgg, their work considers only one compression method. However, because of the trade-offs between accuracy and computational efficiency offered by various forms of quantization, which become particularly relevant for gradients with millions of coordinates, it is important to adopt a modular approach.

Beyond data privacy threats, FL was also shown to be vulnerable to manipulations by malicious clients who alter their local models/updates, affecting the characteristics of the final global model~\cite{SP:ShejwalkarHKR22,usenix:FangCJG20,ICML:XiaoBBFER15,AISTATS:BagdasaryanVHES20,ICML:BhagojiCMC19,NDSS:SheHou21}.
This highlights the need for effective~\emph{defense measures} capable of thwarting such attacks.
While there exists a plethora of such attacks and defenses in the~FL literature~(see~\appref{subsect:flattacks}), we focus on defending against~\emph{untargeted poisoning attacks}, in which the attacker attempts to damage the trained model's performance for a large number of test inputs, typically resulting in a final global model with a high error rate~\cite{usenix:FangCJG20,NDSS:SheHou21,nips:BaruchBG19}.
Unfortunately, existing defenses for such attacks cannot be efficiently translated to secure aggregation in the quantized setting~\cite{USENIX:NRCYMFMMMKSSZ21,KDD:ZhangCJG22}.

\subsection{Our Contributions}
We propose~\FLname{}, a framework that enables~\emph{efficient} and~\emph{robust} secure aggregation for~FL with a distributed aggregator that operates~\emph{directly} on~\emph{quantized} parameter updates.
Specifically, \FLname{} has virtually~\emph{no additional communication overhead for clients} compared to the insecure transfer of quantized plaintext updates. Achieving this is non-trivial for prior single-server solutions that are based on masking techniques or homomorphic encryption~(HE)~\cite{ccs:BonawitzIKMMPRS17,USENIX:ZhangLX00020}.

In~\FLname{}, we use outsourced multi-party computation~(MPC), where the clients apply computationally efficient secret-sharing techniques to distribute their sensitive~(quantized) parameter updates among multiple servers that together form a distributed aggregator. These servers then use MPC protocols~\cite{crypto:CramerDESX18,SP:Damgard0FKSV19,ACNS:BenNieOmr19} to securely compute the aggregation and only reveal the updated global model.

The distributed aggregator model is well established~\cite{SPW:FereidooniMMMMN21,popets:MansouriOBC2022,SP:RatheeSWP23,SPW:GMSWSY23}~(although prior works cannot efficiently handle quantized inputs), offering practical benefits over single-server solutions, such as not requiring an interactive setup between clients~\cite{USENIX:NRCYMFMMMKSSZ21,ARXIV:MMRPVF22} and efficient dropout handling. Moreover, in recent years, numerous studies have demonstrated that single server aggregation methods are susceptible to privacy attacks when an aggregator is corrupt~\cite{aaai:So2023,WEB:BDSSSP22,CCS:PFA22,ARXIV:BDSSSP23,ICLR:FowlGCGG22,ICML:WenGFGG22}. The vulnerability of these methods stems from the fact that the aggregator holds complete authority in selecting clients and the data transmitted to and received from them. On the other hand, with our~MPC protocols, data privacy can be guaranteed even if all servers except one are compromised or their operators receive subpoena requests.

As~MPC protocols typically induce significant overhead in terms of~(inter-server) communication, we propose optimizations for secure aggregation that~\emph{support any~\enquote{linear} quantization scheme}, including~\emph{1-bit quantization} schemes that uses preprocessing like random rotations~\cite{ICML:SureshYKM17} and~Kashin's representation~\cite{ARXIV:CKMT18}.
We also study novel~\emph{approximate}~MPC variants that leverage FL's noise tolerance and might be of independent interest. We formally prove that our resultant secure aggregation scheme is an~\emph{unbiased estimator} of the arithmetic mean and explore efficiency-accuracy trade-offs.

We implement a combined end-to-end~FL evaluation and~MPC simulation environment. Our prototype implementation allows to assess the performance and accuracy of our solution for stochastic quantization schemes, including recent state-of-the-art distributed mean estimation techniques~\cite{ICML:SureshYKM17,ARXIV:CKMT18}.
Our results demonstrate that standard~FL benchmarks' accuracy is barely impacted while our optimizations and approximation can significantly reduce inter-server communication.
For example, when training the~LeNet architecture for image classification on the~MNIST data set~\cite{pieee:LeCunBBH98} using~$1$-bit stochastic quantization with~Kashin's representation~\cite{ARXIV:CKMT18}, training accuracy is only slightly reduced from~99.04\% to~98.71\% after~1000 rounds, while inter-server communication drops from~16.14 GB to~0.94 GB compared to naive~MPC-based secure aggregation when considering~500 clients per round, an improvement by factor~$17.2\times$.

Since clients may act maliciously and try to degrade accuracy with their updates, we design a novel bipartite defense mechanism called~\FLdefensename{} to ensure the robustness of our framework.
Specifically, we provide protection against state-of-the-art untargeted poisoning attacks~\cite{NDSS:SheHou21}, combining magnitude clipping and directional filtering based on the gradients' approximate~L2-norms and cosine distances. Notably, \FLdefensename{} is the first defense mechanism to work directly on quantized inputs and thus enables a highly efficient realization in~MPC, whereas existing works require expensive~MPC conversions of all individual parameters~\cite{ICDCS:AndreinaMMK21,USENIX:NRCYMFMMMKSSZ21}.
We summarize our contributions as follows:

\begin{enumerate}[wide, labelwidth=!, labelindent=0pt, parsep=4pt]
    \item 
    First~\emph{secure aggregation} framework called~\FLname{} to consider~(1-bit) quantization with almost~\emph{no communication overhead for clients} compared to~\emph{plaintext quantized~FL}.
    \item 
    Novel~\emph{optimizations and approximations} to reduce MPC-induced inter-server communication, with performance/accuracy trade-offs.
    \item
    End-to-end~FL evaluation and~MPC simulation environment, demonstrating the efficiency and accuracy of~\FLname{}.
    \item 
    First efficient and effective~FL~(poisoning) defense operating directly on quantized updates.
\end{enumerate}

Though we study~FL as our primary application, our secure aggregation protocols have numerous other applications like privacy-preserving aggregate statistics, for which there are currently~(less efficient) real-world deployments that also rely on distributed aggregators~(e.g., telemetry reporting in~Mozilla's Firefox browser~\cite{LetsEncrypt,SCN:AddankiGJOP22}).
For these settings, we achieve improvements in communication of up to~$4\times$ over prior works like Prio+~\cite{SCN:AddankiGJOP22} and the details are provided in~\secref{sec:communication_costs_detailed}.

\subsection{Related Work}
\label{sec:related_work_concise}
We present a summary of the most relevant related works here, with a more comprehensive discussion in~\appref{sec:related-work}.

\myparatight{Quantization and Compression in FL}
In this work, we focus on quantization to reduce communication in~FL. However, an alternative line of work investigates compression techniques for the same purpose. There are three main directions for gradient compression in~cross-device FL:~(i)~gradient sparsification~(e.g., \cite{sigcomm:Fei0SCS21,NIPS:SCJ18,emnlp:AjiH17,ARXIV:KonecnyMYRSB16}), (ii)~client-side memory and error-feedback~(e.g.,~\cite{interspeech:SeideFDLY14,nips:AlistarhH0KKR18,nips:RichtarikSF21,ARXIV:BHRS20}), and~(iii)~entropy encodings~(e.g., \cite{ICML:SureshYKM17,icml:VargaftikBPMBM22,nips:AlistarhH0KKR18}). Reviews of current state-of-the-art gradient compression techniques and some open challenges can be found in~\cite{ftml:KairouzMABBBBCC21}.

Compression techniques are less suited for the requirements of secure aggregation for cross-device FL than quantization, e.g., due to computational overhead, incompatibility with secure aggregation, and state-requirements on the client side. We discuss more details in~\sect{sec:stochastic_quantization_intro} and~\appref{subsec:rw-compression}.

\myparatight{Secure Aggregation \& Model Inference Attacks}
In conventional~FL with a single aggregator, clients share locally trained model updates with a central party to train a global model.
However, sharing those updates makes the system vulnerable to data leakage. Attacks exploiting this leakage are called \emph{inference attacks}~\cite{ARXIV:BDSSSP21,icml:LamW0RM21,SP:NasrSH19}. Even a semi-honest central server can learn confidential information about the used private training data by analyzing the received local model updates. 

A common countermeasure against inference attacks is to use secure aggregation~\cite{SPM:ErkinTLP13,popets:KursaweDK11} (cf.~\secref{app:related_work_secureagg}).
As~FL poses specific challenges such as a large number of clients and drop-out tolerance, tailored secure aggregation protocols for~FL have been proposed~\cite{ccs:BonawitzIKMMPRS17,SPW:FereidooniMMMMN21,ccs:BellBGL020,NDSS:savPTFBSH21,ARXIV:FastsecAgg,jsait:SoGA21a}.
Those hide individual updates, ensuring that the server has only access to global updates, hence, effectively prohibiting the analysis of individual updates for inference attacks.
The first scheme, SecAgg~\cite{ccs:BonawitzIKMMPRS17}, combines secret sharing with authenticated symmetric encryption, but requires~4 communication rounds per training iteration among servers and client.
Bell~et~al.~\cite{ccs:BellBGL020} improve upon~SecAgg~\cite{ccs:BonawitzIKMMPRS17} by reducing client communication and computation to poly-logarithmic complexity.
However, from a practical point of view also~\cite{ccs:BellBGL020} as well as other existing secure aggregation protocols designed for FL still exhibit significant computation and communication overhead: Due to underlying secret sharing or encryption, those protocols typically encode each local update in $32$-bit and add computational overhead caused by the required cryptographic operations. In contrast, \FLname{} enables highly communication-efficient secure aggregation thanks to~$1$-bit quantization and causes almost~\emph{no} additional communication overhead on the client side compared to plaintext~FL. Fereidonni et al.~\cite{SPW:FereidooniMMMMN21} provide more details by comparing several secure aggregation protocols with respect to efficiency and practicality. Mansouri et al.~\cite{popets:MansouriOBC2022} provide a comprehensive analysis of secure aggregation schemes~w.r.t.\ their suitability for~FL.

To the best of our knowledge, only Chen et al.~\cite{icml:ChenCKS22} and Beguier et al.~\cite{ARXIV:BAT20} have considered both compression and secure aggregation in combination for~FL.
Specifically, Chen et al.~\cite{icml:ChenCKS22} combine~SecAgg~\cite{ccs:BonawitzIKMMPRS17} with sparse random projection and count-sketching~\cite{icml:RothchildPUISB020} for compression.
Moreover, they add noise using a distributed discrete~Gaussian mechanism to generate a differential private output.
Beguier et al.~\cite{ARXIV:BAT20} combine arithmetic secret-sharing with~Top-$k$ sparsification~\cite{NIPS:SCJ18} and~1-bit quantization~\cite{ICML:BernsteinWAA18}.
As we point out in~\secref{subsec:rw-compression}, both sketching and sparsification are sub-optimal for our envisioned cross-device setting given that they require memory and error-feedback on the client side.
In contrast, we focus on a dynamic scenario where clients might drop-out at any time and may contribute only once to the training.

\begin{table}[t]
    \begin{center}
    \resizebox{0.48\textwidth}{!}{%
    \begin{threeparttable}
        \begin{tabular}{lccccccccc} 
        \toprule
        Categories & Reference & Technique & M.P.  & Quant. & P.R. &\makecell{Dist.\\Servers} & \makecell{No Client \\Interaction}
        \\ \midrule
        \multirow{2}{*}{Aggregation} 
        &\cite{nsdi:Corrigan-GibbsB17}&MPC&\cmark&\xmark&\xmark&\cmark&\cmark\\
        &\cite{ccs:BellBGL020}& Masking  &\cmark&\xmark&\xmark&\xmark&\xmark \\
        \midrule
        \multirow{3}{*}{Quantization} 
        &\cite{ICML:SureshYKM17}&--&\xmark&\cmark&\xmark&\xmark&\cmark\\
        &\cite{ARXIV:BAT20}&MPC&\cmark&\xmark&\xmark&\cmark&\cmark\\
        &\cite{icml:ChenCKS22}&Masking+DP &\cmark&\cmark&\xmark&\xmark&\xmark\\
        \midrule
        \multirow{2}{*}{Robustness} 
        &~\cite{USENIX:NRCYMFMMMKSSZ21}&MPC&\cmark&\xmark&\cmark&\cmark&\cmark\\
        &\cite{KDD:ZhangCJG22}&--&\xmark&\xmark&\cmark&\xmark&\cmark   \\
        \midrule
        \FLname
        &\textbf{This}&MPC&\cmark&\cmark&\cmark&\cmark&\cmark\\
        \bottomrule
        \end{tabular}
    \end{threeparttable}
    }
    \end{center}
    \vspace{-2mm}
    \captionsetup{font=small}
    \caption{High-level comparison of \FLname and previous works. Notation: MPC---Secure Multi-Party Computation, DP---Differential Privacy, M.P.---Model Privacy, Quant.---Quantization~(refers to the schemes where compressed gradients are communicated by the clients to the aggregator(s)), P.R.---Poisoning Resilience, Dist.---Distributed. Client Interaction refers to interaction among clients. Since the body of literature is vast, comparison is made against only a subset representing each category.}
    \label{tab:relatedworksconcise}
    \vspace{-5mm}
\end{table}

\myparatight{Poisoning Attacks \& Defenses}
FL was shown vulnerable to manipulations by malicious clients~\cite{AISTATS:BagdasaryanVHES20,nips:BaruchBG19,usenix:FangCJG20,NDSS:SheHou21,ICML:ZhangPSYMMR022}. Targeted or backdoor attacks aim at manipulating the inference results for specific attacker-chosen inputs~\cite{AISTATS:BagdasaryanVHES20,ICML:ZhangPSYMMR022}, while untargeted poisoning attacks~\cite{nips:BaruchBG19,usenix:FangCJG20,NDSS:SheHou21} reduce the overall global model accuracy. 
As untargeted attacks are considered to be more severe~(given they are harder to detect, cf.~\secref{sec:defense}), we focus on defending those using Byzantine-robust defenses like~\cite{KDD:ZhangCJG22,USENIX:NRCYMFMMMKSSZ21}.
In~\appref{subsect:flattacks}, we provide a more detailed overview of poisoning attacks and defenses.

\myparatight{Comparison}
We compare~\FLname qualitatively in~Tab.~\ref{tab:relatedworksconcise} to the state-of-the-art related work with respect to aggregation and quantization, as well as robustness against poisoning.
\section{Problem Statement}
\label{sec:ps}
We now define the precise problem of secure quantized aggregation for~FL, which we address in our work.
We first introduce the necessary preliminaries on~FL and quantization schemes, formalize the functionality we want to compute securely, and finally define our threat and system model for common~(cross-device)~FL scenarios.

\subsection{Aggregation for Federated Learning}
\label{sec:fl-aggregation}
Google introduced federated learning~(FL) as a distributed machine learning~(ML) paradigm in~2016~\cite{ARXIV:KonecnyMYRSB16,PMLR:McMahanMRHA17}. In~FL, $N$~data owners collaboratively train a~ML model~$G$ with the help of a central aggregator~$\server{}$ while keeping their input data~\emph{locally private}. In each training iteration~$t$, the following three steps are executed:

\begin{enumerate}[wide, labelwidth=!, labelindent=0pt, parsep=0pt]
    \item The server~$\server{}$ randomly selects~$\numclients$ out of~$N$ available clients and provides the most recent global model~$G^t$.
    \item Each selected client~$\client{i}$, $i \in [\numclients]$, sets its local model~$w_{i}^{t+1}=G^t$ and improves it using its local dataset~$D_i$ for~$E$ epochs (i.e., local optimization steps):
    
    \begin{align}
        \label{eq:fedavg-local-update}
        w_{i}^{t+1} \leftarrow w_{i}^{t+1}-\eta_{\client{i}} \frac{\partial L (w_{i}^{t+1},B_{i,e})}{\partial w_{i}^{t+1}},
    \end{align}
    where~$L$ is a loss function, $\eta_{\client{i}}$ is the clients' learning rate, and~$B_{i,e}\subseteq D_i$ is a batch drawn from~$D_i$ in epoch~$e$, where~$e \in [E]$.
    After finishing the local training, $\client{i}$ sends its local update~$w_{i}^{t+1}$ to~$\server{}$.
    \item The server updates to a new global model~$G^{t+1}$ by combining all~$w_{i}^{t+1}$ with an aggregation rule~$f_{\textit{agg}}$:
    \begin{align}
        G^{t+1} \leftarrow G^{t}-\eta_\mathsf{S} \cdot f_{\textit{agg}}(w_{1}^{t+1},\ldots,w_{\numclients}^{t+1}),
    \end{align}
    where $\eta_\mathsf{S}$ is the server's learning rate.
    The most commonly used aggregation rule, which we also focus on, is~$\mathsf{FedAvg}$~\cite{PMLR:McMahanMRHA17}.
    It averages the local updates as follows:
    \begin{align}
        \label{eq:fedavg}
        \mathsf{FedAvg} (w_{1}^{t+1},\ldots,w_{\numclients}^{t+1}) = \sum_{i=1}^{\numclients} \frac{|D_i|}{\numclients} w_{i}^{t+1}
    \end{align}
\end{enumerate}

This process is repeated until some stopping criterion is met~(e.g., a fixed number of training iterations or a certain accuracy is reached).

\subsection{Stochastic Quantization}
\label{sec:stochastic_quantization_intro}
Quantization is a central building block in~FL, where data transmission is often a bottleneck.
Thus, compressing the~(thousands or millions of) gradients is essential to adhere to client bandwidth constraints, reducing training time, and allowing better inclusion and scalability.
We now review the desired properties and constructions of quantization schemes that will play a key role in our system design and additional details are provided in \appref{app:stochastic_quantization_intro}.

\myparatight{Unbiasedness}
A well-known and desired design property of gradient compression techniques is being unbiased. That is, for an estimate~$\hat{w}$ of a gradient~$w \in \mathbb{R}^d$, being unbiased means that~$\mathbb{E}[\hat{w}]=w$.
Unbiasedness is desired because it guarantees that the mean squared error~(MSE) of the mean's estimation decays linearly with respect to the number of clients, which can be substantial in~FL. 
In~FL and other optimization techniques based on stochastic gradient descent~(SGD) and its variants~(e.g., \cite{PMLR:McMahanMRHA17,mlsys:LiSZSTS20,icml:KarimireddyKMRS20}), the~MSE measure~(or normalized~MSE, a.k.a.~NMSE, cf.~\appref{sec:approx_evaluation}) is indeed the quantity of interest since it affects the convergence  rate and often the final accuracy of models.
In fact, provable convergence rates for the non-convex compressed~SGD-based algorithms have a linear dependence on the~NMSE.
Accordingly, to keep the estimates unbiased, modern quantization techniques employ stochastic quantization~(SQ) and its variants to compress the gradients. 

\myparatight{1-bit SQ}
Our focus is on the appealing communication budget of a single bit for each gradient entry, resulting in a~32$\times$ compression ratio compared to regular~32-bit floating point entries.
Using~1-bit quantization has been the focus of many recent works concerning distributed and~FL network efficiency~(e.g., \cite{nips:VargaftikBPMBM21,ICML:BernsteinWAA18,icml:TangGARLLLZH21,interspeech:SeideFDLY14,ARXIV:jin2020stochastic}).
These works repeatedly demonstrated that a budget of~1-bit per coordinate is sufficient to achieve model accuracy that is competitive to that of a non-compressed baseline.

In particular, 1-bit~SQ (i.e., SQ using two quantization values) can be done as follows:
For a vector~$X$ with~$m$ dimensions, the client sends each coordinate as~$\sigma_i = \textit{Bernoulli}(\frac{X_i - \stepmin{X}}{\stepmax{X}-\stepmin{X}})$, where~$\stepmax{X} = \max(X)$ and~$\stepmin{X} = \min(X)$.
The coordinate is then reconstructed by the receiver as~$\vecqV{X}{i} = \stepmin{X} + \qvecV{X}{i} \cdot (\stepmax{X} - \stepmin{X})$. 
This simple technique results in an unbiased quantization as desired, i.e., $\E [\vecqV{X}{i}]=\E [\stepmin{X} + \qvecV{X}{i} \cdot (\stepmax{X} - \stepmin{X})] = X_i$.

\myparatight{Linear SQ Techniques}
\label{sec:linear-sq}
A key requirement of being able to perform secure aggregation~\emph{efficiently} is being able to aggregate client gradients in their compressed~(i.e., quantized) form.
The schemes with this property are called~\enquote{linear} henceforth. 
One approach involves a~\enquote{global scales} method, where each client securely learns the maximal and the minimal value across all gradients, i.e.,~$\stepmax{} = \max_c \{\stepmax{X_c}\}$ and $\stepmin{} = \min_c \{\stepmin{X_c}\}$.\footnote{This approach resembles \enquote{scaler sharing} in TernGrad~\cite{NIPS:WenXYWWCL17}.}
Despite its simplicity, this method has several drawbacks:~(i)~it requires a preliminary communication stage, (ii)~it reveals the global extreme values~(even if they are computed securely across all clients), and~(iii)~the reconstruction error~(i.e., the NMSE) is increased as it now depends on the extreme values across all round participants.
Accordingly, we also consider a second approach where each client continues to use its own~\enquote{local scales}.
Since plainly using individual scales is not~\enquote{linear}, i.e., it does not allow for aggregating the quantized gradients efficiently without decoding them, to realize this approach, we use a known approximation~\cite{ARXIV:BAT20} and adapt it to our setting~(cf.~\secref{sec:SepAgg}).

While we focus on the mentioned vanilla~SQ, SQ with random rotation~\cite{ICML:SureshYKM17} and~SQ with~Kashin's representation~\cite{ARXIV:CKMT18,TIT:LyubarskiiV10,ARXIV:SSR20}, our framework seamlessly supports any~\enquote{linear} quantization scheme, namely, any quantization technique that allows for aggregation in a compressed form~(cf.~\secref{sec:extension-multi-bit}).

\myparatight{Other Approaches}
We acknowledge recent advances in gradient quantization~\cite{nips:VargaftikBPMBM21, icml:VargaftikBPMBM22, iclr:DaviesGMAA21,ARXIV:Ran22}), but these non-linear techniques are less applicable to our framework.

\subsection{MPC for Secure Aggregation}
\label{sec:prelims_mpc}
Secure computation in the form of multi-party computation~(MPC) allows a set of parties to privately evaluate any efficiently computable function on confidential inputs.
This paradigm can be utilized to securely run the~FedAvg aggregation algorithm~\cite{SPW:FereidooniMMMMN21,USENIX:NRCYMFMMMKSSZ21,ARXIV:MMRPVF22,esorics:DongCLWZ21}:
The set of selected~FL clients uses additive secret sharing to distribute their sensitive inputs among a set of~MPC servers, which resemble a distributed aggregator.
The~MPC servers securely add the received shares and reconstruct the public result from the resulting shares.
In the next iteration, the public model is distributed to a new set of clients chosen at random, and the process is repeated until the desired accuracy is attained. A visualization of this outsourced~MPC setting for secure aggregation is provided in~Fig.~\ref{fig:sec_agg}.

\begin{figure}[htb!]
    \centering
    \resizebox{\linewidth}{!}{
        \input{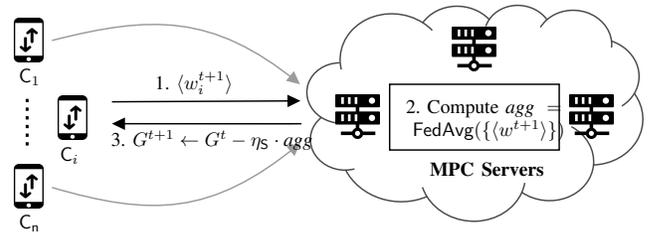}
    }
    \captionsetup{font=footnotesize}
    \caption{Secure aggregation in~FL for~$\numclients$ clients using outsourced~MPC; $\langle w_i^{t+1} \rangle$ denotes the secret-shares of gradient~$w_i^{t+1}$ that client~$C_i$ has in round~$t + 1$; $G^t$ is the model of the previous round and~$\eta_\mathsf{S}$ the server-side learning rate.}
    \label{fig:sec_agg}
\end{figure}

In the remainder of this paper, we work towards a secure aggregation for~FL using~FedAvg on~\emph{quantized} inputs.
In the following, we give our assumptions in terms of the threat model and refine the description of our system model.

\myparatight{MPC Servers}
The~MPC servers are assumed to be semi-honest.
This means, they follow the protocol specification, but may try to learn additional information from the protocol transcripts.
This well-established assumption, both in theory~\cite{CCS:AFLNO16,USENIX:JuvVaiCha18,USENIX:MLSZP20,USENIX:NRCYMFMMMKSSZ21,USENIX:PSSY21,ARES:MP2ML,popets:HegdeMSY21} and in practice~\cite{LetsEncrypt,WEB:OBLIVIOUSDNS20,WEB:ENPA}, is motivated by the fact that companies who run~FL with a secure aggregation scheme try to actively protect their clients' data but want to make sure that someone who monitors or breaks into their systems cannot get plaintext access to the data that is being processed.
Additionally, this assumption is justified as organizations usually cannot afford the reputational and monetary risk when being caught betraying their users' trust. The protocols that we design provide this security guarantee in a dishonest majority setting, where data is protected even when an adversary~$\Adv$ corrupts all except one of the involved servers.
Furthermore, our framework is easily extendable to provide~\emph{malicious privacy}, which ensures input privacy even when the corrupted servers try to actively cheat~\cite{SP:RatheeSWP23}.

\myparatight{Malicious Clients}
We anticipate that some clients might behave maliciously to negatively affect the quality of the global model with manipulated updates~(i.e., poisoning attacks, cf.~\secref{sec:defense}).
This is because there are significantly less incentives for clients to behave honestly.
Also, due to the sheer number of clients in cross-device~FL, it is hard to verify their reputation.
We assume an honest majority of clients as is standard in~FL, yet our framework is secure even when malicious clients collude with corrupted servers.

\myparatight{Preprocessing Model}
\label{sec:mpc-preprocessing-model}
We use~MPC in the preprocessing model~\cite{C:DPSZ12,ESORICS:DKLPSS13,acns:BaumDTZ16,C:DamOrlSim18}.
This means, we try to shift as much computation and communication as possible to a~\emph{data-independent preprocessing} or~\emph{offline} phase that can be executed at an arbitrary time before the actual computation.
This gives several advantages, e.g., service providers can exploit cheap spot instances.
It also guarantees faster results when the~\emph{data-dependent online phase} of the protocol is executed.

\myparatight{Shared Randomness}
\label{sec:mpc-shared-randomness}
We assume that clients and~MPC servers have access to a shared randomness source, e.g., by agreeing on a~PRNG seed.
Such a configuration has been widely employed in~MPC protocols~\cite{CCS:AFLNO16,EC:FLNW17,ccsw:ChaudhariCPS19,CCS:MohRin18,EPRINT:KPPS22} and in~ML systems~\cite{ICML:SureshYKM17,icalp:BasatMV21} to optimize communication.

\subsection{Secure Quantized Aggregation}
\label{sec:secure-quantized-aggregation}
To introduce the problem of secure quantized aggregation, without loss of generality, we consider a simple stochastic binary quantization scheme to begin with.
In this scheme, a~$\numcoordinates$-dimensional vector of the form~$\vec{X} = \{x_1,\ldots,x_{\numcoordinates}\}$ comprising of~$\ell$-bit values will be quantized to obtain a triple~$\vecq{X} = (\qvec{X}, \stepmin{X}, \stepmax{X})$.
Here, $\qvec{X}$ is an~$\numcoordinates$-dimensional binary vector with a zero at an index indicating the value~$\stepmin{X}$ and a one indicating the value~$\stepmax{X}$.
Also, $\stepmin{X}$ and~$\stepmax{X}$ correspond to the minimum and maximum values in the vector~$\vec{X}$.
With this binary quantization~(cf.~\secref{sec:stochastic_quantization_intro}), the quantized value at the~$i$th dimension, denoted by~$\vecqV{X}{i}$, can be written as 
\begin{equation}
  \label{eq:quantized value}
  \vecqV{X}{i} = \stepmin{X} + \qvecV{X}{i} \cdot (\stepmax{X} - \stepmin{X}).
\end{equation}
Before going into the details of aggregation, we provide some of the basic notation that will be utilized throughout the paper.

\myparatight{Notation}
$\matD{Y}{\alpha}{\beta}$ denotes a matrix of dimension~$\alpha\times\beta$ with~$\matrow{Y}{i}$ being the~$i$th row and~$\matcol{Y}{j}$ being the~$j$th column.
An element in the~$i$th row and~$j$th column is denoted by~$\matpos{Y}{i}{j}$.
Also, $\rowAggregate(\matD{Y}{\alpha}{\beta})$ returns a row vector corresponding to an aggregate of the rows of~$\mat{Y}$.
Likewise, $\colAggregate(\matD{Y}{\alpha}{\beta})$ returns an aggregate of the columns of~$\mat{Y}$.

Given~$\matD{U}{\alpha}{\beta}$ and~$\matD{V}{\alpha}{1}$, we use~$\mat{U} \circ \mat{V}$ to denote the column-wise~Hadamard product.
Similarly, $\mat{U} \oplus \mat{V}$ denote the sum of two matrices in a column-wise fashion.
Concretely, for~$\matD{M}{\alpha}{\beta} = \mat{U} \circ \mat{V}$ and~$\matD{N}{\alpha}{\beta} = \mat{U} \oplus \mat{V}$, we have
\[
  \matpos{M}{i}{j} = \matpos{U}{i}{j} \cdot \matpos{V}{i}{1} 
  ~~~\text{and}~~~
  \matpos{N}{i}{j} = \matpos{U}{i}{j} + \matpos{V}{i}{1},
  ~~~\text{where } i \in [\alpha], j \in [\beta].
\]

\myparatight{Quantized Aggregation}
To perform aggregation on quantized inputs, a set of~$\numclients$ clients first locally prepares their quantized triples, $(\qvec{X}, \stepmin{X}, \stepmax{X})$, corresponding to their locally trained~ML model updates (i.e., gradients) and submits these to a parameter server for aggregation. 
Let~$\numcoordinates$ be the dimension of the underlying~ML model. 
The quantized triples can then be consolidated to a matrix triple of the form~$(\matD{B}{\numclients}{\numcoordinates}, \matD{U}{\numclients}{1}, \matD{V}{\numclients}{1})$.
Here, $\mat{B}$ is a binary matrix that corresponds to the~$\qvec{X}$ vector of the clients.
Similarly, $\mat{U}$ and~$\mat{V}$ correspond to the~$\stepmin{X}$ and~$\stepmax{X}$ values of the above-mentioned triple.
The quantized aggregate is defined as
\begin{equation}
  \label{eq:quantized-aggregation}
  \matD{X}{1}{\numcoordinates} = \rowAggregate\left(\matD{U}{\numclients}{1}  ~\oplus~\ \matD{B}{\numclients}{\numcoordinates} \circ (\matD{V}{\numclients}{1} - \matD{U}{\numclients}{1})\right).
\end{equation}

\myparatight{Ideal Functionality}
\label{mpc-agg-ideal-functionality}
To perform secure aggregation of quantized updates using~MPC, we model the aggregation as an ideal functionality~$\FSecAgg$~(\boxref{fig:secure-aggregation-accurate}).
We consider a set of~$\numservers$ servers to which the clients secret-share their quantized updates.
The goal is to compute the aggregate of the inputs as shown in~Eq.~\ref{eq:quantized-aggregation}.
Let~$\shr{\cdot}$ denote the underlying secret sharing scheme. 
Looking ahead, we will use linear secret sharing techniques for~MPC, in which linear operations such as addition and subtraction are local.
As a result, we will concentrate on efficiently computing the column-wise~Hadamard product.

\begin{systembox}{$\FSecAgg$}{Ideal functionality for semi-honest secure quantized aggregation for linear stochastic binary quantization.}{fig:secure-aggregation-accurate}
	\justify
	$\FSecAgg$ interacts with all the~$\numservers$ servers in~$\serverset$ and an adversary~$\Adv$ that controls a subset of the servers in~$\serverset$. 
	\begin{description}
		\item[Input:] $\FSecAgg$ receives $\shr{\cdot}$-shares of the matrix triple $(\matD{B}{\numclients}{\numcoordinates}, \matD{U}{\numclients}{1}, \matD{V}{\numclients}{1})$ from the honest servers in $\serverset$, while the adversary $\Adv$ provides the $\shr{\cdot}$-shares on behalf of the corrupt servers. Here $\matpos{B}{i}{j} \in \{0,1\}$ and $\matpos{U}{i}{j}, \matpos{V}{i}{j} \in \Z{\ell}$.
		\item[Computation:] $\FSecAgg$ reconstructs $(\mat{B}, \mat{U}, \mat{V})$ from its $\shr{\cdot}$-shares.
		\begin{boxenumerate}
			\item[--] Set $\matD{S}{\numclients}{1} = \matD{V}{\numclients}{1} - \matD{U}{\numclients}{1}$ and compute $\matD{W}{\numclients}{\numcoordinates} = \matD{B}{\numclients}{\numcoordinates} \circ \matD{S}{\numclients}{1}$.
			\item[--] Compute $\matD{X}{\numclients}{\numcoordinates} = \matD{U}{\numclients}{1} \oplus \matD{W}{\numclients}{\numcoordinates}$.
			\item[--] Compute $\matD{Y}{1}{\numcoordinates} = \rowAggregate(\matD{X}{\numclients}{\numcoordinates})$, i.e., $\matcol{Y}{j} = \sum_{i=1}^{\numclients} \matpos{X}{i}{j}$ for $j \in [\numcoordinates]$.
		\end{boxenumerate}
		\item[Output:] $\FSecAgg$ computes $\shr{\cdot}$-shares of $\mat{Y}$ and sends the respective shares to the servers in $\serverset$.
	\end{description}
\end{systembox}

\section{Our Framework: \FLname}
\label{sec:system}
We now detail our framework~\enquote{\FLname{}} from an~MPC standpoint, covering the sharing semantics, client interaction with~MPC servers, and secure aggregation of client updates. 
Our generic constructions utilize~MPC in a black-box fashion~\cite{C:DPSZ12,crypto:CramerDESX18,SP:Damgard0FKSV19,ACNS:BenNieOmr19,indocrypt:Rotaru019}, however, the full~MPC protocols, including inner product computation, multiplication, and bit-to-arithmetic conversion, are detailed in~\appref{app:mpc-protocols}.

\myparatight{Masked Evaluation}
\label{sec:masked_evaluation}
In our~MPC protocols, we use the masked evaluation technique~\cite{AC:GorRanWan18,ACNS:BenNieOmr19,ndss:ChaudhariRS20,USENIX:MLRG20,ARXIV:MPCLeague,EPRINT:KPPS22}, which enables costly \emph{data-independent} calculations to be completed in a preprocessing phase, thus enabling a fast and efficient~\emph{data-dependent} online phase~(cf.~\secref{sec:prelims_mpc}). 
In this model, the secret-share of every element~$\val \in \Z{\ell}$, denoted by~$\shr{\val}$, is associated with two values: a random mask~$\lv{\val}{} \in \Z{\ell}$ and a masked value~$\mv{\val} \in \Z{\ell}$ such that~$\val = \mv{\val} + \lv{\val}{}$. While~$\lv{\val}{}$ is split and distributed as~$\numshares$ shares as per the underlying MPC scheme~(cf.~\appref{app:mpc-protocols}), $\mv{\val}$ is given to all~MPC servers.\footnote{Due to differences in the underlying setting, there may be minor differences in how the values~$\mv{\val}$ and~$\lv{\val}{}$ are distributed among the servers.}
Since~$\lv{\val}{}$ is random and independent of~$\val$, all operations involving~$\lv{\val}{}$ values alone can be computed during the preprocessing phase and thereby leading to a fast online phase.

\myparatight{Client-Server Interaction}
\label{sec:client-interaction}
Before going into the details of aggregation among~$\numservers$~MPC servers, we discuss input sharing and the reconstruction of the aggregated vector for clients.

To generate the~$\shr{\cdot}$-shares of a value~$\val \in \Z{\ell}$ owned by client~$\client{}$, it first non-interactively computes the additive shares of the mask~$\lv{\val}{}$ using the shared randomness setup discussed in~\secref{sec:mpc-shared-randomness}.
The masked value is then computed as~$\mv{\val} = \val - \lv{\val}{}$ and sent to a single designated~MPC server, say~$\server{1}$.
The input sharing is completed when~$\server{1}$ sends~$\mv{\val}$ to all remaining~MPC servers.\footnote{If \emph{malicious privacy} is desired, $\client{}$ can send a hash digest of all the~$\mv{\val}$ values to the remaining MPC servers, who can verify the correctness of messages from~$\server{1}$.}
For~Boolean values~(i.e., in~$\Z{}$) the procedure is similar except that addition/subtraction is replaced with~XOR and multiplication with~AND.
We use~$\shrB{\cdot}$ to denote the secret sharing over the domain~$\Z{}$.

After the servers have received all inputs in~$\shr{\cdot}$-shared form, they instantiate the~$\FSecAgg$ functionality specified in~\boxref{fig:secure-aggregation-accurate} and obtain the aggregated vectors in~$\shr{\cdot}$-shared form.
Recall from~\secref{sec:fl-aggregation} that the values to be aggregated in our case correspond to~FL gradients, and the aggregated result can also be made public.
As a result, the servers reconstruct the aggregated result towards a chosen server, say~$\server{1}$, which updates the global model according to~Eq.~\eqref{eq:fedavg}. In the next iteration, $\server{1}$ distributes the updated global model to a fresh set of clients, and the process is repeated.

From a client's perspective, it interacts solely with a single server~(apart from a one-time shared-randomness setup), as in the privacy-free variant with a single parameter server~\cite{PMLR:McMahanMRHA17}.

\subsection{MPC-based Aggregation}
\label{sec:mpc-agg}
We now discuss three approaches for instantiating~$\FSecAgg$ using~MPC protocols that operate on secret-shared values.
Recall from~Eq.~\eqref{eq:quantized-aggregation} that the~MPC servers for the aggregation of quantized values possess~$\shr{\cdot}$-shares of matrices~$\matD{U}{\numclients}{1}$ and~$\matD{V}{\numclients}{1}$ along with the~$\shrB{\cdot}$-shares of~$\matD{B}{\numclients}{\numcoordinates}$.

\myparatight{Approach I}
\label{sec:approach_I}
A naive instantiation of~$\FSecAgg$ would be to have the servers convert binary shares of the matrix~$\mat{B}$ to their arithmetic shares first, as in~Prio+~\cite{SCN:AddankiGJOP22}.
This is possible in~MPC with a bit-to-arithmetic conversion protocol~$\piBitA$~\cite{CCS:MohRin18,NDSS:PatSur20,SCN:AddankiGJOP22,EPRINT:KPPS22}, which computes the arithmetic shares of~$\bitb \in \Z{}$ from its~Boolean sharing.
Once the arithmetic shares are generated, the servers can use the inner-product protocol~$\piDotP$ on each of the~$\numcoordinates$ columns of matrix~$\mat{B}$ with the locally computed column vector~$(\matD{V}{\numclients}{1} - \matD{U}{\numclients}{1})$ to obtain the row aggregate. 
They complete the protocol by adding a row aggregate of~$\mat{U}$ to each column of the matrix computed in the previous step.
The formal protocol~$\piSecAgg{1}$ is given in~\boxref{fig:SecAgg-MPC-Approach-I}. 

\begin{protocolbox}{$\piSecAgg{1} (\shrB{\matD{B}{\numclients}{\numcoordinates}}, \shr{\matD{U}{\numclients}{1}}, \shr{\matD{V}{\numclients}{1}})$}{Secure aggregation -- Approach \Romannum{1}.}{fig:SecAgg-MPC-Approach-I}
	\justify
	\begin{description}
		\item[1.] Locally compute $\shr{\matD{S}{\numclients}{1}} = \shr{\matD{V}{\numclients}{1}} - \shr{\matD{U}{\numclients}{1}}$.
		\item[2.] Compute $\shr{\mat{B}} = \piBitA(\shrB{\mat{B}})$.
		\item[3.] Compute $\shr{\matcol{W}{j}} = \piDotP(\shr{\matcol{B}{j}},\shr{\mat{S}})$, for each $j \in [\numcoordinates]$.
		\item[4.] Locally compute $\shr{\matD{Z}{1}{1}} = \rowAggregate(\shr{\matD{U}{\numclients}{1}})$.
		\item[5.] Locally compute $\shr{\matD{Y}{1}{\numcoordinates}} = \shr{\matD{Z}{1}{1}} \oplus \shr{\matD{W}{1}{\numcoordinates}}$. 
	\end{description} 
\end{protocolbox}

In~\boxref{fig:SecAgg-MPC-Approach-I}, $\piSecAgg{1}$ invokes~$\piBitA$ for each bit in matrix~$\mat{B}$, resulting in~$\numclients\cdot\numcoordinates$ invocations. 
Using the masked evaluation technique~\cite{NDSS:KPRS22,ARXIV:MPCLeague}, the online communication of the inner product protocol~$\piDotP$ can be made independent of the dimension~$\numclients$, which in our case corresponds to the number of clients.

\myparatight{Approach II}
\label{sec:approach_II}
We use the bit injection protocol~\cite{CCS:MohRin18,NDSS:PatSur20,USENIX:KPPS21}, denoted by~$\piBitInj$, which computes the arithmetic sharing of~$\bitb\cdot\scale$ given the~Boolean sharing of~$\bitb \in \Z{}$ and the arithmetic sharing of~$\scale \in \Z{\ell}$.
Given~$\shrB{\matD{M}{\alpha}{1}}$ and~$\shr{\matD{N}{\alpha}{1}}$, the high-level idea is to effectively combine the~$\piBitA$ and~$\piDotP$ protocol to a slightly modified instance of~$\piBitInj$ that directly computes the sum~\cite{NDSS:KPRS22,ARXIV:MPCLeague}~(denoted by~$\sum_{i=1}^{\alpha} \matpos{M}{i}{1} \cdot \matpos{N}{i}{1}$) instead of the individual positions. This can be achieved following~Eq.~\eqref{eq:case-bitinj-masked} and the details appear in~\boxref{fig:bitinj-sum} in~\appref{app:mpc-protocols}.
One significant advantage of this technique is that the overall online communication is no longer dependent on the number of clients~$\numclients$.
$\piSecAgg{2}$ denotes the resulting protocol and the formal details are given in~\boxref{fig:SecAgg-MPC-Approach-II}.

\begin{protocolbox}{$\piSecAgg{2} (\shrB{\matD{B}{\numclients}{\numcoordinates}}, \shr{\matD{U}{\numclients}{1}}, \shr{\matD{V}{\numclients}{1}})$}{Secure aggregation -- Approach \Romannum{2}.}{fig:SecAgg-MPC-Approach-II}
	\justify
	\begin{description}
		\item[1.] Locally compute $\shr{\matD{S}{\numclients}{1}} = \shr{\matD{V}{\numclients}{1}} - \shr{\matD{U}{\numclients}{1}}$.
		\item[2.] Compute $\shr{\matcol{W}{j}} = \piBitInj(\shrB{\matcol{B}{j}},\shr{\mat{S}})$, for each $j \in [\numcoordinates]$.
		\item[3.] Locally compute $\shr{\matD{Z}{1}{1}} = \rowAggregate(\shr{\matD{U}{\numclients}{1}})$.
		\item[4.] Locally compute $\shr{\matD{Y}{1}{\numcoordinates}} = \shr{\matD{Z}{1}{1}} \oplus \shr{\matD{W}{1}{\numcoordinates}}$. 
	\end{description} 
\end{protocolbox}

\myparatight{Approach III using SepAgg}
\label{sec:SepAgg}
In a closely related work~\cite{ARXIV:BAT20}, the authors combine the~SIGNSGD compression technique of~\cite{ICML:BernsteinWAA18} with additive secret sharing for~FL.
Unlike our work, which investigates secure aggregation using various linear quantization algorithms in a cross-device environment, they aim for a cross-silo setting in which clients distribute arithmetically shared values to servers rather than single bits.
In terms of client-server communication, this indicates a non-optimal communication overhead of factor~$\log_2 2\numclients$, where~$\numclients$ is the number of involved participants each round.

However, the authors of~\cite{ARXIV:BAT20} introduce an interesting approximation called~\enquote{SepAgg} for computing the averaged inner product between bits and scales:

\begin{equation}
    \label{eq:sepagg}
    \frac{1}{\numclients} \sum_i^\numclients \qvecV{X}{i} \cdot \stepV{X}{i} \approx \frac{1}{\numclients^2} \left(\sum_i^\numclients \qvecV{X}{i}\right) \left(\sum_i^\numclients \stepV{X}{i}\right).
\end{equation}

We adopt the~SepAgg method to our setting to compute~$\rowAggregate(\mat{B} \circ (\mat{V} - \mat{U}))$ in~Eq.~\eqref{eq:quantized-aggregation}.
In particular, we aggregate the matrices~$\mat{B}$ and~$(\mat{V}-\mat{U})$ independently and then perform one secure multiplication per coordinate, with the other operations being linear and free in our~MPC protocol.
As a result, we can utilize linear quantization schemes with local scales at the cost of global scales~(ignoring the overhead for global scales to securely determine~\stepmin{X} and~\stepmin{X} across all participants).
The formal protocol~$\piSecAgg{3}$ utilizing~SepAgg appears in~\boxref{fig:SecAgg-MPC-Approach-III}.

\begin{protocolbox}{$\piSecAgg{3} (\shrB{\matD{B}{\numclients}{\numcoordinates}}, \shr{\matD{U}{\numclients}{1}}, \shr{\matD{V}{\numclients}{1}})$}{Secure aggregation -- Approach \Romannum{3} (SepAgg~\cite{ARXIV:BAT20}).}{fig:SecAgg-MPC-Approach-III}
	\justify
	\begin{description}
		\item[1.] Locally compute $\shr{\matD{S}{1}{1}} = \rowAggregate(\shr{\matD{V}{\numclients}{1}} - \shr{\matD{U}{\numclients}{1}})$.
		\item[2.] Compute $\shr{\matcol{T}{j}} = \piBitASum(\shrB{\matcol{B}{j}})$, for each $j \in [\numcoordinates]$.
		\item[3.] Compute $\shr{\matcol{W}{j}} = \piMult(\shr{\matcol{T}{j}},\shr{\mat{S}})$, for each $j \in [\numcoordinates]$.
		\item[4.] Locally compute $\shr{\matD{Z}{1}{1}} = \rowAggregate(\shr{\matD{U}{\numclients}{1}})$.
		\item[5.] Locally compute $\shr{\matD{Y}{1}{\numcoordinates}} = \shr{\matD{Z}{1}{1}} \oplus \frac{1}{\numclients} \cdot \shr{\matD{W}{1}{\numcoordinates}}$. 
	\end{description} 
\end{protocolbox}

\myparatight{Accuracy Evaluation}
\label{sec:sepagg_eval}
We next provide strong empirical evidence that applying~SepAgg for~SQ with preprocessing preserves the linear~NMSE decay with respect to the number of clients~(i.e., unbiased estimates).
For this, we simulate the aggregation of random vectors~$\vec{v}_i$ with dimension~$d$ drawn from a~$(0, 1)$-log-normal distribution.\footnote{We use this distribution for preliminary measurements as it was commonly observed in neural network gradients,~e.g., \cite{iclr:ChmielBSHBS21}.}
Then, we measure the normalized mean square error~(NMSE) when comparing the averaged aggregation result~$\textit{agg}$ computed on secret-shared and quantized inputs to the plain averaged aggregation~$\textit{agg}_\textit{orig} = \sum_i^\numclients \vec{v}_i / \numclients$.
Concretely, we measure
\begin{equation}
\label{eq:nmse}
\textit{NMSE} = \frac{\lVert \textit{agg}_\textit{orig} - \textit{agg} \rVert_2^2}{\sum_i^\numclients \lVert \vec{v}_i \rVert_2^2 / \numclients},
\end{equation}
where~$\textit{agg}$ is computed using various linear quantization schemes for~(i)~a regular dot product between converted bits and scales and~(ii) with~SepAgg.
The code for the implementation of our simulation framework is available at \url{https://encrypto.de/code/ScionFL}.
Our results shown in~Fig.~\ref{fig:sepagg_sq} are the average of~10 trials for each experiment with~$\numshares = 3$ shares~(representing a three-server dishonest majority setting using the masked evaluation technique).
While for smaller dimensions we observe a visible effect for~SQ without preprocessing, there is only a minor difference for the other two quantization schemes, and sometimes the~NMSE for~SepAgg is even smaller than for the exact computation.

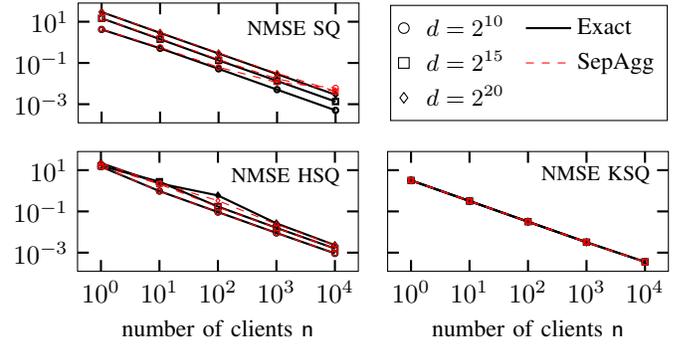
\begin{figure}[htb!]
    \centering
    \begin{tabular}{l}

\hspace{-1.125em}

\begin{tikzpicture}
\begin{axis}[xmode=log,xticklabels=\empty, log basis x={10}, xminorticks=false, ymode=log, width=0.6\columnwidth, height=3.175cm, ymax=80, ymin =0.000125, ylabel style={align=center}, legend style={at={(2.1,0.5)},anchor=east}, legend entries={\phantom{x}$d=2^{10}$,Exact,\phantom{x}$d=2^{15}$,SepAgg,\phantom{x}$d=2^{20}$},legend columns=2,legend style={/tikz/every even column/.append style={column sep=0.2cm}},legend style={font=\small},legend cell align=left]

\addlegendimage{only marks,mark=o}
\addlegendimage{no markers,black,thick}
\addlegendimage{only marks,mark=square}
\addlegendimage{no markers,red, dashed}
\addlegendimage{only marks,mark=diamond}

\addplot[mark=o, thick, mark size=1pt] table {SepAgg_SQ_10_Exact.dat};
\addplot[mark=square, thick, mark size=1pt] table {SepAgg_SQ_15_Exact.dat};
\addplot[mark=diamond, thick, mark size=1pt] table {SepAgg_SQ_20_Exact.dat};

\addplot[mark=o, red, dashed, mark options=solid, mark size=1pt] table {SepAgg_SQ_10_SA.dat};
\addplot[mark=square, red, dashed, mark options=solid, mark size=1pt] table {SepAgg_SQ_15_SA.dat};
\addplot[mark=diamond, red, dashed, mark options=solid, mark size=1pt] table {SepAgg_SQ_20_SA.dat};

\node[below,font=\footnotesize,yshift=0.0cm,xshift=1.0cm] at (current bounding box.north) {NMSE SQ};

\end{axis}
\end{tikzpicture}

\end{tabular}

\begin{tabular}{ll}

\hspace{-1.15em}

\begin{tikzpicture}
\begin{axis}[xmode=log, log basis x={10}, xminorticks=false, ymode=log, width=0.6\columnwidth, height=3.175cm, ymax=80, xlabel={{\small number of clients~$\numclients$}}, ymin =0.000125, ylabel style={align=center}]

\addplot[mark=o, thick, mark size=1pt] table {SepAgg_HD_10_Exact.dat};
\addplot[mark=square, thick, mark size=1pt] table {SepAgg_HD_15_Exact.dat};
\addplot[mark=diamond, thick, mark size=1pt] table {SepAgg_HD_20_Exact.dat};

\addplot[mark=o, red, dashed, mark options=solid, mark size=1pt] table {SepAgg_HD_10_SA.dat};
\addplot[mark=square, red, dashed, mark options=solid, mark size=1pt] table {SepAgg_HD_15_SA.dat};
\addplot[mark=diamond, red, dashed, mark options=solid, mark size=1pt] table {SepAgg_HD_20_SA.dat};

\node[below,font=\footnotesize,yshift=0.05cm,xshift=0.9cm] at (current bounding box.north) {NMSE HSQ};

\end{axis}
\end{tikzpicture}

&

\hspace{-1.3em}

\begin{tikzpicture}
\begin{axis}[xmode=log, yticklabel=\empty, log basis x={10}, xminorticks=false, ymode=log, width=0.6\columnwidth, height=3.175cm, ymax=80, ymin =0.000125, xlabel={{\small number of clients~$\numclients$}}, ylabel style={align=center}]

\addplot[mark=o, thick, mark size=1pt] table {SepAgg_KS_10_Exact.dat};
\addplot[mark=square, thick, mark size=1pt] table {SepAgg_KS_15_Exact.dat};
\addplot[mark=diamond, thick, mark size=1pt] table {SepAgg_KS_20_Exact.dat};

\addplot[mark=o, red, dashed, mark options=solid, mark size=1pt] table {SepAgg_KS_10_SA.dat};
\addplot[mark=square, red, dashed, mark options=solid, mark size=1pt] table {SepAgg_KS_15_SA.dat};
\addplot[mark=diamond, red, dashed, mark options=solid, mark size=1pt] table {SepAgg_KS_20_SA.dat};

\node[below,font=\footnotesize,yshift=0.3cm,xshift=0.9cm] at (current bounding box.north) {NMSE KSQ};

\end{axis}
\end{tikzpicture}

\end{tabular}
    \vspace{-2mm}
    \captionsetup{font=small}
    \caption{NMSE comparison between exact and~SepAgg-based aggregation for vanilla~SQ, SQ using the randomized~Hadamard transform~(HSQ), and~SQ using~Kashin's representation~(KSQ) for various vector dimensions~$d$ and number of clients~$\numclients$.}
    \label{fig:sepagg_sq}
    \vspace{-4mm}
\end{figure}

Formally proving that applying~SepAgg after different preprocessing techniques~(e.g., random rotation and~Kashin's representation) results in an unbiased aggregation is a significant theoretical challenge, left for future work.

\myparatight{Communication Costs}
Tab.~\ref{tab:fl-costs-theory} provides the theoretical communication costs for our approaches when aggregating~$\numclients$ quantized single-dimension vectors. Clearly, our~Approach-III~(cf.~\boxref{fig:SecAgg-MPC-Approach-III}) is the most efficient, with the multiplication-related cost being completely independent of the number of clients~$\numclients$ due to~SepAgg~\cite{ARXIV:BAT20}. The concrete communication costs are provided in ~\appref{sec:communication_costs_detailed}.

The communication costs are primarily determined by the cost of~$\bitAPreCost$. In our masked evaluation technique~(cf.~\appref{app:mpc-protocols}), this relates to the conversion of a random secret-shared bit from~Boolean sharing to its additive sharing form~\cite{SP:Damgard0FKSV19,indocrypt:Rotaru019}. Thus, we propose approximate variants for improving the cost of this operation.

\begin{table}[htb!]
    \centering
    \resizebox{0.95\columnwidth}{!}{%
    \begin{tabular}{rrr}
    \toprule
    Approach     
    & \multicolumn{1}{c}{Offline}   
    & \multicolumn{1}{c}{Online} \\ \midrule
    Approach-I   
    & $\numclients \cdot \bitAPreCost + \numclients \cdot \MultPreCost$ 
    & $\numclients \cdot \bitAOnCost  + \MultOnCost$  \\ 
    Approach-II 
    & $\numclients \cdot \bitAPreCost + \numclients \cdot \MultPreCost$ 
    & $\MultOnCost$ \\ 
    Approach-III
    & $\numclients \cdot \bitAPreCost + \phantom{\numclients} \hspace{0.5em} \phantom{\cdot} \MultPreCost$ 
    & $\MultOnCost$\\ 
    \bottomrule
    \end{tabular}%
    }
    \vspace{-1mm}
    \captionsetup{font=small}
    \caption{Communication costs for aggregating quantized vectors with a single dimension for~$\numclients$ clients. Protocols~$\piBitA$ and~$\pimult$ are treated as black-boxes, and their costs are represented as~$\bitACost$ and~$\MultCost$, respectively. The superscript {\sf pre} in the costs denotes preprocessing and {\sf on} denotes the online phase.}
    \label{tab:fl-costs-theory}
\end{table}

\subsection{Approximate Bit Conversion in MPC}
\label{sec:approximate_mpc}
We reduce communication and computation costs with a novel~\emph{approximate} bit conversion method.
Consider a bit~$\bitb$ represented using two~Boolean shares~$\bitb_1, \bitb_2 \in \{0,1\}$, such that~$\bitb = \bitb_1 \xor \bitb_2$.
Note that when embedding~$\bitb_1$ and~$\bitb_2$ in a larger field/ring\footnote{The bit~(either~$0$ or~$1$) is treated as a ring element in~$\Z{\ell}$ in our protocols.}, it holds that~$\bitb = \bitb_1 + \bitb_2 - 2\bitb_1 \bitb_2$.
Similarly, for~$\bitb = \bitb_1 \xor \bitb_2 \xor \bitb_3$, $\bitb = \bitb_1 + \bitb_2 + \bitb_3 - 2\bitb_1\bitb_2 - 2\bitb_1\bitb_3 - 2\bitb_2\bitb_3 + 4\bitb_1\bitb_2\bitb_3$ holds true.
This concept generalizes to an arbitrary number of shares, denoted by~$\numshares$, as discussed below.

For~$\bitb = \xor_{i=1}^{\numshares} \bitb_i$, let~$\shareset = \{\bitb_i\}_{i \in [\numshares]}$ denote the set of all~$\numshares$ shares of~$\bitb$, and~$\bitbext_i$ the arithmetic equivalent of the share~$\bitb_i$.
Let~$2^{\shareset}$ be the powerset of~$\shareset$ and~$\shareset^{|c|}$ the set of all size-$c$ subsets in~$2^{\shareset}$, that is, $2^{\shareset} = \sum_{i=0}^q \shareset^{|i|}$.
The arithmetic equivalent of~$\bitb$, denoted by~$\bitbAr$, is given as
\begin{center}
\resizebox{.9\linewidth}{!}{
  \begin{minipage}{\linewidth}
    \begin{align}
        \label{eq:bit-arithmetic-general}
    	\bitbAr &= \smashoperator{\sum\limits_{ \{\bitb_e\} \in \shareset^{|1|}}} \bitbext_e 
    	~-~ 2~\cdot~\smashoperator{\sum\limits_{\{\bitb_{e_1},\bitb_{e_2}\} \in \shareset^{|2|}}} \bitbext_{e_1}\bitbext_{e_2}
    	~+~ \ldots + (-2)^{\numshares-1}~\cdot~ \smashoperator{\prod\limits_{\{\bitb_{e_1},\ldots,\bitb_{e_q}\} \in \shareset^{|q|}}} \bitbext_e \nonumber \\
    	&= \sum\limits_{k=1}^{\numshares} (-2)^{k-1} \smashoperator{\sum\limits_{\{\bitb_{e_1},\ldots,\bitb_{e_k}\} \in \shareset^{|k|} }} \bitbext_{e_1} \bitbext_{e_2} \ldots \bitbext_{e_k}
    \end{align}
  \end{minipage}
}
\end{center}

Note that the Eq.~\eqref{eq:bit-arithmetic-general} can be viewed as sum of three terms: Sum~($\termSum$), Middle~($\termMid$), and Product~($\termProd$) as shown in Eq.~\eqref{eq:bit-arithmetic-general-parts-old} below. (Note that~$\shareset^{|q|}=\shareset$).
\resizebox{.9\linewidth}{!}{
  \begin{minipage}{\linewidth}
    \begin{align}
    \label{eq:bit-arithmetic-general-parts-old}
    	\bitbAr = \underbrace{\smashoperator{\sum\limits_{\{\bitb_e\} \in \shareset^{|1|}}} \bitbext_e}_{\text{Sum Term}:~\termSum} + \underbrace{\sum\limits_{k=2}^{\numshares-1} (-2)^{k-1} \smashoperator{\sum\limits_{\{\bitb_{e_1},\ldots,\bitb_{e_k}\} \in \shareset^{|k|} }} \bitbext_{e_1} \bitbext_{e_2} \ldots \bitbext_{e_k}}_{\text{Middle Term}:~\termMid} + \underbrace{(-2)^{\numshares-1} \prod\limits_{\bitb_e \in \shareset} \bitbext_e}_{\text{Product Term}:~\termProd}
    \end{align}
  \end{minipage}
}

\tsparagraph{Our Approach}
\label{sec:our-approx-approach}
Performing this conversion in~MPC requires many additions and multiplications.
While linear operations like additions can be calculated for~\enquote{free} in most~MPC protocols, non-linear operations such as multiplications require some form of communication between the MPC servers.
Hence, computing the middle term is costly, especially when a large number of shares is involved.

To approximate~$\bitbAr$ in~Eq.~\eqref{eq:bit-arithmetic-general-parts-old}, we replace only term~$\termMid$ with its expected value~$\E[\termMid]$ such that the approximate value of~$\bitbAr$, denoted by $\bitbapprox$, retains~$\E[\bitbapprox]=\bitb$.
The expectation of~$\termSum$ and~$\termProd$ in~Eq.~\eqref{eq:bit-arithmetic-general-parts-old} is first calculated, and~$\E[\termMid]$ is inferred using the fact that~$\E[\bitbapprox]=\bitb$.
This analysis is summarised in~Lem.~\ref{lemma:bitA-exp-analysis} and the proof is provided in \secref{app:approximation-proof}.

\begin{restatable}[Expected Values]{lemma}{binomialexpectation} 
   \label{lemma:bitA-exp-analysis}
   Given a bit~$\bitb = \xor_{i=1}^{\numshares} \bitb_i$ and~$b= \termSum + \termMid + \termProd$ with 
   \begin{align*}
       \termSum &= \smashoperator{\sum\limits_{\{\bitb_e\} \in \shareset^{|1|}}} \bitbext_e, 
       ~~\termMid = \sum\limits_{k=2}^{\numshares-1} (-2)^{k-1} \smashoperator{\sum\limits_{\{\bitb_{e_1},\ldots,\bitb_{e_k}\} \in \shareset^{|k|} }} \bitbext_{e_1} \bitbext_{e_2} \ldots \bitbext_{e_k},\\
       \termProd &= (-2)^{\numshares-1} \prod\limits_{\bitb_e \in \shareset} \bitbext_e,
   \end{align*}
   we have~$\E[\termSum \mid \bitb] = \numshares/2$, $\E[\termMid \mid \bitb] = (\numshares\mbox{-}1)\Mod{2} - \numshares/2$, and~$\E[\termProd \mid \bitb] = \bitb - (\numshares\mbox{-}1) \Mod{2}$.
\end{restatable}

\tsparagraph{Our Approximation}
\label{sec:our-approx}
We define the approximate arithmetic equivalent of $\bitb$, denoted by $\bitbapprox$, as follows:
\begin{equation}
    \label{eq:approx-equation-middle-term}
	\bitbapprox = \underbrace{\smashoperator{\sum\limits_{\bitb_e \in \shareset}} \bitbext_e}_{\termSum} + \underbrace{\left((\numshares\mbox{-}1)\Mod{2} - \frac{\numshares}{2}\right)}_{\termMidApx} + \underbrace{(-2)^{\numshares-1} \prod\limits_{\bitb_e \in \shareset} \bitbext_e}_{\termProd}
\end{equation}
While~$\termSum$ is kept because it only involves linear operations on the shares of~$\bitb$~(which are free in~MPC for any linear secret sharing scheme), we observe that~$\termProd$ is required to keep the expected values for~$\bitb = 0$ and~$\bitb = 1$ different. 
This is evident from~Lem.~\ref{lemma:bitA-exp-analysis} where~$\E[\termProd \mid \bitb]$ is the only term that depends on~$\bitb$.

In general, if a term that depends on all the~$\numshares$ shares of~$\bitb$ is missing from the approximation, we get~$\E[\bitb = 0] = \E[\bitb = 1]$.
The intuition is that only such a term can differentiate between~$\bitb = 0$ and~$\bitb = 1$, while all other terms will be symmetrically distributed. 
For instance, consider~$\numshares = 3$ and let~$\bitbAr = c_1\bitbext_1 + c_2\bitbext_2 + c_3\bitbext_3 + c_4\bitbext_1\bitbext_2 + c_5\bitbext_2\bitbext_3 + c_6\bitbext_1\bitbext_3 + c_7$ for some random combiners~$c_i \in \Z{\ell}$ and~$i \in [7]$.
Using the truth table~$\truthtable{\bitb}$ given in~Tab.~\ref{tab:truth-table-three-shares}, it is easy to verify that 
\begin{small}
\begin{equation}
    \E[\bitbAr = 0] = \E[\bitbAr = 1] = \frac{1}{4} \cdot (2c_1 + 2c_2 + 2c_3 + c_4 + c_5 + c_6 + 4c_7)
\end{equation}
\end{small}
\noindent
This argument can be generalized to any value of~$\numshares$.

\begin{claim}
	The approximate arithmetic equivalent~$\bitbapprox$ in~Eq.~\eqref{eq:approx-equation-middle-term} preserves the expectation of the exact bit~$\bitb$ in~Eq.~\eqref{eq:bit-arithmetic-general}, i.e., $\E[\bitbapprox = 0] = 0$ and~$\E[\bitbapprox = 1] = 1$.
\end{claim}
\begin{proof}
   The proof is straightforward as we replace the middle term ($\termMid$) in~Eq.~\eqref{eq:bit-arithmetic-general} with its expected value $\termMidApx$.
\end{proof}

We provide more details regarding the efficiency of the approximation in~\secref{sec:approx_efficiency}.

\subsection{Secure Bit Aggregation with Global Scales}
\label{sec:global-scales-agg}
Here, we consider secure bit aggregation in the context of~\enquote{global scales}, as discussed in~\secref{sec:linear-sq}. In this case, all the clients use the same set of scales for quantization, denoted by~$\stepmin{G}$ and~$\stepmax{G}$. Therefore, it is sufficient to compute
\begin{equation}
    \label{eq:straw_man}
 \matD{X}{1}{\numcoordinates} = \stepmin{G} \ ~\oplus~\ \rowAggregate\left(\matD{B}{\numclients}{\numcoordinates}\right) \circ (\stepmax{G} - \stepmin{G})
\end{equation}
as the aggregation result. Interestingly, when~$\stepmin{G}=0$ and~$\stepmax{G}=~1$, this can also be viewed as an instance of privacy-preserving aggregate statistics computation, as demonstrated in the works of~Prio~\cite{nsdi:Corrigan-GibbsB17} and~Prio+~\cite{SCN:AddankiGJOP22}.

As shown in~Eq.~\eqref{eq:straw_man}, the computation becomes simpler in the case of global scales since all clients utilize the same set of public scales, denoted by~$\stepmin{}$ and~$\stepmax{}$, to compute their quantized vector that corresponds to the rows of~$\mat{B}$. Hence, we just need to compute the column-wise aggregate of the~$\mat{B}$ matrix and use protocol~$\piBitASum$~(\boxref{fig:bitA-sum} in~\secref{app:mpc-protocols}) to do so. The resulting protocol~$\piSecAggG$ appears in~\boxref{fig:SecAgg-MPC-Global-Scales}.

\begin{protocolbox}{$\piSecAggG(\shrB{\matD{B}{\numclients}{\numcoordinates}}, \stepmin{}, \stepmax{})$}{Secure aggregation -- Global Scales.}{fig:SecAgg-MPC-Global-Scales}
	\justify
	\begin{description}
		\item[1.] Compute $\shr{\matcol{W}{j}} = \piBitASum(\shrB{\matcol{B}{j}})$, for each $j \in [\numcoordinates]$.
		\item[2.] Locally compute $\shr{\matD{Y}{1}{\numcoordinates}} = \stepmin{} \oplus  \left(\shr{\matD{W}{1}{\numcoordinates}} \cdot (\stepmax{} - \stepmin{})\right)$. 
	\end{description} 
\end{protocolbox}

\subsection{Accuracy Evaluation}
\label{sec:approx_evaluation}
In~\secref{sec:approximate_mpc}, we showed that our approximate bit conversion preserves the expectation of the exact bits. However, we also want to understand the concrete accuracy impact on the aggregation result due to the increased variance.
For this, we run a simulation similar to the one described in~\secref{sec:sepagg_eval}.
Here, we compare the~NMSE computed as in~Eq.~\eqref{eq:nmse} for an aggregation~$\textit{agg}$ when using various linear quantization schemes with global scales~(i) with an exact bit-to-arithmetic conversion and~(ii) with our approximation enabled.
The implementation is available at \url{https://encrypto.de/code/ScionFL}.
Our results in~Fig.~\ref{fig:approx_sq} are the average of~10 trials for each experiment with~$\numshares = 3$ shares. Consistently, we observe that our approximation increases the~NMSE by about three orders of magnitude for stochastic quantization without rotation, and by about one and a half orders of magnitude for rotation-based algorithms.
In~Fig.~\ref{fig:approx_sq_local}, we provide results considering local scales. In contrast to global scales, we can observe that for stochastic quantization without rotation the effect on the~NMSE is reduced from three to one order of magnitude. Also, for rotation-based algorithms there are significant concrete improvements. Furthermore, as shown in~\secref{sec:sys_performance}, the error is still so small that the impact on the accuracy in common~FL settings is negligible.

\begin{figure}[htb!]
    \centering
    \begin{tabular}{l}

\hspace{-1.125em}

\begin{tikzpicture}
\begin{axis}[xmode=log,xticklabels=\empty, log basis x={10}, xminorticks=false, ymode=log, width=0.6\columnwidth, height=3.175cm, ymax=750000, ymin =0.000125, ylabel style={align=center}, legend style={at={(2.1,0.5)},anchor=east}, legend entries={\phantom{x}$d=2^{10}$,Exact,\phantom{x}$d=2^{15}$,Approx.,\phantom{x}$d=2^{20}$},legend columns=2,legend style={/tikz/every even column/.append style={column sep=0.2cm}},legend style={font=\small},legend cell align=left]

\addlegendimage{only marks,mark=o}
\addlegendimage{no markers,black,thick}
\addlegendimage{only marks,mark=square}
\addlegendimage{no markers,red, dashed}
\addlegendimage{only marks,mark=diamond}

\addplot[mark=o, thick, mark size=1pt] table {Global_SQ_10_Exact.dat};
\addplot[mark=square, thick, mark size=1pt] table {Global_SQ_15_Exact.dat};
\addplot[mark=diamond, thick, mark size=1pt] table {Global_SQ_20_Exact.dat};

\addplot[mark=o, red, dashed, mark options=solid, mark size=1pt] table {Global_SQ_10_Approx.dat};
\addplot[mark=square, red, dashed, mark options=solid, mark size=1pt] table {Global_SQ_15_Approx.dat};
\addplot[mark=diamond, red, dashed, mark options=solid, mark size=1pt] table {Global_SQ_20_Approx.dat};

\node[below,font=\footnotesize,yshift=0.2cm,xshift=1.0cm] at (current bounding box.north) {NMSE SQ};

\end{axis}
\end{tikzpicture}

\end{tabular}

\begin{tabular}{ll}

\hspace{-1.15em}

\begin{tikzpicture}
\begin{axis}[xmode=log, log basis x={10}, xminorticks=false, ymode=log, width=0.6\columnwidth, height=3.175cm, ymax=750000, ymin =0.000125, xlabel={{\small number of clients~$\numclients$}}, ylabel style={align=center}]

\addplot[mark=o, thick, mark size=1pt] table {Global_HD_10_Exact.dat};
\addplot[mark=square, thick, mark size=1pt] table {Global_HD_15_Exact.dat};
\addplot[mark=diamond, thick, mark size=1pt] table {Global_HD_20_Exact.dat};

\addplot[mark=o, red, dashed, mark options=solid, mark size=1pt] table {Global_HD_10_Approx.dat};
\addplot[mark=square, red, dashed, mark options=solid, mark size=1pt] table {Global_HD_15_Approx.dat};
\addplot[mark=diamond, red, dashed, mark options=solid, mark size=1pt] table {Global_HD_20_Approx.dat};

\node[below,font=\footnotesize,yshift=0.05cm,xshift=0.9cm] at (current bounding box.north) {NMSE HSQ};

\end{axis}
\end{tikzpicture}

&

\hspace{-1.3em}

\begin{tikzpicture}
\begin{axis}[xmode=log, log basis x={10}, xminorticks=false, ymode=log, width=0.6\columnwidth, height=3.175cm, ymax=750000, yticklabel=\empty, ymin =0.000125, xlabel={{\small number of clients~$\numclients$}}, ylabel style={align=center}]

\addplot[mark=o, thick, mark size=1pt] table {Global_KS_10_Exact.dat};
\addplot[mark=square, thick, mark size=1pt] table {Global_KS_15_Exact.dat};
\addplot[mark=diamond, thick, mark size=1pt] table {Global_KS_20_Exact.dat};

\addplot[mark=o, red, dashed, mark options=solid, mark size=1pt] table {Global_KS_10_Approx.dat};
\addplot[mark=square, red, dashed, mark options=solid, mark size=1pt] table {Global_KS_15_Approx.dat};
\addplot[mark=diamond, red, dashed, mark options=solid, mark size=1pt] table {Global_KS_20_Approx.dat};

\node[below,font=\footnotesize,yshift=0.55cm,xshift=0.9cm] at (current bounding box.north) {NMSE KSQ};

\end{axis}
\end{tikzpicture}

\end{tabular}
    \vspace{-2mm}
    \captionsetup{font=small}
    \caption{NMSE comparison between exact and~approximation-based aggregation for vanilla~SQ, SQ using the randomized~Hadamard transform~(HSQ), and~SQ using~Kashin's representation~(KSQ) for global scales with~$\numshares = 3$ shares and various vector dimensions~$d$ and number of clients~$\numclients$.}
    \label{fig:approx_sq}
    \vspace{-5mm}
\end{figure}
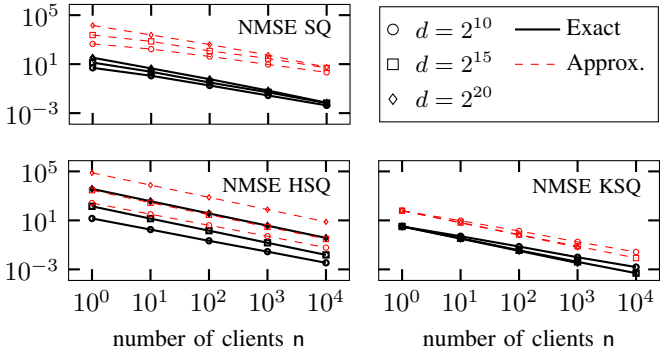

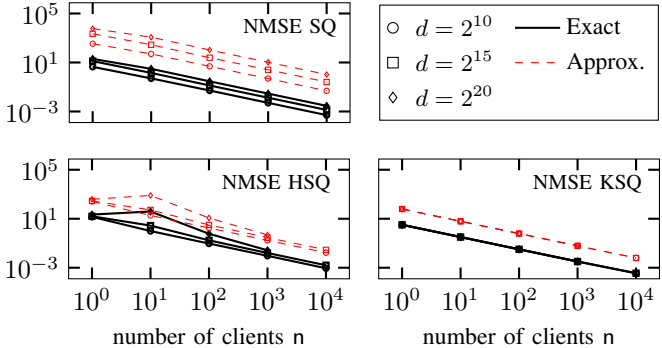
\begin{figure}[htb!]
    \centering
    \begin{tabular}{l}

\hspace{-1.125em}

\begin{tikzpicture}
\begin{axis}[xmode=log,xticklabels=\empty, log basis x={10}, xminorticks=false, ymode=log, width=0.6\columnwidth, height=3.175cm, ymax=750000, ymin =0.000125, ylabel style={align=center}, legend style={at={(2.1,0.5)},anchor=east}, legend entries={\phantom{x}$d=2^{10}$,Exact,\phantom{x}$d=2^{15}$,Approx.,\phantom{x}$d=2^{20}$},legend columns=2,legend style={/tikz/every even column/.append style={column sep=0.2cm}},legend style={font=\small},legend cell align=left]

\addlegendimage{only marks,mark=o}
\addlegendimage{no markers,black,thick}
\addlegendimage{only marks,mark=square}
\addlegendimage{no markers,red, dashed}
\addlegendimage{only marks,mark=diamond}

\addplot[mark=o, thick, mark size=1pt] table {Local_SQ_10_Exact.dat};
\addplot[mark=square, thick, mark size=1pt] table {Local_SQ_15_Exact.dat};
\addplot[mark=diamond, thick, mark size=1pt] table {Local_SQ_20_Exact.dat};

\addplot[mark=o, red, dashed, mark options=solid, mark size=1pt] table {Local_SQ_10_Approx.dat};
\addplot[mark=square, red, dashed, mark options=solid, mark size=1pt] table {Local_SQ_15_Approx.dat};
\addplot[mark=diamond, red, dashed, mark options=solid, mark size=1pt] table {Local_SQ_20_Approx.dat};

\end{axis}

\node[below,font=\footnotesize,yshift=-0.2cm,xshift=-0.5cm] at (current bounding box.north) {NMSE SQ};

\end{tikzpicture}

\end{tabular}

\begin{tabular}{ll}

\hspace{-1.15em}

\begin{tikzpicture}
\begin{axis}[xmode=log, log basis x={10}, xminorticks=false, ymode=log, width=0.6\columnwidth, height=3.175cm, ymax=750000, ymin =0.000125, xlabel={{\small number of clients~$\numclients$}}, ylabel style={align=center}]

\addplot[mark=o, thick, mark size=1pt] table {Local_HD_10_Exact.dat};
\addplot[mark=square, thick, mark size=1pt] table {Local_HD_15_Exact.dat};
\addplot[mark=diamond, thick, mark size=1pt] table {Local_HD_20_Exact.dat};

\addplot[mark=o, red, dashed, mark options=solid, mark size=1pt] table {Local_HD_10_Approx.dat};
\addplot[mark=square, red, dashed, mark options=solid, mark size=1pt] table {Local_HD_15_Approx.dat};
\addplot[mark=diamond, red, dashed, mark options=solid, mark size=1pt] table {Local_HD_20_Approx.dat};

\end{axis}

\node[below,font=\footnotesize,yshift=-0.2cm,xshift=1.35cm] at (current bounding box.north) {NMSE HSQ};

\end{tikzpicture}

&

\hspace{-1.3em}

\begin{tikzpicture}
\begin{axis}[xmode=log, log basis x={10}, xminorticks=false, ymode=log, width=0.6\columnwidth, height=3.175cm, ymax=750000, ymin =0.000125, xlabel={{\small number of clients~$\numclients$}}, ylabel style={align=center},yticklabels=\empty]

\addplot[mark=o, thick, mark size=1pt] table {Local_KS_10_Exact.dat};
\addplot[mark=square, thick, mark size=1pt] table {Local_KS_15_Exact.dat};
\addplot[mark=diamond, thick, mark size=1pt] table {Local_KS_20_Exact.dat};

\addplot[mark=o, red, dashed, mark options=solid, mark size=1pt] table {Local_KS_10_Approx.dat};
\addplot[mark=square, red, dashed, mark options=solid, mark size=1pt] table {Local_KS_15_Approx.dat};
\addplot[mark=diamond, red, dashed, mark options=solid, mark size=1pt] table {Local_KS_20_Approx.dat};

\end{axis}

\node[below,font=\footnotesize,yshift=-0.1cm,xshift=1.0cm] at (current bounding box.north) {NMSE KSQ};
\end{tikzpicture}

\end{tabular}
    \vspace{-2mm}
    \captionsetup{font=small}
    \caption{NMSE comparison between exact and approximation-based aggregation for~SQ, Hadamard~SQ~(HSQ), and~Kashin~SQ~(KSQ) for local scales with~$\numshares = 3$ shares, various vector dimensions~$d$, and number of clients~$\numclients$.}
    \label{fig:approx_sq_local}
    \vspace{-2mm}
\end{figure}

\subsection{Detailed Communication Costs}
\label{sec:communication_costs_detailed}
Here, we provide more insights into the concrete communication costs for our secure aggregation protocols in~\secref{sec:mpc-agg}.

In~Tab.~\ref{tab:fl_communication_full} we provide the detailed communication costs for the secure aggregation approaches discussed in~\secref{sec:mpc-agg} when training the~LeNet architecture for image classification on the~MNIST data set~\cite{pieee:LeCunBBH98} using~$1$-bit~SQ with~Kashin's representation~\cite{ARXIV:CKMT18}.  
We instantiate the~OT instances required in the preprocessing phase, as discussed in~\secref{app:mpc-protocols}, with silent OT~\cite{C:CouteauRR21}, following Prio+~\cite{SCN:AddankiGJOP22}.
Here, we can observe the significant impact of including~SepAgg~\cite{ARXIV:BAT20} in practice with performance improvements between~Approach-II and~Approach-III of up to~$16.6\times$ in the offline phase.

\begin{table}[htb!]
    \centering
    \resizebox{0.95\columnwidth}{!}{%
    \begin{tabular}{rlrrrr}
    \toprule
        && \multicolumn{2}{c}{Exact} & \multicolumn{2}{c}{Approx.} \\\cline{3-4}\cline{5-6}
        \numclients & Method & Offline & Online & Offline & Online \\\midrule
        \multirow{3}{*}{20} 
        & Approach-I   & 644.50 & 1.70 & 620.27 & 1.70 \\
        & Approach-II  & 644.50 & 0.59 & 620.27& 0.59 \\
        & Approach-III & 89.77 & 0.59 & 65.54 & 0.59 \\
        \midrule
        \multirow{3}{*}{100} 
        & Approach-I   & 3222.51 & 6.12 &3101.36 & 6.12 \\
        & Approach-II  & 3222.51 & 0.59  &3101.36 & 0.59 \\
        & Approach-III & 332.08 & 0.59   &210.93  & 0.59 \\
        \midrule
        \multirow{3}{*}{500} 
        & Approach-I   & 16112.56 & 28.24 & 15506.80& 28.24 \\
        & Approach-II  & 16112.56 & 0.59 & 15506.80& 0.59 \\
        & Approach-III & 1543.62 & 0.59 &   937.85 & 0.59 \\
        \bottomrule
    \end{tabular}
    }
    \captionsetup{font=small}
    \caption{Inter-server communication per round in~MiB for our~MNIST/LeNet benchmark for different numbers of clients~$\numclients$ per round. Training is done using~$1$-bit~SQ with~Kashin's representation~(KSQ). We compare~Approach-I~(cf.~\boxref{fig:SecAgg-MPC-Approach-I} in~\secref{sec:approach_I}), Approach-II~(cf.~\boxref{fig:SecAgg-MPC-Approach-II} in~\secref{sec:approach_II}), and~Approach-III~(cf.~\boxref{fig:SecAgg-MPC-Approach-III} in~\secref{sec:SepAgg}). Additionally, we distinguish between using an exact bit-to-arithmetic conversion and our approximation~(cf.~\secref{sec:our-approx}).}
    \label{tab:fl_communication_full}
    \vspace{-4mm}
\end{table}

In~Tab.~\ref{tab:Prio+vsours}, we compare the aggregation of bits~(i.e., when not considering quantized inputs that require scale multiplication and hence without~SepAgg~\cite{ARXIV:BAT20} being applicable) to~Prio+~\cite{SCN:AddankiGJOP22}.
For a fair comparison, we translate the approach in~Prio+~\cite{SCN:AddankiGJOP22} to our three party dishonest-majority setting.
As we can see, even for exact bit-to-arithmetic conversion, we improve over~Prio+ by factor~$2.4\times$ for~$\numclients = 10^5$.
When applying our approximate bit-to-arithmetic conversion~(cf.~\secref{sec:our-approx}), this improvement increases to a factor of~$4\times$.

\begin{table}[htb!]
    \centering
    \scalebox{0.95}{%
    \begin{tabular}{lrrrr}
    \toprule
     Approach & $\numclients=10^2$& $\numclients=10^3$& $\numclients=10^4$& $\numclients=10^5$ \\ \midrule
    Prio+~\cite{SCN:AddankiGJOP22}& 9.45 & 94.50& 945.04& 9450.44\\ 
    Approach-III (Exact)& 3.94& 39.42&394.17& 3941.66  \\ 
    Approach-III (Approx.) & 2.37& 23.75& 237.45& 2374.53 \\ \bottomrule
    \end{tabular}%
    }
    \captionsetup{font=small}
    \caption{Total communication in~MiB of~Approach-III~(cf.~\boxref{fig:SecAgg-MPC-Approach-III} in~\secref{sec:SepAgg}) compared to~Prio+~\cite{SCN:AddankiGJOP22} to calculate the sum of bits for different numbers of clients~$\numclients$ and dimension~$\numcoordinates=1000$. For~Approach-III, we distinguish between using an exact bit-to-arithmetic conversion as in~Prio+~\cite{SCN:AddankiGJOP22} and our approximation~(cf.~\secref{sec:our-approx}).}
    \label{tab:Prio+vsours}
    \vspace{-4mm}
\end{table}

\subsection{Performance Evaluation}
\label{sec:sys_performance}
We implemented an extensive end-to-end~FL evaluation and~MPC simulation framework. We describe our implementation, the parameters for our accuracy evaluation, and present the results.

\myparatight{Implementation}
Our implementation is written in~Python based on~PyTorch.
It supports~multi-GPU acceleration, also for our~MPC simulation.
We used a subset of this framework for measuring the accuracy of~SepAgg~(cf.~\secref{sec:sepagg_eval}) and our approximate bit conversion~(cf.~\appref{sec:approx_evaluation}), and we will describe extensions in~\secref{sec:def_eval} to incorporate evaluations of poisoning attacks and defenses.

Our framework provides a command-line interface to run~FL training tasks and observe the resulting training as well as test accuracy.
Upon execution, the framework distributes training data among the specified number of virtual clients that locally perform training. The server(s) perform aggregation using~FedAvg.
When the~MPC simulation is enabled, the clients' input will be secret-shared before aggregation and the protocol described in~\secref{sec:SepAgg} will be executed locally.
Note that our goal is not to assess the run-time performance of the~MPC protocol but rather precisely measure the impact on accuracy.
Our implementation supports all exact and approximate secure aggregation variants described in this paper.

\myparatight{Parameters}
We evaluate the accuracy on the following standard FL tasks for image classification: training~(i)~LeNet on~MNIST~\cite{pieee:LeCunBBH98} for~1000 rounds and~(ii)~ResNet9 on~CIFAR-10~\cite{CIFARDataset} for~8000 rounds. For all tasks, we set a client batch size of~8, a learning rate of~0.05, and perform~5 local client train steps per round.
For~MNIST, we run training using~$N \in \{200, 1000, 5000\}$ clients and choose 10\% of the clients at random per round.
Due to the memory constraints of our system~(that simulates all clients at once), we restrict training for~CIFAR10 to~$N = 1000$ clients and select~$\numclients = 40$ per round.
As we observed a significant loss in accuracy for plain~SQ in our accuracy evaluation for approximate bit conversion as well as~SepAgg~(cf.~\secref{sec:approx_evaluation}), we focus our evaluation on more accurate linear quantization schemes, i.e., HSQ and~KSQ.
For the~MPC simulation of our approximate secure aggregation following~Approach-III~(cf.~Fig.~\ref{fig:SecAgg-MPC-Approach-III}), we choose a three-server dishonest majority setting.

\begin{figure}[tb!]
    \centering
    \includegraphics[width=0.95\columnwidth]{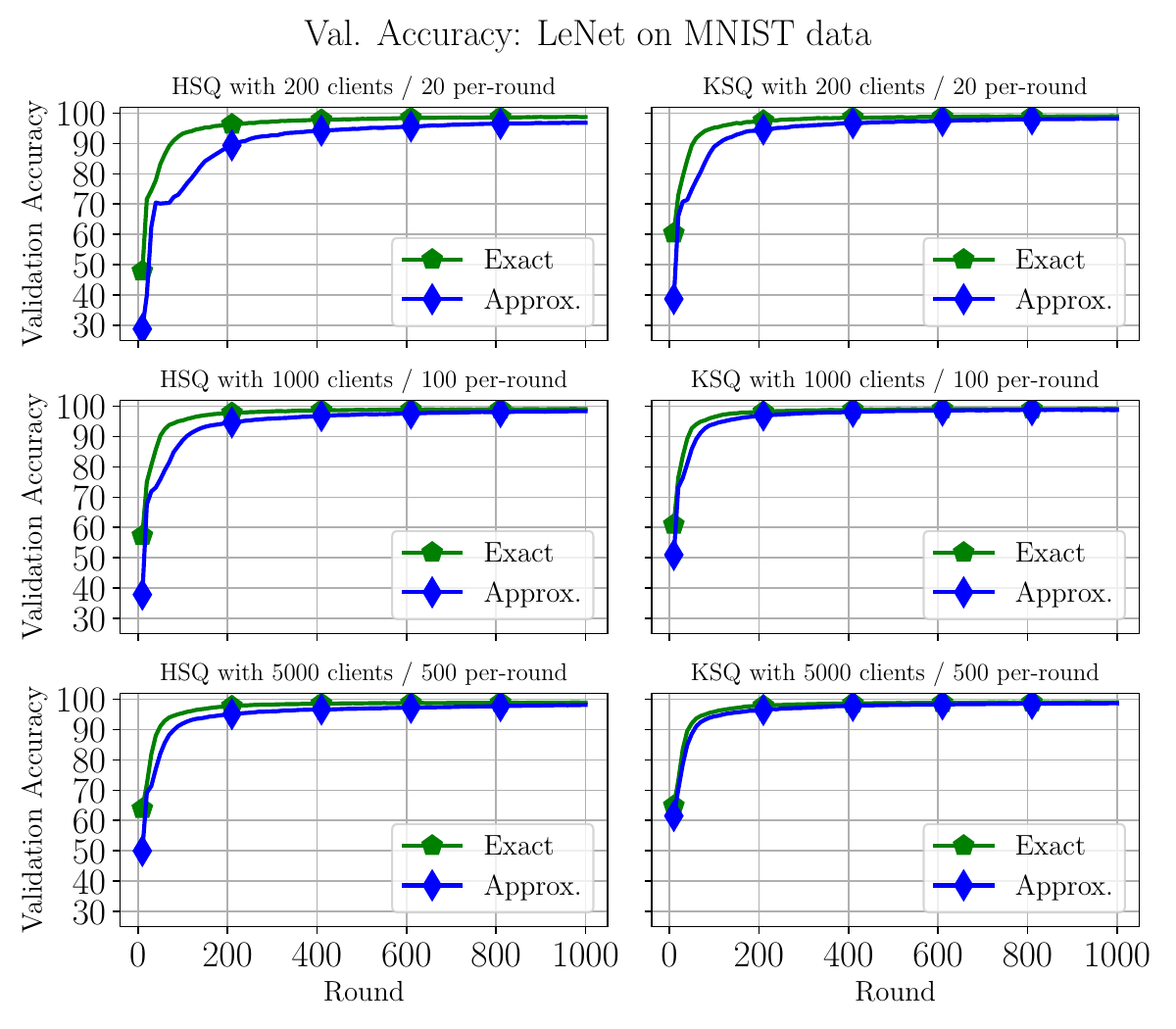}
    \vspace{-2mm}
    \captionsetup{font=small}
    \caption{Validation accuracy for training~LeNet on the~MNIST data set for~$\numclients \in \{200, 1000, 5000\}$ clients when selecting~10\% of the clients at random per round~($\numclients$) for~SQ with~Hadamard~(HSQ, left) and with~Kashin's representation~(KSQ, right); \enquote{Exact} denotes the insecure baseline, \enquote{Approx} the simulation of our~MPC-based approximate secure aggregation including~SepAgg~(cf.~\boxref{fig:SecAgg-MPC-Approach-III}).}
    \label{fig:fl_mnist_lenet}
    \vspace{-5mm}
\end{figure}

\myparatight{Results}
The results for the~MNIST/LeNet training are given in~Fig.~\ref{fig:fl_mnist_lenet}.
Validation accuracy for our approximate version converges to almost the same final accuracy as the insecure exact aggregation.
Specifically, in the final round of training, the difference between the two is diminished to~0.77\% and~0.33\% for~HSQ and~KSQ for~$N = 5000$, respectively.
Similar observations apply to~CIFAR10/ResNet9 in~Fig.~\ref{fig:fl_cifar_resnet}.
However, here the difference between the exact and approximate version for~KSQ is higher with~3.14\% in the final round.
This gap is expected due to the significantly lower number of clients per round, for which our approximate bit conversion and~SepAgg technique result in a comparatively high~NMSE over the baseline~(cf.~Figs.~\ref{fig:sepagg_sq} and~\ref{fig:approx_sq}).
We expect this effect to vanish for a real cross-device setting with thousands of participants per round ~(due to the demonstrated linear decay of the~NMSE when increasing~$\numclients$), which we unfortunately cannot simulate with complex model architectures due to hardware limitations.
Additionally, one may use a~\emph{hybrid approach}, where training uses the approximate version for initial rounds until a baseline accuracy is reached, whereas secure exact training~(potentially including only the~SepAgg~\cite{ARXIV:BAT20} approximation but not our approximate bit-to-arithmetic conversion) is used for fine tuning up to the desired target accuracy.

\begin{figure}[tb!]
    \centering
    \includegraphics[width=0.95\columnwidth]{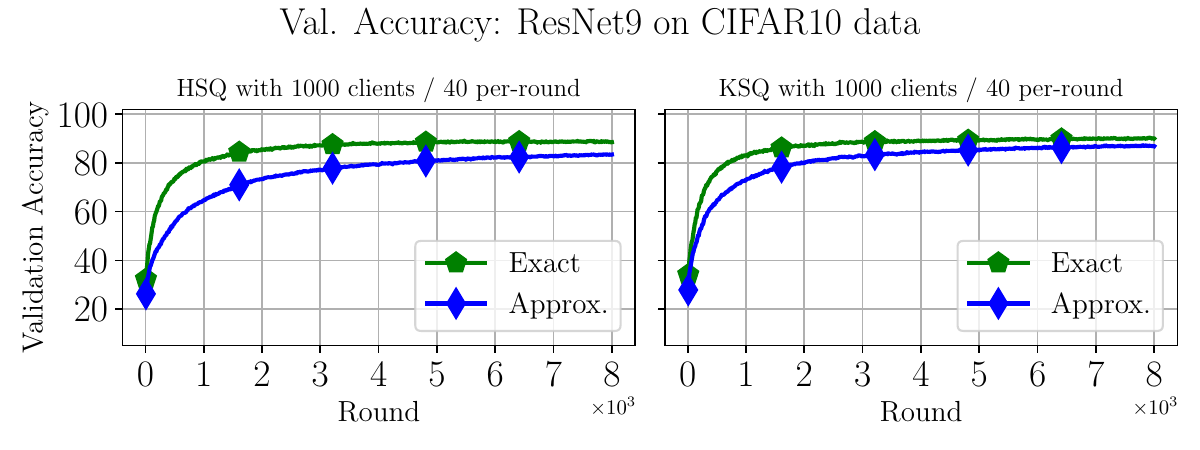}
    \vspace{-2mm}
    \captionsetup{font=small}
    \caption{Validation accuracy for training~ResNet9 on the~CIFAR10 data set for~$N = 1000$ clients with random~$\numclients = 40$ selected per round for~quantization techniques and protocols as in~Fig.~\ref{fig:fl_mnist_lenet}.}
    \label{fig:fl_cifar_resnet}
    \vspace{-5mm}
\end{figure}

In~Tab.~\ref{tab:fl_communication}, we additionally compare the exact inter-server~MPC communication cost for a naive~MPC implementation of the exact computation to our optimized approximate version including~SepAgg. As we can see, we improve the offline communication by factor~$\approx 15\times$. For the online communication, we can see a wide range of improvement factors from~$2.9\times$ to~$48\times$ for~MNIST with~$\numclients = 500$.
This highlights the positive impact when utilizing the~SepAgg approach for aggregating an increasingly large number of rows non-interactively.
Note that there are slight differences in communication overhead for~HSQ and~KSQ.
This is because for an efficient~GPU-friendly implementation of the randomized~Hadamard transform, which we use for both rotating the gradients in~HSQ and for calculating~Kashin's coefficients in~KSQ, we require that the gradients' size are a power of~2.
In~\appref{app:quantization_overhead}, we detail how we can minimize the resulting overhead by dividing the gradients into chunks, and we also give the exact number of bits per gradient that we assume in our calculations for each algorithm.

\begin{table}[htb!]
    \centering
    \resizebox{\columnwidth}{!}{%
    \begin{tabular}{lrrrrrr}
    \toprule
         &&& \multicolumn{2}{c}{Naive Exact~(cf.~Fig.~\ref{fig:SecAgg-MPC-Approach-I})} & \multicolumn{2}{c}{Our Approx.~(cf.~Fig.~\ref{fig:SecAgg-MPC-Approach-III})} \\\cline{4-5}\cline{6-7}
        Benchmark & \numclients & Method & Offline & Online & Offline & Online \\\midrule
        \multirow{2}{*}{\makecell[l]{MNIST/\\LeNet}} 
        & \multirow{2}{*}{20} 
        &   HSQ & 572.89 & 1.51 & 58.26 & 0.52 \\
        & & KSQ & 644.50 & 1.70 & 65.54 & 0.59 \\
        \midrule
        \multirow{2}{*}{\makecell[l]{MNIST/\\LeNet}} 
        & \multirow{2}{*}{100} 
        &   HSQ & 2864.46 & 5.44 & 187.49 & 0.52 \\
        & & KSQ & 3222.51 & 6.12 & 210.93 & 0.59\\
        \midrule
        \multirow{2}{*}{\makecell[l]{MNIST/\\LeNet}} 
        & \multirow{2}{*}{500} 
        &   HSQ & 14322.28 & 25.10 & 833.64 & 0.52 \\
        & & KSQ & 16112.56 & 28.24 & 937.85 & 0.59 \\
        \midrule
        \multirow{2}{*}{\makecell[l]{CIFAR10/\\ResNet9}} 
        & \multirow{2}{*}{40}  
        &   HSQ &  87079.45 & 189.27 & 6883.13 & 39.85 \\
        & & KSQ & 100828.84 & 219.15 & 7969.94 & 46.14\\
        \bottomrule
    \end{tabular}
    }
    \captionsetup{font=small}
    \caption{Inter-server communication per round for our benchmarks for different numbers of clients~$\numclients$ in~MiB.}
    \label{tab:fl_communication}
    \vspace{-4mm}
\end{table}

\section{Defending Untargeted Poisoning Attacks}
\label{sec:defense}
Our defense called~\FLdefensename{} is designed to mitigate untargeted poisoning attacks in the context of secure quantized aggregation.
These attacks pose a significant threat to the deployment of~FL for two reasons:~(i)~Untargeted attacks are particularly difficult to detect because, ignorant of the attack, service providers are unaware that they could have achieved a greater accuracy.
(ii)~Even a minor drop in accuracy can cause enormous~(competitive) damage~\cite{SP:ShejwalkarHKR22}.

Most proposed untargeted poisoning attacks on~FL use the~(unrealistic) assumption that the adversary~$\Adv$ is aware of either the aggregation rule~\cite{usenix:FangCJG20} or all benign updates~\cite{nips:BaruchBG19}.
However, the~\emph{Min-Max} attack proposed by~\cite{NDSS:SheHou21} defies this assumption and constitutes the state-of-the-art attack. 
This attack prevents the manipulations from being detected by allowing the adversary to compute representative benign updates using some clean training data; the attacker can then limit the maximum distance of the manipulated update to any other update by the maximum distances of any two benign updates.
This ensures that the malicious gradients are sufficiently similar to the set of benign gradients.
We refer to~\cite[§IV]{NDSS:SheHou21} for more specifics on the attack. 

In addition to removing assumptions about the adversary's knowledge,~\cite{NDSS:SheHou21} empirically shows that the~\emph{Min-Max} attack outperforms the former state-of-the-art poisoning attack~\cite{nips:BaruchBG19} for almost all tested datasets. 
However, since all benchmarks in~\cite{nips:BaruchBG19,NDSS:SheHou21} were performed on~FL schemes without quantization, the impact of the~\emph{Min-Max} attack on quantized~FL schemes is unclear.
Hence, we first test the attack's effectiveness in our framework using the open-sourced code\footnote{\url{https://github.com/vrt1shjwlkr/NDSS21-Model-Poisoning}} as baseline.
As we discuss in~\secref{sec:def_eval}, we observe that the attacks are effective even in the context of quantization.

\subsection{Our Defense: \FLdefensename}
\label{sec:our-defense}
From an intuitive standpoint, the adversary in an untargeted poisoning attack seeks to manipulate the global update with malicious updates to deviate it as much as possible from the result of an ideal benign training while evading potentially deployed detection mechanisms.
This baseline observation was also used by earlier works to propose defense mechanisms~\cite{NDSS:SheHou21,USENIX:NRCYMFMMMKSSZ21,SP:RatheeSWP23}, however, those cannot be combined trivially with~\FLname{} without having to de-quantize all updates and running expensive secure computation machinery.

We now outline the general design of~\FLdefensename{} and show its effectiveness against the~\emph{Min-Max} attack~\cite{NDSS:SheHou21}.
In~\appref{sec:mpcdefense}, we describe how to efficiently instantiate it in an~MPC-friendly manner to reduce communication overhead.

\begin{figure*}[htb!]
    \centering
    \includegraphics[width=0.75\textwidth]{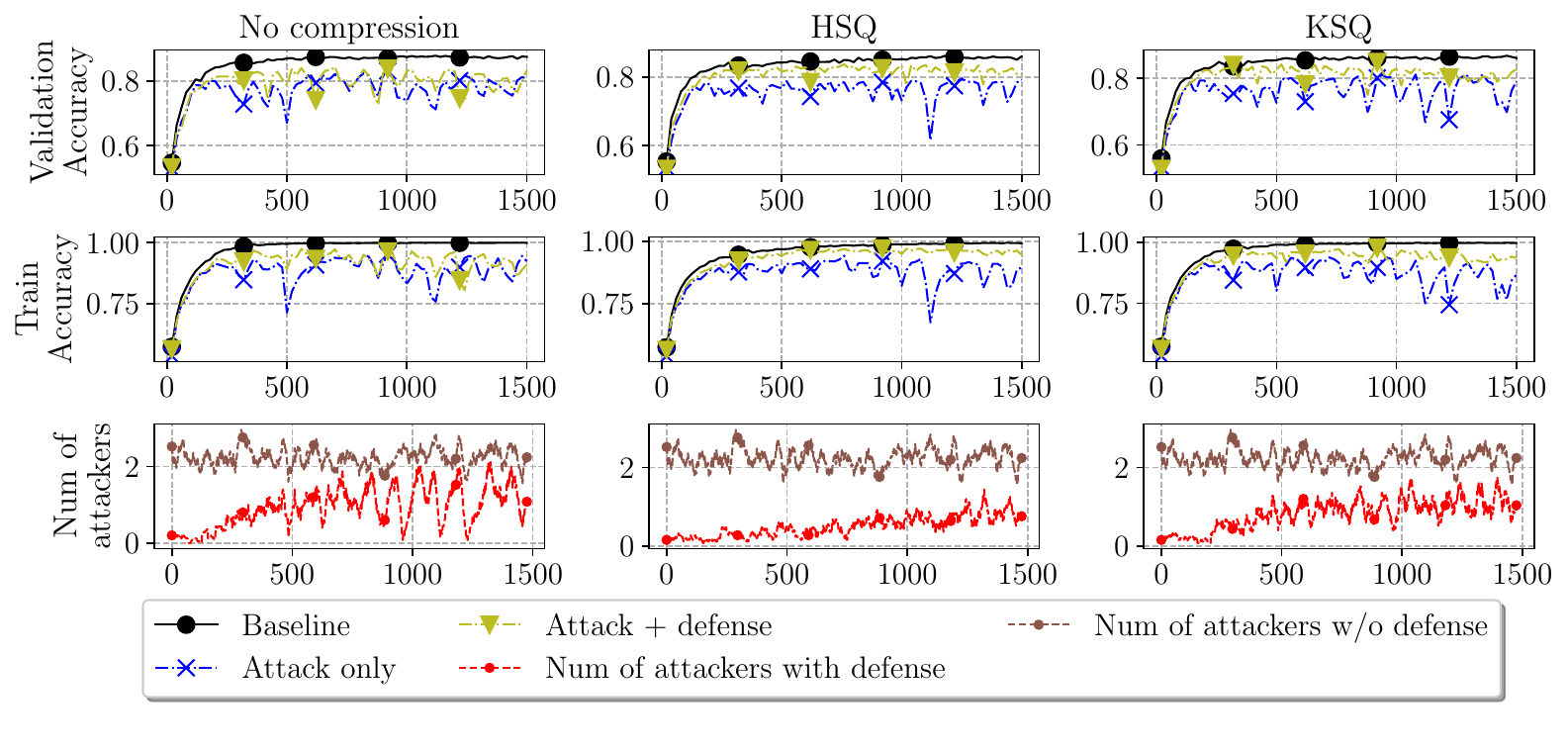}
    \vspace{-4mm}
    \captionsetup{font=small}
    \caption{Effect of~\emph{Min-Max} attack~\cite{NDSS:SheHou21} on training~ResNet9 with~CIFAR10 for~1500 aggregation rounds with and without our defense~\FLdefensename{} assuming~20\% of~$N = 50$ clients are corrupted. Note that the number of attackers included in the global update varies even without defense due to random client selection.}
    \label{fig:fl_defense_resnet}
    \vspace{-5mm}
\end{figure*}

\begin{algorithm}[htb!]
    \captionsetup{font=small}
    \caption{Our Defense: \FLdefensename{}}\label{alg:defense}
    \small 
    \begin{algorithmic}[1]
        \Procedure{\FLdefensename{}}{$\{\qvec{X_i}, \stepmin{X_i}, \stepmax{X_i}\}_{i \in [\numclients]}$}
        \Statex \Comment{Gradient Aggregation including poisoned ones.}
        \State $\aggFull \gets \Call{Aggregate}{\{\qvec{X_i}, \stepmin{X_i}, \stepmax{X_i}\}_{i \in \numclients}}$
        \Statex \Comment{$\normtwo{}$-norm Computation}
        \State $\normtwoavg \gets 0$
        \For{$k \gets 1$ to $\numclients$}
            \State $\normtwo{k} \gets \Call{L2-NormQ}{\qvec{X_k}, \stepmin{X_k}, \stepmax{X_k}}$
            \State $\normtwoavg \gets \normtwoavg + \normtwo{k}$
        \EndFor
        \Statex \Comment{$\normtwo{}$-norm based Scaling}
        \State $\normtwoavg \gets \normtwoavg/\numclients$ \Comment{Average of $\normtwo{}$-norms}
        \For{$k \gets 1$ to $\numclients$}
            \If{$\normtwo{k} > \defthreshold \cdot \normtwoavg$}
                \State $\stepmin{X_k} \gets \stepmin{X_k} \cdot (\defthreshold \cdot \normtwoavg) / \normtwo{k}$
                \State $\stepmax{X_k} \gets \stepmax{X_k} \cdot (\defthreshold \cdot \normtwoavg) / \normtwo{k}$
            \EndIf
        \EndFor
        \Statex \Comment{Cosine-distance based Filtering}
        \For{$k \gets 1$ to $\numclients$}
            \State $\theta^k \gets \Call{Cosine}{(\qvec{X_k}, \stepmin{X_k}, \stepmax{X_k}), \aggFull}$
        \EndFor
        \State $\rejectlist \gets \textsc{Top\mbox{-}K}(\vec{\theta}, \psi)$ \Comment{Returns $k$ for which $\theta^k > \psi$}
        \Statex \Comment{Aggregation of filtered updates}
        \State $\aggDef \gets \Call{Aggregate}{\{\qvec{X_i}, \stepmin{X_i}, \stepmax{X_i}\}_{i \in [\numclients], i \notin \rejectlist}}$
        \EndProcedure
    \end{algorithmic}
\end{algorithm}

\myparatight{Approach}
\label{sec:defense-approach}
\FLdefensename{} uses a hybrid approach, combining ideas from existing FL defenses based on the~$\normtwo{}$-norm~\cite{AISTATS:BagdasaryanVHES20,USENIX:NRCYMFMMMKSSZ21,sun2019can} and cosine similarity~\cite{USENIX:NRCYMFMMMKSSZ21,ndss:CaoF0G21}.
Several works like~\cite{USENIX:NRCYMFMMMKSSZ21} compute these metrics for each client pair, resulting in expensive computation.
In contrast, we aggregate all updates, including the poisoned ones, to produce the vector~$\aggFull$, which we then utilize as the reference.
At a high level,~$\normtwo{}$-norm based scaling of the gradient vectors is used at first to bound the impact of malicious contributions that are potentially overlooked~(i.e., not filtered) in later stages.
In a second step, local updates that significantly deviate from the average update direction are considered to be manipulated and, thus, excluded.
Concretely, \FLdefensename{} consists of the following steps:

\begin{enumerate}[wide, labelwidth=!, labelindent=0pt, parsep=2pt]
    \item
    \emph{$\normtwo{}$-norm-based Scaling.}
    In this step, the~$\normtwo{}$-norm of each gradient vector is compared against a public threshold multiplied with the average of the~$\normtwo{}$-norms.
    Let~$\defthreshold$ denote the threshold and~$\normtwoavg$ denote the average of the~$\normtwo{}$-norms across all clients.
    If~$\normtwo{X} > \defthreshold \cdot \normtwoavg$ for a gradient vector~$\vec{X}$, the vector is scaled\footnote{Scaling a quantized vector requires simply scaling the scales~(cf.~\secref{sec:secure-quantized-aggregation}).} by a factor of~$\defthreshold \cdot \normtwoavg/\normtwo{X}$.
    This ensures that no gradient has an~$\normtwo{}$-norm greater than~$\defthreshold \cdot \normtwoavg$.
    \item
    \emph{Cosine-distance-based Filtering.}
    This step computes the cosine distance for each gradient from the reference vector~$\aggFull$. After that, another aggregation is carried out on the updated vectors, excluding the top-$\psi$ vectors with the highest cosine distances using a secure~$\textsc{Top\mbox{-}K}$ algorithm, which involves sorting and selecting the first~$K$ items. Here, $\psi$ is either a a known bound~(i.e., defined in advance by the service provider) or an accepted percentile determined based on an assumed attacker ratio following a normal distribution.
\end{enumerate}

Alg.~\ref{alg:defense} provides the formal details of~\FLdefensename{}, including support for quantized aggregation.
Note that we use an optimizer with momentum for~FedAvg which ensures that even if the majority of clients picked at random in a training round happens to be malicious, the optimization is still based on benign contributions from the previous round. 

\smallskip
\myparatight{MPC-friendly Variant}
\label{sec:mpcdefense}
A naive secure realization of~\FLdefensename{} outlined in~Alg.~\ref{alg:defense} utilizing~MPC will yield an inefficient solution, particularly over a ring architecture.
This is due to some of the algorithm's non-MPC friendly primitives, for which we discuss viable alternatives below.

\begin{enumerate}[wide, labelwidth=!, labelindent=0pt, parsep=2pt]
    \item
    \emph{(Line 5 in Alg.~\ref{alg:defense}).}
    The computation of~$\normtwo{}$-norm within~\textsc{L2-NormQ} (cf.~Alg.~\ref{alg:defense-ltwonorm} in~\appref{app:defense-additional}) involves calculating the square root of a ring element, which corresponds to a decimal value.
    To alleviate this, we ask the clients to submit the~$\normtwo{}$-norm of their gradient vectors and the~MPC servers verify them.
    To be more specific, the provided~$\normtwo{}$-norm is squared and compared to a squared-$\normtwo{}$-norm computed by the~MPC servers via a secure comparison protocol~\cite{FC:CatSax10,fc:MakriRVW21}.
    \item
    \emph{(Lines 11 \& 12 in Alg.~\ref{alg:defense}).}
    When using~$\normtwo{}$-norm scaling, the scales of the gradient vector must be bounded if the corresponding~$\normtwo{}$-norm is greater than the limit.
    In particular, the procedure entails dividing the vector by its~$\normtwo{}$-norm.
    Because division is expensive in~MPC over rings, we ask the client to submit the reciprocal of the~$\normtwo{}$ norm as well, similar to the method suggested above.
    The provided value is validated by multiplying it by the~$\normtwo{}$-norm supplied by the client and checking whether the product is a~$1$.
    \item \emph{(Line 16 in Alg.~\ref{alg:defense}).}
    The calculation of the cosine distance between the gradient vector and the reference~$\aggFull$ requires computing the~$\normtwo{}$-norm of~$\aggFull$ and dividing by it, as shown in~\textsc{Cosine}~(cf.~Alg.~\ref{alg:defense-cosine} in~\appref{app:defense-additional}).
    However, the cosine distances are only used to filter out the top-$\psi$ vectors with the highest cosine distance, as shown in~Alg.~\ref{alg:defense}~(Line 18). As a result, we may safely disregard the division by the~$\normtwo{}$-norm of~$\aggFull$ when computing the cosine distance for our purpose.
\end{enumerate}

In addition to the aforementioned optimizations, most of the values computed as part of the~$\aggFull$ computation in the~\textsc{Aggregate} function~(Line 2 in~Alg.~\ref{alg:defense}) can be reused in the next steps, thus lowering the overhead of the defense scheme over simple aggregation.
\appref{app:defense-additional} provides details on the sub-protocols utilized in our defense algorithm given in Alg.~\ref{alg:defense}.

\subsection{Effectiveness Evaluation}
\label{sec:def_eval}
\vspace{-1mm}
%
To analyze the effectiveness of~\FLdefensename{}, we test it against the~\emph{Min-Max} attack~\cite{NDSS:SheHou21}.

\vspace{-1mm}
\myparatight{Setup}
Training involves~$N = 50$ clients of which~20\%~(as in~\cite{NDSS:SheHou21}) are corrupted.
Per training iteration, a random subset of~$\numclients = 10$ clients is chosen to train the global model. 
Each client~$\client{}$ runs its local training for~10 iterations with batches of~$B=128$ samples and a learning rate of~$\eta_{\client{}}=0.1$.
The defense threshold~$\defthreshold$ is set to~$3$ and the momentum is~$0.9$.\footnote{\cite{SP:ShejwalkarHKR22} points out that assuming more than~1\% of corrupted clients is unrealistic for most scenarios. However, in our experiments the attack failed to notably reduce the accuracy with such a low corruption level. Thus, we tested against~20\% of corrupted clients as in the original attack paper~\cite{NDSS:SheHou21}.}

\myparatight{Experimental Results}
Our results when training~ResNet9 on~CIFAR10~(i)~without an attack, (ii)~under attack without defense, and~(iii)~under attack with~\FLdefensename{} in place are given in~Fig.~\ref{fig:fl_defense_resnet}.
We compare the attack's effect when no compression is in place as well as when applying~SQ with the randomized~Hadamard transform~(HSQ) or with~Kashin's representation~(KSQ). 
We also provide similar results for training~VGG11 in~\appref{sec:defense_eval_vgg11}. 
As shown in~Fig.~\ref{fig:fl_defense_resnet}, our re-implementation of the~\emph{Min-Max} attack substantially reduces the validation accuracy by up to~20\% when no defense is in place.
This is in line with~\cite{NDSS:SheHou21}, where the authors report an accuracy degradation between~10.1\% and~42.1\% for~CIFAR10, depending on the model architecture and the aggregation scheme.
Furthermore, our experiments show that quantization does not significantly change the impact of the attack.
When~\FLdefensename{} is enabled, we can remove more than half of the malicious updates in each training iteration compared to when no defense is in place. In fact, quantization supports our defense as the additional noise added to synchronized malicious updates overturns the attacker's ability of staying just below the detection threshold. As a result, compared to unprotected training, the validation accuracy decreases by at most~7.7\% for~HSQ and~10.7\% for~KSQ. 
\section{Conclusion}
\label{sec:conclusion}
In this paper, we proposed~\FLname, the first secure aggregation framework for~FL that can efficiently process quantized updates with an optimal client-server communication as low as~1-bit per coordinate.
Together with our novel defense against state-of-the-art~(untargeted) poisoning attacks, this is the first work to unite security, efficiency, and robustness in a single and coherent framework.
As part of future work, we plan to investigate how to efficiently extend our~MPC protocols to guarantee also correctness and not~\enquote{only} privacy when dealing with malicious~MPC servers.

\section*{Acknowledgments}
This project received funding from the ERC under the EU's Horizon 2020 research and innovation program~(grant agreement No. 850990 PSOTI). It was co-funded by the DFG within SFB~1119 CROSSING/236615297 and GRK~2050 Privacy \& Trust/251805230, and by the  Private AI Collaborative Research Institute, funded by Intel, Avast, and VMware.

\bibliographystyle{IEEEtranS}
\bibliography{references}

\begin{thebibliography}{100}
\providecommand{\url}[1]{#1}
\csname url@samestyle\endcsname
\providecommand{\newblock}{\relax}
\providecommand{\bibinfo}[2]{#2}
\providecommand{\BIBentrySTDinterwordspacing}{\spaceskip=0pt\relax}
\providecommand{\BIBentryALTinterwordstretchfactor}{4}
\providecommand{\BIBentryALTinterwordspacing}{\spaceskip=\fontdimen2\font plus
\BIBentryALTinterwordstretchfactor\fontdimen3\font minus
  \fontdimen4\font\relax}
\providecommand{\BIBforeignlanguage}[2]{{%
\expandafter\ifx\csname l@#1\endcsname\relax
\typeout{** WARNING: IEEEtranS.bst: No hyphenation pattern has been}%
\typeout{** loaded for the language `#1'. Using the pattern for}%
\typeout{** the default language instead.}%
\else
\language=\csname l@#1\endcsname
\fi
#2}}
\providecommand{\BIBdecl}{\relax}
\BIBdecl

\bibitem{LetsEncrypt}
J.~Aas and T.~Geoghegan. {Introducing ISRG Prio Services for Privacy Respecting
  Metrics}.
  \url{https://www.abetterinternet.org/post/introducing-prio-services/}.

\bibitem{SCN:AddankiGJOP22}
S.~Addanki, K.~Garbe, E.~Jaffe, R.~Ostrovsky, and A.~Polychroniadou, ``{Prio+:
  Privacy Preserving Aggregate Statistics via Boolean Shares},'' in
  \emph{{SCN}}, 2022.

\bibitem{emnlp:AjiH17}
A.~F. Aji and K.~Heafield, ``Sparse communication for distributed gradient
  descent,'' in \emph{{EMNLP}}, 2017.

\bibitem{nips:AlistarhH0KKR18}
D.~Alistarh, T.~Hoefler, M.~Johansson, N.~Konstantinov, S.~Khirirat, and
  C.~Renggli, ``{The Convergence of Sparsified Gradient Methods},'' in
  \emph{{NeurIPS}}, 2018.

\bibitem{ICDCS:AndreinaMMK21}
S.~Andreina, G.~A. Marson, H.~M{\"{o}}llering, and G.~Karame, ``{BaFFLe:
  Backdoor Detection via Feedback-based Federated Learning},'' in \emph{{IEEE
  ICDCS}}, 2021.

\bibitem{WEB:ENPA}
Apple and Google, ``{Exposure Notification Privacy-preserving Analytics (ENPA)
  White Paper},''
  \url{https://covid19-static.cdn-apple.com/applications/covid19/current/static/contact-tracing/pdf/ENPA_White_Paper.pdf}.

\bibitem{CCS:AFLNO16}
T.~Araki, J.~Furukawa, Y.~Lindell, A.~Nof, and K.~Ohara, ``{High-Throughput
  Semi-Honest Secure Three-Party Computation with an Honest Majority},'' in
  \emph{{ACM CCS}}, 2016.

\bibitem{popets:AsharovHLR18}
G.~Asharov, S.~Halevi, Y.~Lindell, and T.~Rabin, ``{Privacy-Preserving Search
  of Similar Patients in Genomic Data},'' \emph{PETS}, vol. 2018, 2018.

\bibitem{CCS:AsharovL0Z13}
G.~Asharov, Y.~Lindell, T.~Schneider, and M.~Zohner, ``{More efficient
  oblivious transfer and extensions for faster secure computation},'' in
  \emph{{ACM} {CCS}}, 2013.

\bibitem{AISTATS:BagdasaryanVHES20}
E.~Bagdasaryan, A.~Veit, Y.~Hua, D.~Estrin, and V.~Shmatikov, ``{How To
  Backdoor Federated Learning},'' in \emph{{AISTATS}}, 2020.

\bibitem{nips:BaruchBG19}
G.~Baruch, M.~Baruch, and Y.~Goldberg, ``{A Little Is Enough: Circumventing
  Defenses For Distributed Learning},'' in \emph{NeurIPS}, 2019.

\bibitem{icalp:BasatMV21}
R.~B. Basat, M.~Mitzenmacher, and S.~Vargaftik, ``{How to Send a Real Number
  Using a Single Bit (And Some Shared Randomness)},'' in \emph{{ICALP}}, 2021.

\bibitem{ARXIV:Ran22}
R.~B. Basat, S.~Vargaftik, A.~Portnoy, G.~Einziger, Y.~Ben{-}Itzhak, and
  M.~Mitzenmacher, ``{QUIC-FL: Quick Unbiased Compression for Federated
  Learning},'' 2022, \url{https://arxiv.org/abs/2205.13341}.

\bibitem{acns:BaumDTZ16}
C.~Baum, I.~Damg{\aa}rd, T.~Toft, and R.~W. Zakarias, ``{Better Preprocessing
  for Secure Multiparty Computation},'' in \emph{{ACNS}}, 2016.

\bibitem{ARXIV:BAT20}
C.~Beguier, M.~Andreux, and E.~W. Tramel, ``{Efficient Sparse Secure
  Aggregation for Federated Learning},'' 2020,
  \url{https://arxiv.org/abs/2007.14861}.

\bibitem{ccs:BellBGL020}
J.~H. Bell, K.~A. Bonawitz, A.~Gasc{\'{o}}n, T.~Lepoint, and M.~Raykova,
  ``{Secure Single-Server Aggregation with (Poly)Logarithmic Overhead},'' in
  \emph{{ACM} {CCS}}, 2020.

\bibitem{ACNS:BenNieOmr19}
A.~{Ben-Efraim}, M.~Nielsen, and E.~Omri, ``{Turbospeedz: Double Your Online
  {SPDZ}! {I}mproving {SPDZ} Using Function Dependent Preprocessing},'' in
  \emph{{ACNS}}, 2019.

\bibitem{SATML:BMP0ST24}
Y.~Ben{-}Itzhak, H.~M{\"{o}}llering, B.~Pinkas, T.~Schneider, A.~Suresh,
  O.~Tkachenko, S.~Vargaftik, C.~Weinert, H.~Yalame, and A.~Yanai, ``{ScionFL:
  Efficient and Robust Secure Quantized Aggregation},'' in \emph{{IEEE}
  {SaTML}}, 2024, \url{https://doi.org/10.1109/SaTML59370.2024.00031}.

\bibitem{ICML:BernsteinWAA18}
J.~Bernstein, Y.~Wang, K.~Azizzadenesheli, and A.~Anandkumar, ``{SIGNSGD:
  Compressed Optimisation for Non-Convex Problems},'' in \emph{{ICML}}, 2018.

\bibitem{ARXIV:BHRS20}
A.~Beznosikov, S.~Horv{\'{a}}th, P.~Richt{\'{a}}rik, and M.~Safaryan, ``{On
  Biased Compression for Distributed Learning},'' 2020,
  \url{https://arxiv.org/abs/2002.12410}.

\bibitem{ICML:BhagojiCMC19}
A.~N. Bhagoji, S.~Chakraborty, P.~Mittal, and S.~B. Calo, ``{Analyzing
  Federated Learning through an Adversarial Lens},'' in \emph{{ICML}}, 2019.

\bibitem{nips:BlanchardMGS17}
P.~Blanchard, E.~M.~E. Mhamdi, R.~Guerraoui, and J.~Stainer, ``{Machine
  Learning with Adversaries: Byzantine Tolerant Gradient Descent},'' in
  \emph{NeurIPS}, 2017.

\bibitem{ARES:MP2ML}
F.~Boemer, R.~Cammarota, D.~Demmler, T.~Schneider, and H.~Yalame, ``{MP2ML: a
  mixed-protocol machine learning framework for private inference},'' in
  \emph{{ARES}}, 2020.

\bibitem{ARXIV:BDSSSP21}
F.~Boenisch, A.~Dziedzic, R.~Schuster, A.~S. Shamsabadi, I.~Shumailov, and
  N.~Papernot, ``{When the Curious Abandon Honesty: Federated Learning Is Not
  Private},'' 2021, \url{https://arxiv.org/abs/2112.02918}.

\bibitem{WEB:BDSSSP22}
------, ``{All You Need Is Matplotlib},''
  \url{http://www.cleverhans.io/2022/04/17/fl-privacy.html}, 2022.

\bibitem{ARXIV:BDSSSP23}
------, ``{Is Federated Learning a Practical PET Yet?}'' \emph{CoRR}, 2023.

\bibitem{ccs:BonawitzIKMMPRS17}
K.~A. Bonawitz, V.~Ivanov, B.~Kreuter, A.~Marcedone, H.~B. McMahan, S.~Patel,
  D.~Ramage, A.~Segal, and K.~Seth, ``{Practical Secure Aggregation for
  Privacy-Preserving Machine Learning},'' in \emph{{ACM} {CCS}}, 2017.

\bibitem{crypto:BoyleCGIKS19}
E.~Boyle, G.~Couteau, N.~Gilboa, Y.~Ishai, L.~Kohl, and P.~Scholl, ``{Efficient
  Pseudorandom Correlation Generators: Silent {OT} Extension and More},'' in
  \emph{{CRYPTO}}, 2019.

\bibitem{TOPS:BraunDST22}
L.~Braun, D.~Demmler, T.~Schneider, and O.~Tkachenko, ``{MOTION - A Framework
  for Mixed-Protocol Multi-Party Computation},'' \emph{{ACM} Trans. Priv.
  Secur.}, 2022.

\bibitem{SP:BSSSY24}
A.~Br{\"{u}}ggemann, O.~Schick, T.~Schneider, A.~Suresh, and H.~Yalame,
  ``{Don't Eject the Impostor: Fast Three-Party Computation With a Known
  Cheater},'' in \emph{{IEEE S\&P}}, 2024.

\bibitem{popets:ByaliCPS20}
M.~Byali, H.~Chaudhari, A.~Patra, and A.~Suresh, ``{FLASH: Fast and Robust
  Framework for Privacy-preserving Machine Learning},'' \emph{PETS}, 2020.

\bibitem{ARXIV:CKMT18}
S.~Caldas, J.~Kone{\v{c}}n{\'y}, H.~B. McMahan, and A.~Talwalkar, ``{Expanding
  the Reach of Federated Learning by Reducing Client Resource Requirements},''
  2018, \url{http://arxiv.org/abs/1812.07210}.

\bibitem{ndss:CaoF0G21}
X.~Cao, M.~Fang, J.~Liu, and N.~Z. Gong, ``Fltrust: Byzantine-robust federated
  learning via trust bootstrapping,'' in \emph{{NDSS}}, 2021.

\bibitem{FC:CatSax10}
O.~Catrina and A.~Saxena, ``{Secure Computation with Fixed-Point Numbers},'' in
  \emph{{FC}}, 2010.

\bibitem{ccsw:ChaudhariCPS19}
H.~Chaudhari, A.~Choudhury, A.~Patra, and A.~Suresh, ``{ASTRA: High Throughput
  3PC over Rings with Application to Secure Prediction},'' in \emph{{ACM}
  Conference on Cloud Computing Security Workshop, CCSW@CCS}, 2019.

\bibitem{ndss:ChaudhariRS20}
H.~Chaudhari, R.~Rachuri, and A.~Suresh, ``{Trident: Efficient 4PC Framework
  for Privacy Preserving Machine Learning},'' in \emph{{NDSS}}, 2020.

\bibitem{icml:ChenCKS22}
W.~Chen, C.~A. Choquette{-}Choo, P.~Kairouz, and A.~T. Suresh, ``{The
  Fundamental Price of Secure Aggregation in Differentially Private Federated
  Learning},'' in \emph{{ICML}}, 2022.

\bibitem{asiacrypt:CheonKKS17}
J.~H. Cheon, A.~Kim, M.~Kim, and Y.~S. Song, ``{Homomorphic Encryption for
  Arithmetic of Approximate Numbers},'' in \emph{{ASIACRYPT}}.\hskip 1em plus
  0.5em minus 0.4em\relax Springer, 2017.

\bibitem{iclr:ChmielBSHBS21}
B.~Chmiel, L.~Ben{-}Uri, M.~Shkolnik, E.~Hoffer, R.~Banner, and D.~Soudry,
  ``{Neural gradients are near-lognormal: improved quantized and sparse
  training},'' in \emph{{ICLR}}, 2021.

\bibitem{nsdi:Corrigan-GibbsB17}
H.~Corrigan{-}Gibbs and D.~Boneh, ``{Prio: Private, Robust, and Scalable
  Computation of Aggregate Statistics},'' in \emph{{USENIX} {NSDI}}, 2017.

\bibitem{C:CouteauRR21}
G.~Couteau, P.~Rindal, and S.~Raghuraman, ``{Silver: Silent {VOLE} and
  Oblivious Transfer from Hardness of Decoding Structured {LDPC} Codes},'' in
  \emph{{CRYPTO}}, 2021.

\bibitem{crypto:CramerDESX18}
R.~Cramer, I.~Damg{\aa}rd, D.~Escudero, P.~Scholl, and C.~Xing,
  ``{SPDZ\({\mbox{2\({}^{\mbox{k}}\)}}\): Efficient {MPC} mod
  2\({}^{\mbox{k}}\) for Dishonest Majority},'' in \emph{{CRYPTO}}, 2018.

\bibitem{SP:Damgard0FKSV19}
I.~Damg{\aa}rd, D.~Escudero, T.~K. Frederiksen, M.~Keller, P.~Scholl, and
  N.~Volgushev, ``{New Primitives for Actively-Secure {MPC} over Rings with
  Applications to Private Machine Learning},'' in \emph{{IEEE} {S\&P}}, 2019.

\bibitem{ESORICS:DKLPSS13}
I.~Damg{\r a}rd, M.~Keller, E.~Larraia, V.~Pastro, P.~Scholl, and N.~P. Smart,
  ``{Practical Covertly Secure {MPC} for Dishonest Majority - Or: Breaking the
  {SPDZ} Limits},'' in \emph{ESORICS}, 2013.

\bibitem{C:DamOrlSim18}
I.~Damg{\r a}rd, C.~Orlandi, and M.~Simkin, ``{Yet Another Compiler for Active
  Security or: Efficient {MPC} Over Arbitrary Rings},'' in \emph{{CRYPTO}},
  2018.

\bibitem{C:DPSZ12}
I.~Damg{\aa}rd, V.~Pastro, N.~Smart, and S.~Zakarias, ``{Multiparty computation
  from somewhat homomorphic encryption},'' in \emph{CRYPTO}, 2012.

\bibitem{iclr:DaviesGMAA21}
P.~Davies, V.~Gurunanthan, N.~Moshrefi, S.~Ashkboos, and D.~Alistarh, ``New
  bounds for distributed mean estimation and variance reduction,'' in
  \emph{{ICLR}}, 2021.

\bibitem{ndss:Demmler0Z15}
D.~Demmler, T.~Schneider, and M.~Zohner, ``{ABY - A Framework for Efficient
  Mixed-Protocol Secure Two-Party Computation},'' in \emph{{NDSS}}, 2015.

\bibitem{esorics:DongCLWZ21}
Y.~Dong, X.~Chen, K.~Li, D.~Wang, and S.~Zeng, ``{FLOD:} oblivious defender for
  private byzantine-robust federated learning with dishonest-majority,'' in
  \emph{{ESORICS}}, 2021.

\bibitem{ccs:ElahiDG14}
T.~Elahi, G.~Danezis, and I.~Goldberg, ``{PrivEx: Private Collection of Traffic
  Statistics for Anonymous Communication Networks},'' in \emph{{ACM} {CCS}},
  2014.

\bibitem{SPM:ErkinTLP13}
Z.~Erkin, J.~R. Troncoso{-}Pastoriza, R.~L. Lagendijk, and
  F.~P{\'{e}}rez{-}Gonz{\'{a}}lez, ``{Privacy-Preserving Data Aggregation in
  Smart Metering Systems: An Overview},'' \emph{{IEEE} Signal Process. Mag.},
  2013.

\bibitem{C:GKRS20}
D.~Escudero, S.~Ghosh, M.~Keller, R.~Rachuri, and P.~Scholl, ``{Improved
  Primitives for {MPC} over Mixed Arithmetic-Binary Circuits},'' in
  \emph{{CRYPTO}}, 2020.

\bibitem{usenix:FangCJG20}
M.~Fang, X.~Cao, J.~Jia, and N.~Z. Gong, ``{Local Model Poisoning Attacks to
  Byzantine-Robust Federated Learning},'' in \emph{{USENIX} Security}, 2020.

\bibitem{sigcomm:Fei0SCS21}
J.~Fei, C.~Ho, A.~N. Sahu, M.~Canini, and A.~Sapio, ``{Efficient sparse
  collective communication and its application to accelerate distributed deep
  learning},'' in \emph{{ACM} {SIGCOMM} Conference}, 2021.

\bibitem{SPW:FereidooniMMMMN21}
H.~Fereidooni, S.~Marchal, M.~Miettinen, A.~Mirhoseini, H.~M{\"{o}}llering,
  T.~D. Nguyen, P.~Rieger, A.~Sadeghi, T.~Schneider, H.~Yalame, and
  S.~Zeitouni, ``{SAFELearn: Secure Aggregation for private FEderated
  Learning},'' in \emph{{IEEE} {S\&P} Workshops}, 2021.

\bibitem{ICLR:FowlGCGG22}
L.~H. Fowl, J.~Geiping, W.~Czaja, M.~Goldblum, and T.~Goldstein, ``{Robbing the
  Fed: Directly Obtaining Private Data in Federated Learning with Modified
  Models},'' in \emph{{ICLR}}, 2022.

\bibitem{EC:FLNW17}
J.~Furukawa, Y.~Lindell, A.~Nof, and O.~Weinstein, ``{High-Throughput Secure
  Three-Party Computation for Malicious Adversaries and an Honest Majority},''
  in \emph{{EUROCRYPT}}, 2017.

\bibitem{ccs:GanjuWYGB18}
K.~Ganju, Q.~Wang, W.~Yang, C.~A. Gunter, and N.~Borisov, ``{Property Inference
  Attacks on Fully Connected Neural Networks using Permutation Invariant
  Representations},'' in \emph{{ACM} {CCS}}, 2018.

\bibitem{SPW:GMSWSY23}
T.~Gehlhar, F.~Marx, T.~Schneider, T.~Wehrle, A.~Suresh, and H.~Yalame,
  ``{SafeFL: MPC‐friendly framework for Private and Robust Federated
  Learning},'' in \emph{{IEEE} {S\&P} Workshops}, 2023.

\bibitem{stoc:GoldreichMW87}
O.~Goldreich, S.~Micali, and A.~Wigderson, ``{How to Play any Mental Game or
  {A} Completeness Theorem for Protocols with Honest Majority},'' in \emph{{ACM
  STOC}}, 1987.

\bibitem{AC:GorRanWan18}
S.~D. Gordon, S.~Ranellucci, and X.~Wang, ``{Secure Computation with Low
  Communication from Cross-Checking},'' in \emph{{ASIACRYPT}}, 2018.

\bibitem{C:GoyalLOPS21}
V.~Goyal, H.~Li, R.~Ostrovsky, A.~Polychroniadou, and Y.~Song, ``{ATLAS:
  Efficient and Scalable {MPC} in the Honest Majority Setting},'' in
  \emph{{CRYPTO}}, 2021.

\bibitem{popets:HegdeMSY21}
A.~Hegde, H.~M{\"{o}}llering, T.~Schneider, and H.~Yalame, ``{SoK: Efficient
  Privacy-preserving Clustering},'' \emph{PETS}, 2021.

\bibitem{popets:HesamifardTGW18}
E.~Hesamifard, H.~Takabi, M.~Ghasemi, and R.~N. Wright, ``{Privacy-preserving
  Machine Learning as a Service},'' \emph{PETS}, 2018.

\bibitem{nips:IvkinRUBSA19}
N.~Ivkin, D.~Rothchild, E.~Ullah, V.~Braverman, I.~Stoica, and R.~Arora,
  ``{Communication-efficient Distributed {SGD} with Sketching},'' in
  \emph{NeurIPS}, 2019.

\bibitem{ARXIV:jin2020stochastic}
R.~Jin, Y.~Huang, X.~He, H.~Dai, and T.~Wu, ``{Stochastic-Sign SGD for
  Federated Learning with Theoretical Guarantees},'' 2020,
  \url{https://arxiv.org/abs/2002.10940}.

\bibitem{fc:JoyeL13}
M.~Joye and B.~Libert, ``{A Scalable Scheme for Privacy-Preserving Aggregation
  of Time-Series Data},'' in \emph{{FC}}, 2013.

\bibitem{USENIX:JuvVaiCha18}
C.~Juvekar, V.~Vaikuntanathan, and A.~Chandrakasan, ``{{GAZELLE}: A low latency
  framework for secure neural network inference},'' in \emph{USENIX Security},
  2018.

\bibitem{ARXIV:FastsecAgg}
S.~Kadhe, N.~Rajaraman, O.~O. Koyluoglu, and K.~Ramchandran, ``{FastSecAgg:
  Scalable Secure Aggregation for Privacy-Preserving Federated Learning},''
  2020, \url{https://arxiv.org/abs/2009.11248}.

\bibitem{ftml:KairouzMABBBBCC21}
P.~Kairouz, H.~B. McMahan, B.~Avent, A.~Bellet, M.~Bennis, and et~al.,
  ``{Advances and Open Problems in Federated Learning},'' \emph{Found. Trends
  Mach. Learn.}, 2021.

\bibitem{icml:KarimireddyKMRS20}
S.~P. Karimireddy, S.~Kale, M.~Mohri, S.~J. Reddi, S.~U. Stich, and A.~T.
  Suresh, ``{SCAFFOLD: Stochastic Controlled Averaging for Federated
  Learning},'' in \emph{{ICML}}, 2020.

\bibitem{ccs:KellerOS16}
M.~Keller, E.~Orsini, and P.~Scholl, ``{MASCOT: Faster Malicious Arithmetic
  Secure Computation with Oblivious Transfer},'' in \emph{{ACM} {SIGSAC}},
  2016.

\bibitem{eurocrypt:KellerPR18}
M.~Keller, V.~Pastro, and D.~Rotaru, ``{Overdrive: Making {SPDZ} Great
  Again},'' in \emph{{EUROCRYPT}}, 2018.

\bibitem{CCS:KelSchSma13}
M.~Keller, P.~Scholl, and N.~P. Smart, ``{An architecture for practical
  actively secure {MPC} with dishonest majority},'' in \emph{{ACM CCS}}, 2013.

\bibitem{EPRINT:KimSKKHC19}
D.~Kim, Y.~Son, D.~Kim, A.~Kim, S.~Hong, and J.~H. Cheon, ``{Privacy-preserving
  Approximate GWAS computation based on Homomorphic Encryption},'' 2019,
  \url{https://eprint.iacr.org/2019/152}.

\bibitem{icalp:KolesnikovS08}
V.~Kolesnikov and T.~Schneider, ``{Improved Garbled Circuit: Free {XOR} Gates
  and Applications},'' in \emph{{ICALP}}, 2008.

\bibitem{ARXIV:KonecnyMYRSB16}
J.~Kone{\v{c}}n{\'y}, H.~B. McMahan, F.~X. Yu, P.~Richt{\'{a}}rik, A.~T.
  Suresh, and D.~Bacon, ``{Federated Learning: Strategies for Improving
  Communication Efficiency},'' 2016, \url{http://arxiv.org/abs/1610.05492}.

\bibitem{USENIX:KPPS21}
N.~Koti, M.~Pancholi, A.~Patra, and A.~Suresh, ``{SWIFT}: Super-fast and robust
  privacy-preserving machine learning,'' in \emph{{USENIX Security}}, 2021.

\bibitem{EPRINT:KPPS22}
N.~Koti, S.~Patil, A.~Patra, and A.~Suresh, ``{MPClan}: Protocol suite for
  privacy-conscious computations,'' \emph{J. Cryptol.}, 2023.

\bibitem{NDSS:KPRS22}
N.~Koti, A.~Patra, R.~Rachuri, and A.~Suresh, ``{Tetrad: Actively Secure 4PC
  for Secure Training and Inference},'' in \emph{NDSS}, 2022.

\bibitem{CIFARDataset}
A.~Krizhevsky, ``{Learning multiple layers of features from tiny images},''
  Tech. Rep., 2009.

\bibitem{popets:KursaweDK11}
K.~Kursawe, G.~Danezis, and M.~Kohlweiss, ``{Privacy-Friendly Aggregation for
  the Smart-Grid},'' in \emph{PETS}, 2011.

\bibitem{icml:LamW0RM21}
M.~Lam, G.~Wei, D.~Brooks, V.~J. Reddi, and M.~Mitzenmacher, ``Gradient
  disaggregation: Breaking privacy in federated learning by reconstructing the
  user participant matrix,'' in \emph{{ICML}}, 2021.

\bibitem{pieee:LeCunBBH98}
Y.~LeCun, L.~Bottou, Y.~Bengio, and P.~Haffner, ``Gradient-based learning
  applied to document recognition,'' \emph{Proc. {IEEE}}, 1998.

\bibitem{conext:LiGBL21}
K.~H. Li, P.~P.~B. de~Gusm{\~{a}}o, D.~J. Beutel, and N.~D. Lane, ``{Secure
  aggregation for federated learning in flower},'' in \emph{{ACM} International
  Workshop on Distributed Machine Learning}, 2021.

\bibitem{AAAI:LiXCGL19}
L.~Li, W.~Xu, T.~Chen, G.~B. Giannakis, and Q.~Ling, ``{RSA: Byzantine-Robust
  Stochastic Aggregation Methods for Distributed Learning from Heterogeneous
  Datasets},'' in \emph{{AAAI}}, 2019.

\bibitem{mlsys:LiSZSTS20}
T.~Li, A.~K. Sahu, M.~Zaheer, M.~Sanjabi, A.~Talwalkar, and V.~Smith,
  ``{Federated Optimization in Heterogeneous Networks},'' in \emph{{MLSys}},
  2020.

\bibitem{C:LPSY15}
Y.~Lindell, B.~Pinkas, N.~P. Smart, and A.~Yanai, ``{Efficient Constant Round
  Multi-party Computation Combining {BMR} and {SPDZ}},'' in \emph{{CRYPTO}},
  2015.

\bibitem{TIT:LyubarskiiV10}
Y.~Lyubarskii and R.~Vershynin, ``{Uncertainty principles and vector
  quantization},'' \emph{{IEEE} Trans. Inf. Theory}, 2010.

\bibitem{fc:MakriRVW21}
E.~Makri, D.~Rotaru, F.~Vercauteren, and S.~Wagh, ``{Rabbit: Efficient
  Comparison for Secure Multi-Party Computation},'' in \emph{{FC}}, 2021.

\bibitem{popets:MansouriOBC2022}
M.~Mansouri, M.~Önen, W.~Ben~Jaballah, and M.~Conti, ``{SoK: Secure
  aggregation based on cryptographic schemes for federated Learning},'' in
  \emph{{PETS}}, 2023.

\bibitem{ICML:MarchandLMTP23}
T.~Marchand, R.~Loeb, U.~Marteau{-}Ferey, J.~O. du~Terrail, and A.~Pignet,
  ``{SRATTA:} sample re-attribution attack of secure aggregation in federated
  learning,'' in \emph{{ICML}}, 2023.

\bibitem{ARXIV:MarxSSWWY23}
F.~Marx, T.~Schneider, A.~Suresh, T.~Wehrle, C.~Weinert, and H.~Yalame,
  ``{WW-FL: Secure and Private Large-Scale Federated Learning},'' 2023,
  \url{https://arxiv.org/abs/2302.09904}.

\bibitem{USENIX:MLRG20}
S.~Mazloom, P.~H. Le, S.~Ranellucci, and S.~D. Gordon, ``{Secure parallel
  computation on national scale volumes of data},'' in \emph{{USENIX
  Security}}, 2020.

\bibitem{PMLR:McMahanMRHA17}
B.~McMahan, E.~Moore, D.~Ramage, S.~Hampson, and B.~A. y~Arcas,
  ``{Communication-Efficient Learning of Deep Networks from Decentralized
  Data},'' in \emph{{AISTATS}}, 2017.

\bibitem{SP:MelisSCS19}
L.~Melis, C.~Song, E.~D. Cristofaro, and V.~Shmatikov, ``{Exploiting Unintended
  Feature Leakage in Collaborative Learning},'' in \emph{{IEEE} {S\&P}}, 2019.

\bibitem{USENIX:MLSZP20}
P.~Mishra, R.~Lehmkuhl, A.~Srinivasan, W.~Zheng, and R.~A. Popa, ``{Delphi: A
  cryptographic inference service for neural networks},'' in \emph{USENIX
  Security}, 2020.

\bibitem{CCS:MohRin18}
P.~Mohassel and P.~Rindal, ``{ABY\({}^{\mbox{3}}\): {A} Mixed Protocol
  Framework for Machine Learning},'' in \emph{{ACM} {CCS}}, 2018.

\bibitem{SP:MohZha17}
P.~Mohassel and Y.~Zhang, ``{SecureML: A System for Scalable Privacy-Preserving
  Machine Learning},'' in \emph{{IEEE S\&P}}, 2017.

\bibitem{ARXIV:MMRPVF22}
A.~Mondal, Y.~More, P.~Ramachandran, P.~Panda, H.~Virk, and D.~Gupta,
  ``{Scotch: An Efficient Secure Computation Framework for Secure
  Aggregation},'' 2022, \url{https://arxiv.org/abs/2201.07730}.

\bibitem{acsac:Munch0Y21}
J.~M{\"{u}}nch, T.~Schneider, and H.~Yalame, ``{VASA: Vector {AES} Instructions
  for Security Applications},'' in \emph{{ACM ACSAC}}, 2021.

\bibitem{SP:NasrSH19}
M.~Nasr, R.~Shokri, and A.~Houmansadr, ``{Comprehensive Privacy Analysis of
  Deep Learning: Passive and Active White-box Inference Attacks against
  Centralized and Federated Learning},'' in \emph{{IEEE} {S\&P}}, 2019.

\bibitem{USENIX:NRCYMFMMMKSSZ21}
T.~D. Nguyen, P.~Rieger, H.~Chen, H.~Yalame, H.~M{\"o}llering, H.~Fereidooni,
  S.~Marchal, M.~Miettinen, A.~Mirhoseini, F.~Koushanfar, A.-R. Sadeghi,
  T.~Schneider, and S.~Zeitouni, ``{FLAME: Taming Backdoors in Federated
  Learning},'' in \emph{USENIX Security}, 2022.

\bibitem{access:OuadrhiriA22}
A.~E. Ouadrhiri and A.~Abdelhadi, ``{Differential Privacy for Deep and
  Federated Learning: {A} Survey},'' \emph{{IEEE} Access}, 2022.

\bibitem{CCS:PFA22}
D.~Pasquini, D.~Francati, and G.~Ateniese, ``{Eluding Secure Aggregation in
  Federated Learning via Model Inconsistency},'' in \emph{{CCS}}, 2022.

\bibitem{USENIX:PSSY21}
A.~Patra, T.~Schneider, A.~Suresh, and H.~Yalame, ``{ABY2.0: Improved
  Mixed-Protocol Secure Two-Party Computation},'' in \emph{{USENIX Security}},
  2021.

\bibitem{HOST:SynCirc}
------, ``{SynCirc: Efficient Synthesis of Depth-Optimized Circuits for Secure
  Computation},'' in \emph{{IEEE} {HOST}}, 2021.

\bibitem{NDSS:PatSur20}
A.~Patra and A.~Suresh, ``{BLAZE: Blazing Fast Privacy-Preserving Machine
  Learning},'' in \emph{{NDSS}}, 2020.

\bibitem{ccs:PopaBBL11}
R.~A. Popa, A.~J. Blumberg, H.~Balakrishnan, and F.~H. Li, ``{Privacy and
  accountability for location-based aggregate statistics},'' in \emph{{ACM
  CCS}}, 2011.

\bibitem{ndss:PyrgelisTC18}
A.~Pyrgelis, C.~Troncoso, and E.~D. Cristofaro, ``{Knock Knock, Who's There?
  Membership Inference on Aggregate Location Data},'' in \emph{{NDSS}}, 2018.

\bibitem{SP:RatheeSWP23}
M.~Rathee, C.~Shen, S.~Wagh, and R.~A. Popa, ``{ELSA: Secure Aggregation for
  Federated Learning with Malicious Actors},'' in \emph{{IEEE} {S\&P}}, 2023.

\bibitem{usenix:RiaziS0LLK19}
M.~S. Riazi, M.~Samragh, H.~Chen, K.~Laine, K.~E. Lauter, and F.~Koushanfar,
  ``{XONN: XNOR-based Oblivious Deep Neural Network Inference},'' in
  \emph{{USENIX} Security}, 2019.

\bibitem{nips:RichtarikSF21}
P.~Richt{\'{a}}rik, I.~Sokolov, and I.~Fatkhullin, ``{EF21: A New, Simpler,
  Theoretically Better, and Practically Faster Error Feedback},'' in
  \emph{{NeurIPS}}, 2021.

\bibitem{indocrypt:Rotaru019}
D.~Rotaru and T.~Wood, ``{MArBled Circuits: Mixing Arithmetic and Boolean
  Circuits with Active Security},'' in \emph{{INDOCRYPT}}, 2019.

\bibitem{icml:RothchildPUISB020}
D.~Rothchild, A.~Panda, E.~Ullah, N.~Ivkin, I.~Stoica, V.~Braverman,
  J.~Gonzalez, and R.~Arora, ``{FetchSGD: Communication-Efficient Federated
  Learning with Sketching},'' in \emph{{ICML}}, 2020.

\bibitem{ARXIV:SSR20}
M.~Safaryan, E.~Shulgin, and P.~Richt{\'{a}}rik, ``{Uncertainty Principle for
  Communication Compression in Distributed and Federated Learning and the
  Search for an Optimal Compressor},'' 2020,
  \url{https://arxiv.org/abs/2002.08958}.

\bibitem{NDSS:savPTFBSH21}
S.~Sav, A.~Pyrgelis, J.~R. Troncoso-Pastoriza, D.~Froelicher, J.-P. Bossuat,
  J.~S. Sousa, and J.-P. Hubaux, ``{POSEIDON: Privacy-preserving federated
  neural network learning},'' in \emph{NDSS}, 2021.

\bibitem{interspeech:SeideFDLY14}
F.~Seide, H.~Fu, J.~Droppo, G.~Li, and D.~Yu, ``{1-bit stochastic gradient
  descent and its application to data-parallel distributed training of speech
  DNNs},'' in \emph{{INTERSPEECH}}, 2014.

\bibitem{NDSS:SheHou21}
V.~Shejwalkar and A.~Houmansadr, ``{Manipulating the Byzantine: Optimizing
  Model Poisoning Attacks and Defenses for Federated Learning},'' in
  \emph{{NDSS}}, 2021.

\bibitem{SP:ShejwalkarHKR22}
V.~Shejwalkar, A.~Houmansadr, P.~Kairouz, and D.~Ramage, ``{Back to the Drawing
  Board: {A} Critical Evaluation of Poisoning Attacks on Production Federated
  Learning},'' in \emph{{IEEE S\&P}}, 2022.

\bibitem{sp:ShokriSSS17}
R.~Shokri, M.~Stronati, C.~Song, and V.~Shmatikov, ``{Membership Inference
  Attacks Against Machine Learning Models},'' in \emph{{IEEE} {S\&P}}, 2017.

\bibitem{aaai:So2023}
J.~So, R.~E. Ali, B.~Guler, J.~Jiao, and S.~Avestimehr, ``{Securing Secure
  Aggregation: Mitigating Multi-Round Privacy Leakage in Federated Learning},''
  \emph{{AAAI}}, 2021.

\bibitem{jsait:SoGA21a}
J.~So, B.~G{\"{u}}ler, and A.~S. Avestimehr, ``{Turbo-Aggregate: Breaking the
  Quadratic Aggregation Barrier in Secure Federated Learning},'' \emph{{IEEE}
  J. Sel. Areas Inf. Theory}, 2021.

\bibitem{NIPS:SCJ18}
S.~U. Stich, J.-B. Cordonnier, and M.~Jaggi, ``{Sparsified SGD with Memory},''
  in \emph{NeurIPS}, 2018.

\bibitem{sun2019can}
Z.~Sun, P.~Kairouz, A.~T. Suresh, and H.~B. McMahan, ``Can you really backdoor
  federated learning?'' in \emph{NeurIPS FL Workshop}, 2019.

\bibitem{ARXIV:MPCLeague}
A.~Suresh, ``{MPCLeague: Robust MPC Platform for Privacy-Preserving Machine
  Learning},'' {PhD Thesis}, 2021, \url{https://arxiv.org/abs/2112.13338}.

\bibitem{ICML:SureshYKM17}
A.~T. Suresh, F.~X. Yu, S.~Kumar, and H.~B. McMahan, ``{Distributed Mean
  Estimation with Limited Communication},'' in \emph{ICML}, 2017.

\bibitem{icml:TangGARLLLZH21}
H.~Tang, S.~Gan, A.~A. Awan, S.~Rajbhandari, C.~Li, X.~Lian, J.~Liu, C.~Zhang,
  and Y.~He, ``{1-bit Adam: Communication Efficient Large-Scale Training with
  Adam's Convergence Speed},'' in \emph{{ICML}}, 2021.

\bibitem{icml:VargaftikBPMBM22}
S.~Vargaftik, R.~B. Basat, A.~Portnoy, G.~Mendelson, Y.~Ben{-}Itzhak, and
  M.~Mitzenmacher, ``{EDEN: Communication-Efficient and Robust Distributed Mean
  Estimation for Federated Learning},'' in \emph{{ICML}}, 2022.

\bibitem{nips:VargaftikBPMBM21}
S.~Vargaftik, R.~Ben{-}Basat, A.~Portnoy, G.~Mendelson, Y.~Ben{-}Itzhak, and
  M.~Mitzenmacher, ``{DRIVE: One-bit Distributed Mean Estimation},'' in
  \emph{{NeurIPS}}, 2021.

\bibitem{WEB:OBLIVIOUSDNS20}
T.~Verma and S.~Singanamalla, ``{Improving DNS Privacy with Oblivious DoH in
  1.1.1.1},'' \url{https://blog.cloudflare.com/oblivious-dns/l}, 2020.

\bibitem{ARXIV:Wangetal}
J.~Wang, Z.~Charles, Z.~Xu, G.~Joshi, H.~B. McMahan, and et~al., ``{A Field
  Guide to Federated Optimization},'' 2021,
  \url{https://arxiv.org/abs/2107.06917}.

\bibitem{ARXIV:WXWZ19}
L.~Wang, S.~Xu, X.~Wang, and Q.~Zhu, ``{Eavesdrop the Composition Proportion of
  Training Labels in Federated Learning},'' 2019,
  \url{http://arxiv.org/abs/1910.06044}.

\bibitem{INFOCOM:WangSZSWQ19}
Z.~Wang, M.~Song, Z.~Zhang, Y.~Song, Q.~Wang, and H.~Qi, ``Beyond inferring
  class representatives: User-level privacy leakage from federated learning,''
  in \emph{{INFOCOM}}, 2019.

\bibitem{NIPS:WenXYWWCL17}
W.~Wen, C.~Xu, F.~Yan, C.~Wu, Y.~Wang, Y.~Chen, and H.~Li, ``{TernGrad: Ternary
  Gradients to Reduce Communication in Distributed Deep Learning},'' in
  \emph{NeurIPS}, 2017.

\bibitem{ICML:WenGFGG22}
Y.~Wen, J.~Geiping, L.~Fowl, M.~Goldblum, and T.~Goldstein, ``{Fishing for User
  Data in Large-Batch Federated Learning via Gradient Magnification},'' in
  \emph{{ICML}}, 2022.

\bibitem{ICML:XiaoBBFER15}
H.~Xiao, B.~Biggio, G.~Brown, G.~Fumera, C.~Eckert, and F.~Roli, ``{Is Feature
  Selection Secure against Training Data Poisoning?}'' in \emph{{ICML}}, 2015.

\bibitem{ARXIV:YSHLYA21}
C.~Yang, J.~So, C.~He, S.~Li, Q.~Yu, and S.~Avestimehr, ``{LightSecAgg:
  Rethinking Secure Aggregation in Federated Learning},'' 2021,
  \url{https://arxiv.org/abs/2109.14236}.

\bibitem{focs:Yao82b}
A.~C.-C. Yao, ``{Protocols for Secure Computations (Extended Abstract)},'' in
  \emph{{FOCS}}, 1982.

\bibitem{ICML:YinCRB18}
D.~Yin, Y.~Chen, K.~Ramchandran, and P.~L. Bartlett, ``{Byzantine-Robust
  Distributed Learning: Towards Optimal Statistical Rates},'' in \emph{{ICML}},
  2018.

\bibitem{USENIX:ZhangLX00020}
C.~Zhang, S.~Li, J.~Xia, W.~Wang, F.~Yan, and Y.~Liu, ``{BatchCrypt: Efficient
  Homomorphic Encryption for Cross-Silo Federated Learning},'' in
  \emph{{USENIX} {ATC}}, 2020.

\bibitem{KDD:ZhangCJG22}
Z.~Zhang, X.~Cao, J.~Jia, and N.~Z. Gong, ``{FLDetector: Defending Federated
  Learning Against Model Poisoning Attacks via Detecting Malicious Clients},''
  in \emph{KDD}, 2022.

\bibitem{ICML:ZhangPSYMMR022}
Z.~Zhang, A.~Panda, L.~Song, Y.~Yang, M.~W. Mahoney, P.~Mittal, K.~Ramchandran,
  and J.~Gonzalez, ``{Neurotoxin: Durable Backdoors in Federated Learning},''
  in \emph{{ICML}}, 2022.

\end{thebibliography}
\clearpage
\appendices
\section{Related Work \& Background Information}
\label{sec:related-work}
This section provides additional details regarding stochastic quantization schemes discussed in~\secref{sec:stochastic_quantization_intro}.

\myparatight{Preprocessing via Random Rotations}
To deal with possible limitations of vanilla~SQ, recent state-of-the-art works suggest to~\emph{randomly rotate} the input vector prior to~SQ~\cite{ICML:SureshYKM17}.
That is, the clients and the aggregator draw rotation matrices according to some known distribution; the clients then send the quantization of the rotated vectors while the aggregator applies the inverse rotation on the estimated rotated vector.
Intuitively, the coordinates of a randomly rotated vector are identically distributed, and thus the expected difference between the coordinates is smaller, allowing for a more accurate quantization.
For~$\numclients$ clients and a gradient with~$\numcoordinates$ coordinates, this approach achieves a~NMSE\footnote{The normalized~MSE is the mean's estimate~MSE normalized by the mean clients' gradient squared norms} of~$O(\frac{\log \numcoordinates}{\numclients})$ using~$O(\numcoordinates)$ bits, which asymptotically improves over the~$O(\frac{\numcoordinates}{\numclients})$~NMSE bound of vanilla~SQ.
The computational complexity, on the other hand, is increased from~$O(\numcoordinates)$ to~$O(\numcoordinates \log \numcoordinates)$ when utilizing the randomized~Hadamard transform for rotations.

\myparatight{Preprocessing via Kashin's Representation}
The rotation approach was recently improved using~Kashin's representation~\cite{ARXIV:CKMT18,TIT:LyubarskiiV10,ARXIV:SSR20}.
Roughly speaking, it allows representing an~$\numcoordinates$-dimensional vector using a slightly larger vector with~$\lambda \cdot \numcoordinates$ smaller coefficients~($\lambda>1$).
It can be shown that applying~SQ to the~Kashin coefficients allows for an optimal~NMSE of~$O(\frac{1}{\numclients})$ using~$O(\lambda \cdot \numcoordinates)$ bits.
Compared with~\cite{ICML:SureshYKM17}, Kashin's representation yields a lower~NMSE bound by a factor of~$\log \numcoordinates$ at the cost of increasing the computational complexity by the same factor~\cite{ARXIV:Ran22,ARXIV:CKMT18}.

\subsection{Additional Compression Techniques}
\label{subsec:rw-compression}
In this work, we focus on quantization as a means to reduce bandwidth.
We nevertheless briefly overview some additional techniques considered for~FL gradient compression.

\myparatight{Sparsification}
Some works like~\cite{sigcomm:Fei0SCS21,NIPS:SCJ18,emnlp:AjiH17,ARXIV:KonecnyMYRSB16} consider sparsifying the gradients.
Quantization can also be applied to these sparsified gradients as it reduces the number of bits used per entry, while sparsification reduces the number of entries.

\myparatight{Client-side Memory-based Techniques}
Some compression techniques, including~Top-$k$~\cite{NIPS:SCJ18} and~sketching~\cite{nips:IvkinRUBSA19}, rely on client-side memory and error-feedback~\cite{interspeech:SeideFDLY14,nips:AlistarhH0KKR18,nips:RichtarikSF21,ARXIV:BHRS20} to ensure convergence.
We consider the cross-device~FL setup where clients are stateless~(e.g., a client may appear only once during a training procedure).
Therefore, client-side-memory-based techniques are mostly designed for the cross-silo~FL setup and are less applicable to cross-device~FL.

\myparatight{Entropy Encodings}
Some techniques use entropy encoding such as arithmetic encoding and~Huffman encoding~(e.g., \cite{ICML:SureshYKM17,icml:VargaftikBPMBM22,nips:AlistarhH0KKR18}).
While such techniques are appealing in their bandwidth-to-accuracy trade-offs, it is unclear how to allow for an efficient secure aggregation as gradients must be decoded before being averaged.
Also, such techniques usually incur a higher computational overhead at the clients than fixed-length representations.
An additional review of current state-of-the-art gradient compression techniques and some open challenges can be found in~\cite{ftml:KairouzMABBBBCC21,ARXIV:KonecnyMYRSB16,ARXIV:Wangetal}.

\subsection{Secure Multi-party Computation}
The field of secure multi-party computation~(MPC) started with the seminal work of~Yao~\cite{focs:Yao82b} in~1982.
It enables to securely compute arbitrary functions on private inputs without leaking anything beyond what can be inferred from the output. Since then, the field of~MPC has seen a variety of advancements of used primitives effectively improving communication and computation efficiency, e.g.,~\cite{icalp:KolesnikovS08,CCS:AsharovL0Z13,crypto:BoyleCGIKS19,C:GKRS20,acsac:Munch0Y21}.
Also, tailored efficient optimizations for varying number of computation parties have been explored, e.g.,~\cite{ndss:Demmler0Z15,USENIX:PSSY21,ccsw:ChaudhariCPS19,ndss:ChaudhariRS20,NDSS:KPRS22,HOST:SynCirc}.
Moreover, MPC research considers different assumptions regarding adversarial behavior such as the well-known semi-honest~\cite{ndss:Demmler0Z15,TOPS:BraunDST22} and malicious security model~\cite{SP:Damgard0FKSV19,ccs:KellerOS16,CCS:KelSchSma13,eurocrypt:KellerPR18,SP:BSSSY24}) as well as numbers of corrupted computation parties~(e.g., honest majority~\cite{stoc:GoldreichMW87,C:GoyalLOPS21} or dishonest majority/full threshold security~\cite{crypto:CramerDESX18,CCS:KelSchSma13,TOPS:BraunDST22,USENIX:PSSY21,ndss:Demmler0Z15}).
Beyond running the computation among several non-colluding parties, another well-established system model~(which we use in our work) is outsourcing, where the data owners secret-share their private input data among a set of non-colluding computing parties which then run the private computation on their behalf~\cite{LetsEncrypt,WEB:OBLIVIOUSDNS20,WEB:ENPA}.

\subsection{Approximate Secure Computation}
To improve efficiency of~MPC, few works already considered approximations of the exact computation.
Such approximations in~MPC include using integer or fixed-point instead of floating-point operations~(too many works to cite), approximations in genomic computation~\cite{popets:AsharovHLR18}, and in privacy-preserving machine learning such as for division~\cite{FC:CatSax10}, activation functions~\cite{SP:MohZha17,popets:HesamifardTGW18,popets:ByaliCPS20}, and completely changing the classifier to be~MPC-friendly~\cite{usenix:RiaziS0LLK19}.
Also, for~FHE, approximations are used such as in the approximate~HE scheme~CKKS~\cite{asiacrypt:CheonKKS17}, which is implemented in the~HEAAN library\footnote{\url{https://github.com/snucrypto/HEAAN}} and was used for approximate genomic computations in~\cite{EPRINT:KimSKKHC19}.
In this work, we propose for the first time to use approximations to substantially improve efficiency of~FL when combined with~MPC and give detailed evaluations on the errors introduced thereby.

\subsection{Secure Aggregation}
\label{app:related_work_secureagg}
Performing secure aggregation without revealing anything about the aggregated input values beyond what can be inferred from the output was already investigated more than~10~years ago, for example, in the context of smart metering, e.g.,~\cite{SPM:ErkinTLP13,popets:KursaweDK11}.
It has come a long way since then, resulting in practical solutions for real-world applications nowadays.

For example, Prio~\cite{nsdi:Corrigan-GibbsB17} introduces secure protocols for aggregate statistics such as sum, mean, variance, standard deviation, min/max, and frequency.
It uses additive arithmetic secret sharing, offers full-threshold security among a small set of servers running the secure computation, and validates inputs to protect against malicious clients.
Prio+~\cite{SCN:AddankiGJOP22} optimizes client computation and communication compared to~\cite{nsdi:Corrigan-GibbsB17} with a~Boolean secret sharing-based client input validation and an additional conversion from~Boolean to arithmetic sharing.
Similar to our work, it has a multi-server setup to jointly compute statistical functions on private inputs.
Compared to Prio+~\cite{SCN:AddankiGJOP22}, we optimize the naive bit-to-arithmetic conversion presented in~\cite{SCN:AddankiGJOP22} for our~$\FSecAgg$ protocol~(cf.~\secref{sec:mpc-agg}), resulting in reduced communication cost of $2.4\times$ with exact results and $4\times$ with our novel approximating variant for $\numclients= 10^5$, where $\numclients$ is the number of clients.
Popa~et~al.~\cite{ccs:PopaBBL11} specifically focus on secure location-based aggregation statistics, Joye et al.~\cite{fc:JoyeL13} on time-series data, and~PrivEx~\cite{ccs:ElahiDG14} on traffic data in anonymous communication.

So et al.~\cite{aaai:So2023} point out that differences among securely aggregated updates across multiple training iterations can also leak information about the contribution of individual clients. Most existing secure aggregation schemes are executed on one training iteration, i.e., they cannot protect against multi-round attacks. An exception is POSEIDON~\cite{NDSS:savPTFBSH21} which runs FL fully under encryption, but at the cost of significant computational overhead on clients' and server's side. Instead, So et al.~\cite{aaai:So2023} propose to organize clients in batches that can only be chosen together for a training iteration. This approach is orthogonal and fully compatible with \FLname.

\subsection{Poisoning Attacks \& Defenses}\label{subsect:flattacks}
Poisoning attacks can be categorized into untargeted and targeted attacks based on the goals of the attacker~\cite{usenix:FangCJG20}.
In the former case, the attacker aims to corrupt the global model so that it reduces or even destroys the performance of the trained model for a large number of test inputs, yielding a final global model with a high error rate~\cite{usenix:FangCJG20,NDSS:SheHou21,nips:BaruchBG19}. 
In the latter case, the attacker aims to activate attacker-defined triggers that cause a victim model to do targeted misclassifications, which can then be activated in the inference phase~\cite{ICML:BhagojiCMC19,sun2019can}.
Notably, other classification results without the trigger behave normally and main task accuracy remains high.
The second class of attacks is sometimes also referred to as~\emph{backdoor attacks}~\cite{AISTATS:BagdasaryanVHES20}.
As discussed in~\secref{sec:defense}, we consider only untargeted poisoning following the argument in~\cite{SP:ShejwalkarHKR22}: This class of attacks is particularly challenging as service providers may not notice they are under attack given they do not know which accuracy is achievable in a fresh training of a new model. Also, even small accuracy reductions can lead to serious economical losses.

Below, we detail three state-of-the-art untargeted poisoning attacks, LIE~\cite{nips:BaruchBG19}, Fang~\cite{usenix:FangCJG20}, and~Shejwalkar et al.~\cite{NDSS:SheHou21}, which are most relevant to our work.

\begin{myitemize}
\item[--] \emph{Little is Enough (LIE) attack}~\cite{nips:BaruchBG19}: In~LIE~\cite{nips:BaruchBG19}, malicious clients manipulate their local updates by adding noise drawn from the normal distribution to~\enquote{clean} updates they created following the normal training process to cause a disorientation.
LIE assumes independent and identically distributed~(iid) data and was tested against various robust aggregations such as trimmed-mean~\cite{ICML:YinCRB18}.
\item[--] \emph{Fang et al.}~\cite{usenix:FangCJG20}: The authors of~\cite{usenix:FangCJG20} formulate their untargeted poisoning attack as an optimization problem where the manipulated updates aim at maximally disorienting the global model from the benign direction.
However, they assume the adversary to either know or guess the deployed~(robust) aggregation mechanism.
Additionally, the attack was shown to be ineffective for~iid as well as severely unbalanced non-iid training datasets~\cite{NDSS:SheHou21}. 
\item[--] \emph{Shejwalkar and Houmansadr}~\cite{NDSS:SheHou21}: The attacks of~\cite{NDSS:SheHou21} follow a similar idea as~\cite{usenix:FangCJG20}: they maximize the distances between benign and malicious updates while using the evasion of outlier-based detection mechanisms as a boundary.
Concretely, they formalize the following~\enquote{Min-Max} optimization problem:
\vspace{-1mm}
\begin{align}
\label{eq:minmax}
    \argmax_{\gamma}\ \max_{i\in [n]} \lVert \nabla^m -\nabla_i \rVert_2 \leq \max_{i,j\in[n]}\lVert \nabla_i - \nabla_j\lVert_2\\
    \nabla^m = \mathsf{f}_{\mathsf{avg}}(\nabla_{\{i\in[n]\}})+\gamma\nabla^{p},
\end{align}

where~$\mathsf{f}_{\mathsf{avg}}(\nabla_{\{i\in[n]\}})$ is the average gradient and~$\gamma$ $\nabla^p$ is the adversary's perturbation vector, i.e., either the inverse unit vector of the~(simulated) benign gradients, the inverse average standard deviations, or the average gradient with flipped sign of all updates.
For details, we refer to~\S{}IV in~\cite{NDSS:SheHou21}.
\end{myitemize}

Note that although the authors of~\cite{NDSS:SheHou21} suggest several flavours of their attack based on different levels of adversarial knowledge, we compare to their~Min-Max attack as it~(i) does not make the unrealistic assumption that an adversary knows defenses in place and~(ii) it is more destructive than~LIE~\cite{nips:BaruchBG19} for almost all datasets~\cite{NDSS:SheHou21}.
We do not consider~Fang et al.~\cite{usenix:FangCJG20}'s attack as it requires the guess of the robust aggregation rule, i.e., defense mechanism, which is unrealistic in a real-world deployment.
Taking those considerations into account, we evaluate the robustness of~\FLdefensename against the state-of-the-art Min-Max attack of~\cite{NDSS:SheHou21} in~\secref{sec:def_eval}. 

\myparatight{Poisoning Defenses}
Simple parameter-wise averaging is very sensitive to outliers and, thus, can easily hamper accuracy.
Therefore, Byzantine-robust defenses aim to make~FL robust against~(untargeted) attacks.
To do so, Krum~\cite{nips:BlanchardMGS17} selects only one local update, namely the one with the closest~$n-m-2$ local updates as update for the global model, where~$n$ is the number of clients and~$m$ the number of anticipated malicious clients. Multi-krum~\cite{nips:BlanchardMGS17} extends this idea to a selection of~$c$~(instead of just one) updates. Median~\cite{ICML:YinCRB18} is an another coordinate-wise aggregation selecting the coordinate-wise median of each update parameter.
A straightforward idea to assess~(to some extent) if a specific gradient is malicious is to use an auxiliary dataset~(rootset) at the aggregator to validate the performance of the updated global model~\cite{ndss:CaoF0G21,esorics:DongCLWZ21,AAAI:LiXCGL19}. 
FLTrust~\cite{ndss:CaoF0G21} and~FLOD~\cite{esorics:DongCLWZ21} use the ReLU-clipped cosine-similarity/Hamming distance between each received update and the aggregator-computed baseline update based on the auxiliary dataset. FLDetector~\cite{KDD:ZhangCJG22} detects malicious clients by checking their model updates' consistency based on historical model updates. 
RSA~\cite{AAAI:LiXCGL19} uses an~${\sf L}_{1}$-norm-based regularization, which is also comparing to the aggregator-computed baseline update.
The recently proposed~Divider and~Conquer~(DnC) aggregation~\cite{NDSS:SheHou21} combines dimensionality reduction using random sampling with an outlier-based filtering.

The so far discussed poisoning defenses are not compatible with secure aggregation protocols in a straight-forward manner or lead to an intolerable overhead.
Only two works, namely~FLAME~\cite{USENIX:NRCYMFMMMKSSZ21} and~BaFFLe~\cite{ICDCS:AndreinaMMK21} simultaneously consider both threats.
Concretely, FLAME~\cite{USENIX:NRCYMFMMMKSSZ21} uses a density-based clustering to remove updates with significantly different cosine distances~(i.e., different directions) combined with clipping~(for more subtle manipulations). BaFFLe~\cite{ICDCS:AndreinaMMK21} introduces a feedback loop enabling a subset of clients to evaluate each global model update, while being compatible with arbitrary secure aggregation schemes. 

Recently, ELSA~\cite{SP:RatheeSWP23} considered a distributed aggregator setup and proposed methods to address poisoning attacks from malicious clients. However, ELSA's defense methods are designed to work independently on the gradients of each client, specifically using~$\ltwoplain$ and~$\linfplain$ norms. ELSA does not support defenses such as trimmed mean, median, or~Krum~\cite{nips:BlanchardMGS17}, as already mentioned in their work. Consequently, ELSA's defense mechanism is not sufficiently robust to guard against stronger attacks like Min-Max. The defenses against these types of attacks require collective information about the gradients instead of treating each gradient individually. To illustrate this point, we conducted an evaluation of~ELSA against the~Min-Max attack with~10\% corruption on a three-layer~Convolutional Neural Network using the~FashionMNIST dataset. Even after~1000 epochs of training, we observed a significant drop in accuracy to below~70\%. 

\subsection{Global Model Privacy}\label{subsect:model_privacy}
Secure aggregation addresses the concern of the aggregator observing individual model updates in the clear, potentially leading to the leakage of private information~(cf.~\secref{sec:related_work_concise}).
However, existing works~(e.g., \cite{ICML:MarchandLMTP23}) have noted that even from the aggregated global model~(computed via secure aggregation but distributed in the clear), attackers can deduce private information of individual clients, e.g., through model inversion attacks~\cite{INFOCOM:WangSZSWQ19}.
To mitigate such issues, one can apply orthogonal techniques such as differential privacy on top of~\FLname~\cite{ftml:KairouzMABBBBCC21,access:OuadrhiriA22}.
Additionally, there are works like~HyFL~\cite{ARXIV:MarxSSWWY23} proposing a framework to ensure full model privacy in~FL, but, they do not consider communication-efficient secure aggregation, as in~\FLname. 

\section{Preliminaries}
\label{app:Prelims}
This section provides relevant details regarding the primitives used in this work. We begin with providing the necessary MPC background and protocols.
The protocols are presented in a generic manner because our approach is not restricted to any specific~MPC setting. Hence, some of the sub-protocols are treated as black-boxes that can be instantiated using any efficient protocols in the underlying~MPC setting.
Since we consider dishonest majority setting to work with, we utilize the (semi-honest variant of) primitives from~\cite{C:DPSZ12,crypto:CramerDESX18,SP:Damgard0FKSV19,ACNS:BenNieOmr19} in a black-box manner.

\subsection{MPC Protocols}
\label{app:mpc-protocols}
In this section, we go over the details of the underlying~MPC protocols used in our scheme.
We consider three~MPC servers, $\serverset = \{\server{1}, \server{2}, \server{3}\}$, to which the clients delegate the aggregation computation, as shown in~Fig.~\ref{fig:sec_agg}.
All the operations are carried out in either an~$\ell$-bit ring, $\Z{\ell}$, or a binary ring, $\Z{}$.
Before we go into the protocols, we provide additional details regarding the masked evaluation scheme~\cite{C:LPSY15,ACNS:BenNieOmr19,ARXIV:MPCLeague} discussed in~\secref{sec:masked_evaluation}, starting with the sharing semantics.

\smallskip
\myparatight{Sharing Semantics}
We use two different sharing schemes:
\begin{enumerate}[wide, labelwidth=!, labelindent=0pt, parsep=0pt]
    \item
    \emph{$\sqr{\cdot}$-sharing.} A value~$\val \in \Z{\ell}$ is said to be~$\sqr{\cdot}$-shared among~MPC servers in~$\serverset$, if each server~$\server{i}$, for~$i \in [3]$, holds~$\val_i \in \Z{\ell}$ such that~$\val_1 + \val_2 + \val_3 = \val$.
    \item
    \emph{$\shr{\cdot}$-sharing.} In this sharing, every~$\val \in \Z{\ell}$ is associated with two values: a random mask~$\lv{\val}{} \in \Z{\ell}$ and a masked value~$\mv{\val} \in \Z{\ell}$, such that~$\val = \mv{\val} + \lv{\val}{}$. Here, the share of an~MPC server is defined as a tuple of the form~$(\mv{\val}, \sqr{\lv{\val}{}})$.
\end{enumerate}

\smallskip
\myparatight{Handling Decimal Values}
The~MPC protocol we use is designed over a ring architecture, while the underlying~FL algorithms handle decimal numbers.
To address this compatibility issue, we employ the well-known Fixed-Point Arithmetic~(FPA) technique~\cite{FC:CatSax10,SP:MohZha17,CCS:MohRin18}, which encodes a decimal number in~$\ell$-bits using the~2's complement representation.
The sign bit is represented by the most significant bit, while the~$\truncsize$ least significant bits are kept for the fractional component.
We use~$\ell= 32$ bit values with~$\truncsize=16$ in this work.

We will now go over the~MPC protocols used in our scheme. 
We assume that the protocols' inputs are in~$\shr{\cdot}$-shared form, and that the output is generated in~$\shr{\cdot}$-shared form among the~MPC servers.

\myparatight{Inner Product Computation}
%
For simplicity, consider the multiplication of two values~$x, y \in \Z{\ell}$ as per the~$\shr{\cdot}$-sharing semantics.
We have
\begin{align*}
   \label{eq:masked-mult}
   z = xy &= (\mv{x}+\lv{y}{})(\mv{x}+\lv{y}{}) \\ &= \mv{x}\mv{y} + \mv{x}\lv{y}{} + \mv{y}\lv{x}{} + \lv{x}{}\lv{y}{}.
\end{align*}
Since the~$\lv{}{}$ values are independent of the underlying secret, the servers can compute~$\sqr{\cdot}$-shares of the term~$\lv{x}{}\lv{y}{}$ during preprocessing using the $\piDotPre()$ protocol~\cite{crypto:CramerDESX18,ACNS:BenNieOmr19}.
This enables the servers to locally compute~$\sqr{\cdot}$-shares of $z$ during the online phase. 

In addition to the above observation, since we operate over~FPA representation, truncation~\cite{FC:CatSax10,SP:MohZha17} must be performed in order to keep the result~$z$ in~FPA format after a multiplication.
For this, we use the truncation pair method~\cite{CCS:MohRin18}, wherein a tuple of the form~$(r, r/2^{\truncsize})$ is generated in~$\shr{\cdot}$-shared form among the servers during preprocessing using the $\piTr()$ protocol~\cite{SP:Damgard0FKSV19}.
Then, with very high probability, we have
\[
  z/2^{\truncsize} = (z-r)/2^{\truncsize} + r/2^{\truncsize}.
\]
Hence, during the online phase, servers publicly open the value~$(z-r)$ and apply the above transformation to obtain the~$\shr{\cdot}$-shares of truncated~$z$, completing the protocol.

For the case of the inner-product computation~(\boxref{fig:inner-product-preprocessing}), the task can be divided into~$\vecsize$ multiplications and the result obtained accordingly.
Furthermore, because the desired result is the sum of the individual multiplication results, servers can sum them and communicate in a single shot, saving communication cost~\cite{USENIX:PSSY21}.

\begin{protocolbox}{$\piDotP(\shr{\matD{X}{\vecsize}{1}}, \shr{\matD{Y}{\vecsize}{1}}, \truncsize)$}{Inner product protocol.}{fig:inner-product-preprocessing}
	\justify
	\algoHead{Preprocessing:}
	\begin{description}
		\item[1.] Execute $\piDotPre(\sqr{\mat{\lv{X}{}}}, \sqr{\mat{\lv{Y}{}}})$ to obtain $\sqr{\gv{z}{}}$ with $\gv{z}{} = \mat{\lv{X}{}} \band \mat{\lv{Y}{}}$.
		\item[2.] Execute $\piTr()$ to generate $(\sqr{r}, \shr{r/2^{\truncsize}})$.
	\end{description} 
	\algoHead{Online:}
	\begin{description}
		\item[1.] $\server{j}$, for $j \in [\numservers]$, locally computes as follows ($\Delta = 1$ if $j = 1$, else $0$):
		\begin{itemize}
			\item $\sqr{(z - r)}_j = \Delta\cdot(\mv{\mat{X}}\band\mv{\mat{Y}}) + \mv{\mat{X}}\band\sqr{\lv{\mat{Y}}{}}_j + \mv{\mat{Y}}\band\sqr{\lv{\mat{X}}{}}_j + \sqr{\gv{z}{}}_j - \sqr{r}_j$. 
		\end{itemize}
		\item[2.] $\server{j}$, for $j \in [\numservers]$, sends $\sqr{(z - r)}_j$ to $\server{1}$, who computes $(z-r)$ and sends to all the servers.
		\item[3.] Locally compute $\shr{z} = \shr{(z-r)/2^{\truncsize}} + \shr{r/2^{\truncsize}}$.
	\end{description}
\end{protocolbox}

\myparatight{Bit-to-Arithmetic Protocol}
Given the~Boolean sharing of~$\bitb \in \Z{}$, protocol~$\piBitA$ computes the arithmetic sharing of the bit~$\bitb$ over~$\Z{\ell}$.
As shown in~Eq.~\ref{eq:case-bitA-masked}, the arithmetic equivalent~$\bitbext$ for a bit~$\bitb =  \mv{\bitb} \xor \lv{\bitb}{}$ can be obtained as 
\begin{align}
\label{eq:case-bitA-masked}
\bitbext = \mv{\bitb} \xor \lv{\bitb}{} = \Mv{\bitb} + (1 - 2\mv{\bitb}) \cdot \Lv{\bitb}{}.
\end{align}
Here, $\Mv{\bitb}$ and $\Lv{\bitb}{}$ denote the arithmetic equivalents of~$\mv{\bitb}$ and~$\lv{\bitb}{}$ respectively.
In our protocol shown in~\boxref{fig:bitA}, MPC servers invoke~$\piBitAPre$ protocol~\cite{SP:Damgard0FKSV19,USENIX:PSSY21} on the~Boolean~$\sqr{\cdot}$-shares of~$\lv{\bitb}{}$ in the preprocessing phase to obtain its respective arithmetic shares.
This enables the servers to locally compute an additive sharing of~$\bitbext$ during the online phase, as shown above.
The rest of the steps proceed similar to the inner-product protocol and we omit the details.

\begin{protocolbox}{$\piBitA(\shrB{\bitb})$}{Bit-to-arithmetic conversion protocol.}{fig:bitA}
	\justify
	\algoHead{Preprocessing:}
	\begin{description}
		\item[1.] Execute $\piBitAPre(\sqrB{\lv{\bitb}{}})$ to obtain $\sqr{\lv{\bitb}{}}$.
		\item[2.] Locally generate $(\sqr{r}, \shr{r})$ for a random $r \in \Z{\ell}$.
	\end{description} 
	\algoHead{Online:}
	\begin{description}
		\item[1.] $\server{j}$, for $j \in [\numservers]$, locally computes as follows ($\Delta = 1$ if $j = 1$, else $0$):
		\begin{itemize}
			\item $\sqr{(z - r)}_j = \Delta\cdot \mv{\bitb} + (1 - 2\mv{\bitb}) \cdot \sqr{\lv{\bitb}{}}_j - \sqr{r}_j$. 
		\end{itemize}
		\item[2.] $\server{j}$, for $j \in [\numservers]$, sends $\sqr{(z - r)}_j$ to $\server{1}$, who computes $(z-r)$ and sends to all the servers.
		\item[3.] Locally compute $\shr{z} = \shr{(z-r)} + \shr{r}$.
	\end{description}
\end{protocolbox}

\medskip
To instantiate~$\piBitAPre$, we use~SPDZ-style computations~\cite{ccs:KellerOS16,indocrypt:Rotaru019}, where~oblivious transfer~(OT) instances~\cite{CCS:AsharovL0Z13,crypto:BoyleCGIKS19,C:CouteauRR21} are used among every pair of servers.
Let~$\piOT{ij}$ denote an instance of~1-out-of-2~OT with~$\server{\sf i}$ being the sender and~$\server{\sf j}$ being the receiver.
Here, $\server{\sf i}$ inputs the sender messages~$(x_0, x_1)$ while~$\server{\sf j}$ inputs the receiver choice bit~$c \in \Z{}$ and obtains~$x_c$ as the output, for $x_0, x_1 \in \Z{\ell}$.

\begin{protocolbox}{$\piBitAPre(\sqrB{\bitb})$}{Bit-to-arithmetic preprocessing.}{fig:bitAPre}
	\justify
	\algoHead{OT Instance - I: $\sqr{\bitb}_1\sqr{\bitb}_2$}
	\begin{description}
		\item[1.] $\server{1}$ samples random $r_{12} \in \Z{\ell}$.
		\item[2.] $\server{1}$ and $\server{2}$ executes $\piOT{12}((r_{12}, r_{12} + \sqr{\bitb}_1), \sqrB{\bitb}_2)$.
		\item[3.] $\server{1}$ sets $y_{12}^1 = -r_{12}$ and $\server{2}$ sets the OT output as $y_{12}^2$.
	\end{description} 
	\justify
	\algoHead{OT Instances - II \& III: $\sqr{\bitb}_1\sqr{\bitb}_3$, $\sqr{\bitb}_2\sqr{\bitb}_3$}
	\begin{description}
		\item These are similar to the computation of $\sqr{\bitb}_1\sqr{\bitb}_2$ discussed above.
	\end{description} 
	\justify
	\algoHead{OT Instances - IV \& IV: $\sqr{\bitb}_1\sqr{\bitb}_2\sqr{\bitb}_3$}
	\begin{description}
		\item[1.] Computation can be broken down to $(\sqr{\bitb}_1\sqr{\bitb}_2)\cdot\sqr{\bitb}_3 = (y_{12}^1 + y_{12}^2)\cdot\sqr{\bitb}_3$.
		\item[2.] Execute $\piOT{13}$ for $y_{12}^1 \cdot \sqrB{\bitb}_3$. Let $z_{13}^1$ and $z_{13}^2$ denote the respective shares of $\server{1}$ and $\server{3}$. 
		\item[3.] Execute $\piOT{23}$ for $y_{12}^2 \cdot \sqrB{\bitb}_3$. Let $z_{23}^1$ and $z_{23}^2$ denote the respective shares of $\server{2}$ and $\server{3}$. 
	\end{description} 
	\justify
	\algoHead{Computation of final shares}
	\begin{description}
		\item $\server{1}$: $\sqr{\bitb}_1 = \bitb_1 - 2y_{12}^1 - 2y_{13}^1 + 4z_{13}^1$.
		\item $\server{2}$: $\sqr{\bitb}_2 = \bitb_2 - 2y_{12}^2 - 2y_{23}^1 + 4z_{23}^1$.
		\item $\server{3}$: $\sqr{\bitb}_3 = \bitb_3 - 2y_{13}^2 - 2y_{23}^2 + 4z_{13}^2 + 4z_{23}^2$.
	\end{description}
\end{protocolbox}

\medskip
To generate the arithmetic sharing of~$\lv{\bitb}{}$ from its~Boolean shares in~$\sqr{\cdot}$-shared form, a simple method would be to apply a~3-XOR using a~daBit-style approach~\cite{indocrypt:Rotaru019}, but would result in~12 executions of~1-out-of-2 OTs.
However, as pointed out in~Prio+~\cite{SCN:AddankiGJOP22}, the cost could be further optimized due to the semi-honest security model being considered in this work rather than the malicious in~\cite{indocrypt:Rotaru019}.
Since~Prio+ operates over two~MPC servers, we extend their optimized~daBit-generation protocol~(cf.~\cite[$\daBitPrio$]{SCN:AddankiGJOP22}) to our setting with three servers.

Given two bits~$\bitb_{\sf i}, \bitb_{\sf j} \in \Z{}$, the arithmetic share corresponding to their product can be generated using one instance of~$\piOT{ij}$ with~$(x_0 = r, x_1 = r + \bitb_{\sf i})$ as the~OT-sender messages and~$\bitb_{\sf j}$ as the~OT-receiver choice bit.
With this observation and using~Eq.~\ref{eq:bit-arithmetic-general}, servers can compute~$\sqr{\cdot}$-shares corresponding to the bit~$\lv{\bitb}{}$ using five~OT invocations.
The formal details appear in~\boxref{fig:bitAPre}.

\begin{protocolbox}{$\piBitASum(\shrB{\matD{M}{\vecsize}{1}})$}{Bit-to-arithmetic sum protocol.}{fig:bitA-sum}
	\justify
	\algoHead{Preprocessing:}
	\begin{description}
		\item[1.] Execute $\piBitAPre(\sqrB{\lv{\mat{M}}{}})$ to obtain $\sqr{\lv{\mat{M}}{}}$.
		\item[2.] Locally generate $(\sqr{r}, \shr{r})$ for a random $r \in \Z{\ell}$.
	\end{description} 
	\algoHead{Online:}
	\begin{description}
		\item[1.] $\server{j}$, for $j \in [\numservers]$, locally computes as follows ($\Delta = 1$ if $j = 1$, else $0$):
		\begin{itemize}
			\item $\sqr{(z - r)}_j = \Delta\cdot \rowAggregate(\mv{\mat{M}}) + (1 - 2\mv{\mat{X}})\band\sqr{\lv{\mat{Y}}{}}_j - \sqr{r}_j$. 
		\end{itemize}
		\item[2.] $\server{j}$, for $j \in [\numservers]$, sends $\sqr{(z - r)}_j$ to $\server{1}$, who computes $(z-r)$ and sends to all the servers.
		\item[3.] Locally compute $\shr{z} = \shr{(z-r)} + \shr{r}$.
	\end{description}
\end{protocolbox}

\medskip
For the case of approximate bit conversion discussed in~\secref{sec:our-approx}, the number of~OT instances can be further reduced to three following~Eq.~\ref{eq:approx-equation-middle-term}.
Concretely, the conversion involves computation of just~$\sqr{\bitb}_1\sqr{\bitb}_2\sqr{\bitb}_3$ and hence the~OT instances~{II \& III} described in~\boxref{fig:bitAPre} are no longer needed. 

When computing the sum of bits directly, the online communication can be optimized following inner-product protocol and the resulting protocol~$\piBitASum$ is given in~\boxref{fig:bitA-sum}.

\medskip
\myparatight{Bit Injection Protocol}
Given a bit~$\bitb =  \mv{\bitb} \xor \lv{\bitb}{}$ and~$\scale = \Mv{\scale} + \Lv{\scale}{}$, the bit injection operation involves computing the value~$\bitb \cdot \scale$ that can be obtained as
%
\begin{align}
\label{eq:case-bitinj-masked}
    \bitb \cdot \scale &= (\Mv{\bitb} + (1 - 2\mv{\bitb}) \cdot \Lv{\bitb}{}) \cdot (\Mv{\scale} + \Lv{\scale}{}) \nonumber \\
    &= \Mv{\bitb}\Mv{\scale} + \Mv{\bitb} \Lv{\scale}{} + (1 - 2\mv{\bitb}) \cdot (\Lv{\bitb}{} \Mv{\scale} +  \Lv{\bitb}{}\Lv{\scale}{}).
\end{align}
%

Given a boolean vector~$\matD{M}{\vecsize}{1}$ and an arithmetic vector~$\matD{N}{\vecsize}{1}$ in the secret-shared form, protocol~$\piBitInj$ computes the inner product of the two vectors, defined as~$z = \mat{M} \band \mat{N}$. This protocol is similar to the inner product protocol $\piDotP$~(\boxref{fig:inner-product-preprocessing}), with the main difference being that $\mat{M}$ is a boolean vector. 

During the preprocessing, servers first generate the arithmetic shares of $\lv{\mat{M}}{}$ from its boolean shares, similar to the bit-to-arithmetic protocol $\piBitA$ in \boxref{fig:bitA}. 
In this case, $\piBitInjPre$ is same as the $\piDotPre$ primitive discussed in~\boxref{fig:inner-product-preprocessing}.
The remaining steps are similar to the $\piDotP$ in \boxref{fig:inner-product-preprocessing} and we omit the details.

\begin{protocolbox}{$\piBitInj(\shrB{\matD{M}{\vecsize}{1}}, \shr{\matD{N}{\vecsize}{1}}, \truncsize)$}{Bit injection (sum) protocol.}{fig:bitinj-sum}
	\justify
	\algoHead{Preprocessing:}
	\begin{description}
	    \item[1.] Execute $\piBitAPre(\sqrB{\lv{\mat{M}}{}})$ to obtain $\sqr{\lv{\mat{M}}{}}$.
		\item[2.] Execute $\piBitInjPre(\sqr{\lv{\mat{M}}{}}, \sqr{\lv{\mat{N}}{}})$\footnote{$\piBitInjPre$ is same as $\piDotPre$~(\boxref{fig:inner-product-preprocessing}) in the setting considered in this work.} to obtain $\sqr{\gv{\mat{Q}}{}}$ with $\gv{\mat{Q}}{} = \lv{\mat{M}}{} \circ \lv{\mat{N}}{}$.
		\item[3.] Execute $\piTr()$ to generate $(\sqr{r}, \shr{r/2^{\truncsize}})$.
	\end{description} 
	\algoHead{Online:}
	\begin{description}
		\item[1.] $\server{j}$, for $j \in [\numservers]$, locally computes as follows ($\Delta = 1$ if $j = 1$, else $0$):
		\begin{itemize}
			\item $T_j^1 = \Delta\cdot(\mv{\mat{M}}\band\mv{\mat{N}})\ +\ \mv{\mat{M}}\band\sqr{\lv{\mat{N}}{}}_j$. 
			\item $T_j^2 = ((1 - 2\mv{\mat{M}})\circ\mv{\mat{N}}) \band\sqr{\lv{\mat{M}}{}}_j\ +\ (1 - 2\mv{\mat{M}})\band\sqr{\gv{\mat{Q}}{}}_j$. 
			\item $\sqr{(z - r)}_j = T_j^1 + T_j^2 - \sqr{r}_j$.
		\end{itemize}
		\item[2.] $\server{j}$, for $j \in [\numservers]$, sends $\sqr{(z - r)}_j$ to $\server{1}$, who computes $(z-r)$ and sends to all the servers.
		\item[3.] Locally compute $\shr{z} = \shr{(z-r)/2^{\truncsize}} + \shr{r/2^{\truncsize}}$.
	\end{description}
\end{protocolbox}

\medskip
\subsection{Binomial Sum}
\label{app:expected_values}

\begin{restatable}[Expected Values]{lemma}{binomial} 
Given $n,p \in \ZZ$, we have 
    \begin{enumerate}
        \item $\sum\limits_{p=0}^{n} p\cdot\binom{n}{p} = n \cdot 2^{n-1}.$
        \item $\sum\limits_{p=0}^{\floor{n/2}} 2p\cdot\binom{n}{2p} = \sum\limits_{p=0}^{\floor{n/2}} (2p+1)\cdot\binom{n}{2p+1} = n \cdot 2^{n-2}.$
    \end{enumerate}
    \label{lemma:binomial}
\end{restatable}
%
\begin{proof}
    Consider the binomial formula for~$(1+y)^n$, given by
    \begin{equation}
        \label{eq:binomial-x-one}
        \sum\limits_{p=0}^{n} \binom{n}{p} y^p = (1 + y)^n 
    \end{equation}
    Differentiating~Eq.~\eqref{eq:binomial-x-one} with respect to~$y$ will give 
    \begin{equation}
        \label{eq:binomial-x-one-diff}
        \sum\limits_{p=0}^{n} \binom{n}{p} p \cdot y^{p-1} = n \cdot (1 + y)^{n-1}
    \end{equation}
    Substituting~$y=1$ in~Eq.~\eqref{eq:binomial-x-one-diff} gives the first result~(1). Similarly, setting~$y=-1$ in~Eq.~\eqref{eq:binomial-x-one-diff} gives
    \begin{equation}
        \label{eq:binomial-y}
        \sum\limits_{p=0}^{n} (-1)^{p-1} p \cdot \binom{n}{p} = 0 
    \end{equation}
    Combining~Eq.~\eqref{eq:binomial-y} with the first result~(1) will give the second result~(2).
\end{proof}
%

\subsection{Overhead of HSQ and KSQ Quantization}
\label{app:quantization_overhead}
For being able to use an efficient~GPU-friendly implementation of the randomized~Hadamard transform, which we use for both rotating the gradients in~HSQ and for calculating~Kashin's coefficients in~KSQ, we require that the gradients' size to be a power of~2. 
A simple solution to meet this requirement is padding.
For example, for the~LeNet architecture with~$\approx 60k$ parameters, we can pad the gradient to~$2^{16}=65536$ entries with a small resulting overhead of~$\approx 6.2\%$~(i.e., using~$\approx 1.06$~bits per coordinate instead of~1).
However, a more sophisticated approach is to divide the gradient into decreasing power-of-two-sized chunks and inflate only the last~(smallest) chunk.\footnote{The size of the last chunk is kept above some threshold, e.g.,~$2^9$ to keep the overhead of the scales small.}
For example, for the~LeNet architecture, we can decompose it into chunks of size~$32768, 16384, 8192, 4096, 512$, that sum up to~$61952$~(with an additional overhead of two floats per chunk) with a resulting overhead of only~$\approx 1.44\%$.
Also, for~Kashin's representation, we use~$\lambda = 1.15$ for each chunk~(an extra~15\% of space) as used in previous works~(e.g., \cite{nips:VargaftikBPMBM21}).
To summarize, we state these resulting overheads in~Tab.~\ref{tab:quantization_overhead}.

\begin{table}[htb!]
    \centering
    \resizebox{0.9\columnwidth}{!}{%
    \begin{tabular}{lrrrr}
    \toprule
     Architecture                 & $n$      & SQ       & HSQ & KSQ      \\ \midrule
    LeNet    & 61706    & 61706    & 62272                    & 73024    \\ 
    ResNet9  & 4903242  & 4903242  & 4915456                  & 5767424  \\ 
    ResNet18 & 11220132 & 11220132 & 11272192                 & 12583040 \\ \bottomrule
    \end{tabular}%
    }
    \captionsetup{font=small}
    \caption{Exact number of bits used for different network architectures and quantization schemes compared to the baseline number of coordinates~$n$.}
    \label{tab:quantization_overhead}
    \vspace{-2mm}
\end{table}

\section{\FLname{}: Additional Details}
\label{app:ourfl-additional}
This section provides addition details of our~FL framework~\FLname{} presented in~\secref{sec:system}.
We begin with providing additional details regarding the approximate bit conversion discussed in~\secref{sec:approximate_mpc}.

\subsection{Multi-bit Quantization Schemes}
\label{sec:extension-multi-bit}
This section describes how our scheme \FLname{} can be extended to support multi-bit linear quantization schemes, in which each coordinate is classified into more than two levels, resulting in each coordinate being represented by more than a single bit.

For instance, consider the quantization in TernGrad~\cite{NIPS:WenXYWWCL17}, where each coordinate is compressed to one of the three levels~$\{-1,0,1\}$. Here, each coordinate can be represented using two bits, say~$\bitb_1$ and~$\bitb_2$ and the quantized level can be computed as~$2\bitb_1-\bitb_2$.  

To use our scheme, each client~$\client{i}$ share the bits separately using the underlying boolean secret sharing scheme, i.e., $\shrB{\bitb_1}_i$ and~$\shrB{\bitb_2}_i$. MPC servers use our instantiations of~$\FSecAgg$ functionality discussed in~\secref{sec:mpc-agg} to aggregate each of the bits and obtain the result in arithmetic sharing format, i.e, $\shr{\bitb_1}$ and $\shr{\bitb_2}$. The final result can be locally computed by the~MPC servers as $2\shr{\bitb_1} + \shr{\bitb_2}$, since the underlying~MPC protocol used in~\FLname{} is linear.

\subsection{Approximate Bit Conversion}
\label{app:approximation-proof}

\binomialexpectation*

\begin{restatable}[]{proof}{binomialexpproof} 
For the analysis, we use the truth table of~$\bitb$, denoted by~$\truthtable{\bitb}$, which has~$2^{\numshares}$ rows.
Half of the rows in~$\truthtable{\bitb}$ correspond to~$\bitb = 0$, while the other half correspond to~$\bitb = 1$.
The truth table for three shares~($\numshares = 3$) is given in~Tab.~\ref{tab:truth-table-three-shares} as a reference.

\begin{table}[htb!]
    \centering
    \resizebox{0.9\columnwidth}{!}{
    \begin{tabular}{c|ccc|ccc|c}
      \toprule
        $\bitb$ & $\bitb_1$ & $\bitb_2$ & $\bitb_3$ & $\termSum$& $\termMid$ & $\termProd$ & $\bitbAr$\\
      \hline
        \rowcolor{lightgray}
        0 & 0 & 0 & 0 & 0 & 0  & 0 & 0 \\
        1 & 0 & 0 & 1 & 1 & 0  & 0 & 1 \\
        1 & 0 & 1 & 0 & 1 & 0  & 0 & 1 \\
        \rowcolor{lightgray}
        0 & 0 & 1 & 1 & 2 & -2 & 0 & 0 \\
        1 & 1 & 0 & 0 & 1 & 0  & 0 & 1 \\
        \rowcolor{lightgray}
        0 & 1 & 0 & 1 & 2 & -2 & 0 & 0 \\
        \rowcolor{lightgray}
        0 & 1 & 1 & 0 & 2 & -2 & 0 & 0 \\
        1 & 1 & 1 & 1 & 3 & -6 & 4 & 1 \\
       \bottomrule
    \end{tabular}
    }
    \captionsetup{font=small}
    \caption{Truth table for $\bitb = \bitb_1 \xor \bitb_2 \xor \bitb_3$. The rows corresponding to $\bitb = 0$ are \colorbox{lightgray}{highlighted}. $\bitbAr$ denotes the arithmetic equivalent of $\bitb$.}
    \label{tab:truth-table-three-shares}
\end{table}

\noindent \emph{Sum Term $(\termSum)$:}
For each row of the form~$(\bitb_1,\ldots,\bitb_{\numshares})$ in~$\truthtable{\bitb}$, $\termSum$ equals~$\bitbAr_1+\ldots+\bitbAr_{\numshares}$, which can be interpreted as the number of~$\bitbAr_i$'s selected out of the~$\numshares$ possible. Furthermore, there are a total of~$\binom{\numshares}{k}$ rows with sums equal to~$k$, with~$k$ being odd corresponding to the row for~$\bitb = 1$ and~$k$ being even corresponding to the row for~$\bitb = 0$.
As a result, given~$\bitb = 0$, the expectation of the sum term can be calculated as the product of~$1/2^{\numshares-1}$~(corresponding to rows in~$\truthtable{\bitb}$ with~$\bitb = 0$) and the sum of terms of the form~$k \cdot \binom{\numshares}{k}$ with~$k$ being even.
Using~Lem.~\ref{lemma:binomial} in~\appref{app:Prelims}, we get
\begin{small}
\begin{align*}
    \E[\termSum \mid(\bitb = 0)] &= \frac{1}{2^{\numshares-1}} \cdot \sum\limits_{k=0}^{\floor{\numshares/2}} 2k  \binom{\numshares}{2k} = \frac{1}{2^{\numshares-1}} \cdot \numshares \cdot 2^{\numshares-2} = \numshares/2.
\end{align*}
\end{small}
Similarly, we obtain~$\E[\termSum \mid(\bitb = 1)] = \numshares/2$. To summarize, we have~$\E[\termSum \mid \bitb] = \numshares / 2$.

\medskip
\noindent \emph{Product Term $(\termProd)$:}
The product of all the~$\numshares$ shares will be~$1$ only if all the shares are~$1$, otherwise it will be~$0$.
Moreover, all shares of~$\bitb$ being~$1$ correspond to~$\bitb = 1$ if~$\numshares$ is odd, and~$\bitb = 0$ otherwise.
Now, when~$\numshares$ is odd then~$\termProd=(-2)^{\numshares-1}$ with probability $\frac{1}{2^{\numshares-1}}$~(when all the shares of $\bitb$ are $1$, given that at least one share is 1 as we are in the case $b=1$), and~$0$ otherwise. 
In this case, we can write
\begin{small}
\begin{align*}
    \E[\termProd \mid(\bitb = 0 \wedge \text{$\numshares$ is odd})] &= \frac{1}{2^{\numshares-1}} \cdot 0~~~~~~~~~~~~ = 0 \\ 
    \E[\termProd \mid(\bitb = 1 \wedge \text{$\numshares$ is odd})] &= \frac{1}{2^{\numshares-1}} \cdot (-2)^{\numshares-1} = 1
\end{align*}
\end{small}
Similarly, the case for even~$\numshares$ can be written as
\begin{small}
\begin{align*}
    \E[\termProd \mid(\bitb = 0 \wedge \text{$\numshares$ is even})] &= \frac{1}{2^{\numshares-1}} \cdot (-2)^{\numshares-1} = -1 \\ 
    \E[\termProd \mid(\bitb = 1 \wedge \text{$\numshares$ is even})] &= \frac{1}{2^{\numshares-1}} \cdot 0~~~~~~~~~~~~= 0
\end{align*}
\end{small}
The above observation can be summarized as~$\E[\termProd \mid \bitb] = \bitb - (\numshares\mbox{-}1)\Mod{2}$.

\medskip
\noindent \emph{Middle Term $(\termMid)$:}
Given~$\E[\bitb] = \bitb$, $\E[\termSum \mid \bitb]$ and~$\E[\termProd \mid \bitb]$, the expectation of $\termMid$ can be calculated as
\begin{small}
\begin{align*}
    \E[\termMid \mid \bitb]  &= \E[\bitb] - \E[\termSum \mid \bitb] - \E[\termProd \mid \bitb]\\ 
                        &= \bitb - \numshares/2 - (\bitb - (\numshares\mbox{-}1)\Mod{2}) = (\numshares\mbox{-}1)\Mod{2} - \numshares/2.
\end{align*}
\end{small}
\noindent
This concludes the proof of Lem.~\ref{lemma:bitA-exp-analysis}.
\end{restatable}

\myparatight{Efficiency Analysis} 
\label{sec:approx_efficiency}
We measure the efficiency gains achieved by our approximation method, discussed in~\secref{sec:our-approx-approach}, by counting the number of~\emph{cross terms}\footnote{Terms for which interaction among~MPC servers is necessary.} that must be computed securely using~MPC. Cross terms are terms that compute the product of two or more shares. While the exact amount of computation and communication varies depending on the~MPC protocol and setting~(e.g., honest vs.\ dishonest majority or semi-honest vs.\ malicious security), we believe cross terms can provide a protocol-independent and realistic assessment of scalability.\footnote{We acknowledge that the analysis cannot provide an exact comparison, owing to the presence of the product term in the approximation. e.g., depending on the underlying~MPC setup, the product term~($\termProd$) may require more communication than the middle terms~($\termMid$), and therefore the effect of approximation may be minimized.}

\begin{table}[htb!]
    \centering
    \begin{tabular}{r r r}
         \toprule
         \multirow{2}{*}{Computation} 
         & \multicolumn{2}{c}{\#cross-terms} \\ \cmidrule{2-3}
         & Exact ($\bitbAr$) & Approximate ($\bitbapprox$) \\
         \midrule
         Bit-to-Arithmetic & $2^{\numshares} - \numshares - 1$                 & $1$\\
         Bit Injection   & $2^{\numshares} + \numshares^2 - 2\numshares - 1$ & $\numshares^2 - \numshares +1$\\
         \bottomrule
    \end{tabular}
    \captionsetup{font=small}
    \caption{Efficiency analysis via approximate bit conversion with respect to the \#cross-terms involved.}
    \label{tab:efficiency-via-approximation}
    \vspace{-3mm}
\end{table}

Tab.~\ref{tab:efficiency-via-approximation} provides details regarding the number of cross terms involved in obtaining the arithmetic equivalent of~$\bitb = \xor_{i=1}^{\numshares} \bitb_i$.
The gains increase significantly with a higher number of shares~$\numshares$ due to the exponential growth in the number of cross terms for the exact computation.
Tab.~\ref{tab:efficiency-via-approximation} also provides details for a bit injection operation in which the product of a~Boolean bit~$\bitb$ and a scale value~$\scale$ is securely computed.
Given~$\scale = \sum_{i=1}^{\numshares} \scale_i$, the value~$\bitb \cdot \scale$ can be computed by first computing either~$\bitbAr$ or~$\bitbapprox$~(depending on whether an exact or approximate value is required) and then multiplying by~$\scale$.

\subsection{\FLdefensename{}: Additional Details}
\label{app:defense-additional}
In this subsection, we provide additional details of our~\FLdefensename{}.

\begin{figure*}[htb!]
    \centering
    \includegraphics[width=0.95\textwidth]{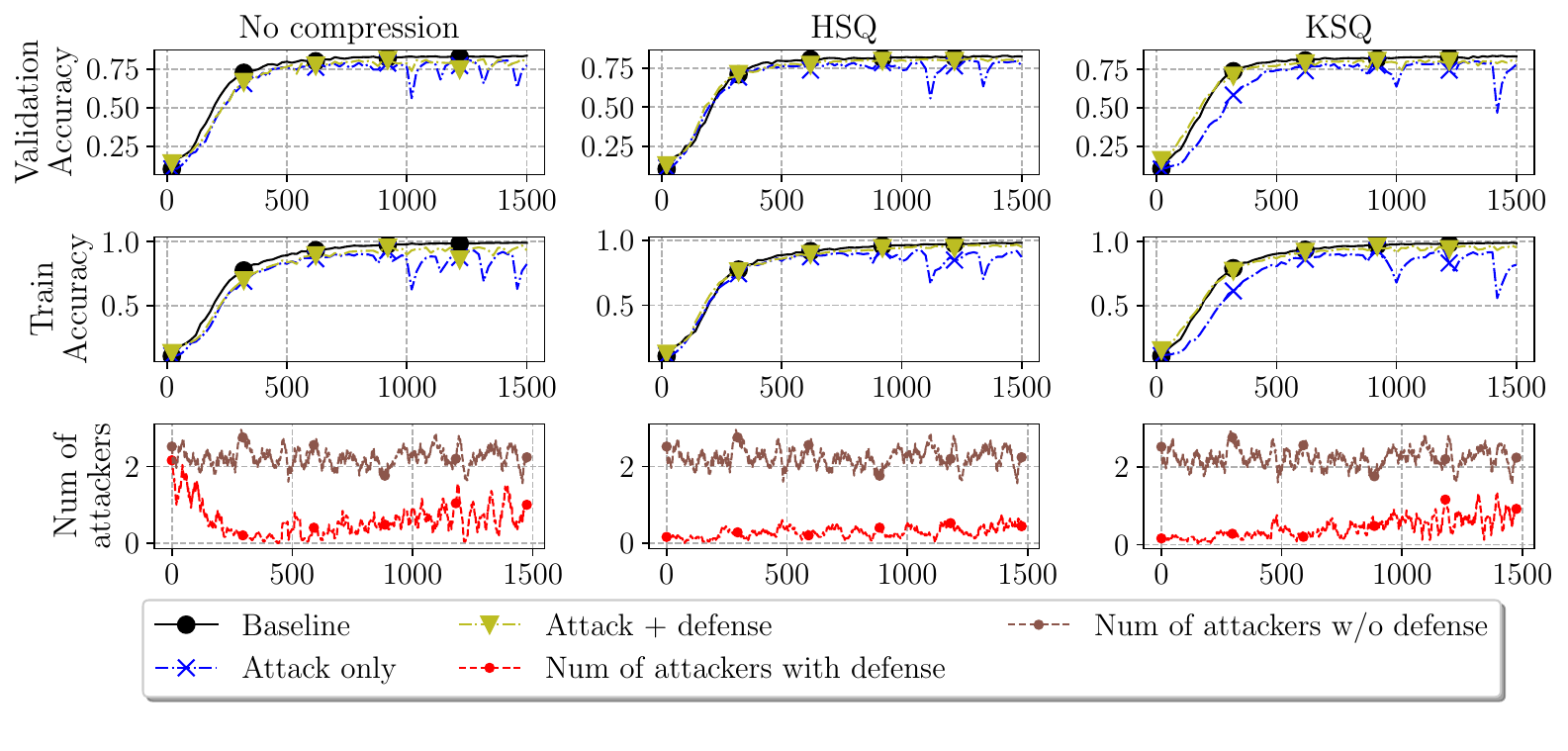}
    \vspace{-4mm}
    \captionsetup{font=small}
    \caption{Effect of~\emph{Min-Max} attack~\cite{NDSS:SheHou21} on training~VGG11 with~CIFAR10 for~1500 aggregation rounds with and without our defense~\FLdefensename{} assuming~20\% of~$N = 50$ clients are corrupted. Note that the number of attackers included in the global update varies even without defense due to random client selection.}
    \label{fig:fl_defense_vgg}
\end{figure*}

\myparatight{Sub-protocols}
\label{app:mpc-friendly-defense}
Here, we provide the details of the sub-protocols used in \FLdefensename{}~(cf.~Alg.~\ref{alg:defense} in~\secref{sec:defense-approach}).

\begin{algorithm}[htb!]
    \captionsetup{font=small}
    \caption{Quantized Aggregation}\label{alg:defense-aggregation}
    \small
    \begin{algorithmic}[1]
        \Procedure{Aggregate}{$\{\qvec{Y_i}, \stepmin{Y_i}, \stepmax{Y_i}\}_{i \in \alpha}$}
            \State $\vec{Z} \gets \vec{0}$
            \For{$k \gets 1$ to $\alpha$}
                \State $\vec{Z} \gets \vec{Z} + \left( \stepmin{Y_k} \oplus \qvec{Y_k} \circ (\stepmax{Y_k} - \stepmin{Y_k}) \right)$ 
            \EndFor
            \State $\vec{Z} \gets \vec{Z} / \alpha$
            \State \Return $\vec{Z}$
        \EndProcedure
    \end{algorithmic}
\end{algorithm}

Alg.~\ref{alg:defense-aggregation} computes the aggregation of~$\alpha$ quantized vectors. As shown in Eq.~\ref{eq:quantized value}, the dequantized value of a vector $\vec{Y}$, given its quantized form $(\qvec{Y}, \stepmin{Y}, \stepmax{Y})$, can be computed as
\[
  \vec{Y} = \stepmin{Y}\ \oplus\ \qvec{Y} \circ (\stepmax{Y} - \stepmin{Y}).
\]
The above operation essentially places $\stepmin{Y}$ in those positions of the vector $\vec{Y}$ with the corresponding bit in $\qvec{Y}$ being zero, and the rest with $\stepmax{Y}$.

\begin{algorithm}[htb!]
    \captionsetup{font=small}
    \caption{L2-Norm Computation (Quantized)}\label{alg:defense-ltwonorm}
    \small
    \begin{algorithmic}[1]
        \Procedure{L2-NormQ}{$\qvec{Y}, \stepmin{Y}, \stepmax{Y}$}
            \State $\beta \gets \textsc{Len}(\qvec{Y})$ \Comment{Dimension of $\qvec{Y}$}
            \State $N_O \gets \textsc{Sum}(\qvec{Y})$ \Comment{Number of ones in $\qvec{Y}$}
            \State $N_Z \gets \beta - N_O$ \Comment{Number of zeros in $\qvec{Y}$}
            \State \Return $\sqrt{N_Z\cdot(\stepmin{Y})^2 + N_O\cdot(\stepmax{Y})^2}$
        \EndProcedure
    \end{algorithmic}
\end{algorithm}

Alg.~\ref{alg:defense-ltwonorm} computes the~$\normtwo{}$-norm of a quantized vector. As discussed in \secref{sec:secure-quantized-aggregation}, a quantized vector $\vecqV{Y}{}$ consists of a binary vector~$\qvec{Y}$ and the respective min.\ and max.\ scales~$\stepmin{Y}/\stepmax{Y}$. In this case, we observe that the squared~$\normtwo{}$-norm can be obtained by first counting the number of zeroes and ones in the vector, denoted by $N_Z$ and and $N_O$ respectively, followed by multiplying them with the square of the respective scales and adding the results, i.e. $N_Z\cdot(\stepmin{Y})^2 + N_O\cdot(\stepmax{Y})^2$. Furthermore, computing the number of ones $N_O$ corresponds to the bit-aggregation of the vector $\vec{Y}$, for which our aggregation methods discussed in \secref{sec:mpc-agg} can be utilized. 

\begin{algorithm}[htb!]
    \captionsetup{font=small}
    \caption{Cosine Distance Calculation}\label{alg:defense-cosine}
    \small
    \begin{algorithmic}[1]
        \Procedure{Cosine}{$(\qvec{Y}, \stepmin{Y}, \stepmax{Y})$, $\vec{S}$}
            \State $\normtwo{Y} \gets \Call{L2-NormQ}{\qvec{Y}, \stepmin{Y}, \stepmax{Y}}$
            \State $\normtwo{S} \gets \lVert \vec{S} \rVert$ \Comment{Computes $\normtwo{}$-norm}
            \State $\alpha \gets \textsc{Sum}(\vec{S})$ \Comment{Sum of elements of $\vec{S}$}
            \State $\beta \gets \textsc{Inner\mbox{-}Product}(\qvec{Y}, \vec{S})$
            \State $\gamma = \stepmin{Y} \cdot \alpha + \beta \cdot (\stepmax{Y} - \stepmin{Y})$
            \State \Return $\gamma/(\normtwo{Y} \cdot \normtwo{S})$
        \EndProcedure
    \end{algorithmic}
\end{algorithm}

Alg.~\ref{alg:defense-cosine} is used to compute the cosine distance between a quantized vector $\vecqV{Y}{}$ and a reference vector $\vec{S}$. The cosine distance is given by $\frac{\vecqV{Y}{} \band \vec{S}}{\lVert \vecqV{Y}{} \rVert \cdot \lVert \vec{S} \rVert}$, where $\lVert \cdot \rVert$ corresponds to the $\normtwo{}$-norm of the input vector. Using Eq.~\ref{eq:quantized value}, we can write
\begin{align*}
    \vecqV{Y}{} \band \vec{S} &= (\stepmin{Y} \oplus \qvec{Y} \circ (\stepmax{Y}  - \stepmin{Y})) \band \vec{S} \\
    &= \stepmin{Y} \band \vec{S}\ +\  (\qvec{Y}\band\vec{S}) \cdot (\stepmax{Y}  - \stepmin{Y}).
\end{align*}
Thus, the inner product computation of~$\vecqV{Y}{} \band \vec{S}$ reduces to computing~$\qvec{Y} \band \vec{S}$, followed by two multiplications.

\medskip
\myparatight{Evaluation on~VGG11}
\label{sec:defense_eval_vgg11}
In addition to our results in~\secref{sec:def_eval}, we evaluate the~\emph{Min-Max} attack on~VGG11 trained with~CIFAR10. The experimental setup is identical to~\secref{sec:def_eval}. 
The results are shown in~Fig.~\ref{fig:fl_defense_vgg}.

Similarly as for~ResNet9~(cf.~Fig.~\ref{fig:fl_defense_resnet}), the~\emph{Min-Max} attack substantially reduces the validation accuracy when training~VGG11: We observe drops of up to~36.8\%.
However, on average, VGG11 is less impacted by the attack. Concretely, only~15\% of the iterations observe a validation accuracy reduction of about~10\% or more when using no compression.
One third of the training rounds are impacted by about~10\% or more when using~Kashin's representation~(KSQ) while with the~Hadamard transform~(HSQ) only very few training rounds showed a significant accuracy reduction.
Thus, HSQ seems to be more robust against untargeted poisoning.

With~\FLdefensename{}, the accuracy reduction is still smaller for all variants.
With~HSQ, on average~0.28~malicious updates are included in global updated instead of~2.24 without defense.
With respect to the validation accuracy, the difference between having no attack and employing~\FLdefensename{} when under attack is less than~4\% in almost all training iterations.
When using~KSQ, a global update includes just~0.44 malicious updates on average, and the attack impact is at least halved in two third of the training iterations.


\end{document}